\documentclass[12pt]{article}
\usepackage{amsmath}
\usepackage{amsthm}
\usepackage{graphicx,psfrag,epsf}
\usepackage{enumerate}
\usepackage{natbib}
\usepackage[utf8]{inputenc} 
\usepackage{hyperref}       
\usepackage{url}            
\usepackage{booktabs}       
\usepackage{amsfonts}       
\usepackage{nicefrac}       
\usepackage{microtype}      
\usepackage{authblk}
\usepackage{natbib}
\usepackage{doi}
\usepackage{float}
\usepackage{algorithmic}
\usepackage{xcolor}
\usepackage[toc,page]{appendix}

\newtheorem{theorem}{Theorem}[section]
 
\newtheorem{lemma}[theorem]{Lemma}

\newtheorem{definition}[theorem]{Definition}
\newtheorem{example}[theorem]{Example}
\newtheorem{assumption}[theorem]{Assumption}
\newtheorem{proposition}[theorem]{Proposition}

\begin{document}


  \title{Distributed estimation through parallel approximants}
  \author{Aritra Chakravorty, William S. Cleveland, and Patrick J. Wolfe}
  \affil{\texttt{\{chakrav0,wsc,patrick\}@purdue.edu}}
\maketitle

\begin{abstract}
Designing scalable estimation algorithms is a core challenge in modern statistics. Here we introduce a framework to address this challenge based on parallel approximants, which yields estimators with provable properties that operate on the entirety of very large, distributed data sets. We first formalize the class of statistics which admit straightforward calculation in distributed environments through independent parallelization. We then show how to use such statistics to approximate arbitrary functional operators in appropriate spaces, yielding a general estimation framework that does not require data to reside entirely in memory. We characterize the $L^2$ approximation properties of our approach, and provide fully implemented examples of sample quantile calculation and local polynomial regression in a distributed computing environment. A variety of avenues and extensions remain open for future work.
\end{abstract}

\noindent%
{\it Keywords:}  Approximate inference, distributed learning, nonparametric estimation, parallel algorithms, statistical scalability

\vspace{\baselineskip}%
\noindent%
{\it AMS subject classifications:} 
62G05, 62B10, 65Y05, 68W10, 68W15

\section{Introduction}\label{sec:intro}

Historically, many canonical inference algorithms rely implicitly on data being stored entirely in working computer memory. By contrast, modern statistics demands scalability, requiring not only algorithms that can function efficiently within a partition-agnostic distributed data storage and computation framework, but also---equally critical but less well developed until recently---theory and methodology to enable trade-offs amongst inferential accuracy, speed, and robustness (e.g., \citep{AOS_2020_Szabo}).

Sampling-based approaches with requisite asymptotic requirements hold great appeal for general-purpose implementation (e.g., \citep{AOS_2018_Battey, JASA_2019_Wang, AOS_2021_Chen}). Modern computing environments, however, typically make use of distributed file systems enabling large-scale data storage and manipulation (e.g., Hadoop \citep{MSST_2010_Shvachko}) along with parallel processing frameworks that operate across the multiple compute and storage nodes of a scalable cluster (e.g., map-reduce \citep{CACM_2008_Dean}). This places emphasis on the need to balance independent parallel computations (easily repeatable at each node with commodity hardware) with communication between nodes (potentially costly in terms of time and memory) to aggregate intermediate results prior to final inferential output for non-parametric distributed estimation  \citet{arXiv_2022_Szabo}. Similar concerns arise in privacy-preserving computation and inference \citet{AOS_2020_Cai}.

Authors therefore tend to emphasize communication-efficient estimation algorithms (with a main earlier reference being \citep{JMLR_2013_Zhang}; see \citep{AOS_2021_Cai} for recent results on adaptivity over estimand function classes) within a general divide-and-conquer framework; for an overview of the growing literature emphasizing mathematical statistics, see for example \citet{AOS_2019_Banerjee} among others we cite here. The simplest approach to estimate a parameter of interest $\theta$ is of course simply to pool for variance reduction: supposing a data set is decomposed into $R$ disjoint subsets, one simply estimates $\hat{\theta}_{\:r}$ of $\theta$ at every $r$th node, and then averages to obtain the pooled estimator $\hat{\theta} := R^{-1} \sum_{r=1}^R \hat{\theta}_{\:r}$. Results on distributed linear regression via averaging \citet{AOS_2021_Dobriban}, quantile regression \citet{AOS_2019_Chen, AOS_2019_Volgushev}, and the estimation of principal eigenspaces \citet{AOS_2019_Fan} all contain thorough reviews of such methods and related work. Recently testing has also become a subject of investigation, exhibiting fundamentally different distributed properties than estimation in some regimes \citet{AOS_2022_Szabo}.

In this article we provide a framework to design and implement scalable estimation algorithms for very large data sets within modern partition-agnostic distributed computing environments. We focus first on what the authors of \citet{AOS_2021_Cai} and others call an independent distributed protocol: the class of statistics whose computation can be straightforwardly parallelized. We then use this class to approximate the more general functional operators that can typically arise in statistical inference. This leads in turn to concrete examples of scalable algorithms with performance guarantees. Our contribution provides practitioners with a generic technique, straightforwardly implemented using programming models such as map-reduce so that data are not required to reside entirely in memory. 

The remainder of this article is organized as follows. First, in section \ref{sec:EPS}, we introduce the notion of embarrassingly parallel statistics, providing definitions and examples along with connections to the classical notions of minimally sufficient statistics and exponential families. Next, in section \ref{sec:eps:func_opr}, we extend these concepts to complex-valued functions, which will depend on a parameter playing the role of an estimand in inference. We show how to approximate such functions through embarrassingly parallel statistics both in theory, by characterizing the $L^2$ approximation properties of our approach, as well as in practice, describing how to implement this estimation framework in a parallel, distributed computing environment. Then, in section \ref{sec:algo}, we illustrate our approach by providing two examples of statistically scalable algorithms with provable properties: a deterministic parallel scheme to approximate arbitrary sample quantiles, and a parallel approach to fitting a local regression model. We implement these approaches and provide a simulation study in section \ref{sec:sim}, and finally, we conclude in section \ref{sec:Disc} with a brief discussion of avenues and extensions that remain open for future work.

\section{Embarrassingly parallel statistics}\label{sec:EPS}

Consider the analysis of data sets whose elements $x$ have generic domain $\mathbb{X}$. Typically, an indexing scheme will be used to subdivide a generic data set $X$ (which we will often take to be a statistical sample, though without necessarily any implied technical restrictions) for the purpose of distributed storage and computation. Formally, let $X$ be a finite, nonempty multi-set and $I$ be a set representing the indexing scheme, where $|X|=|I|$ counts the total number of observations, including any repetitions. The index set $I$ may be partitioned $(R\,!)^{-1}\sum_{k=0}^R(-1)^k\binom{R}{k}(R-k)^{|I|}$ ways into $R\in\{1,\dots,|I|\}$ mutually exclusive and exhaustive nonempty subsets $I_{\:1},\dots,I_{\:R}$, such that $X$ in turn decomposes into multi-sets $X_{\:r}:=\{x_{\:i}:i\in I_{\:r}\}$ for $r\in\{1,\ldots,R\}$. Let $[I]$ denote such a partition of $I$, and $X_{[I]}$ its corresponding collection of data multi-subsets of $X$. Denote the number of subsets in $[I]$ as $\#[I] = R$. Finally, let the set of all possible sample data from $\mathbb{X}$ be denoted $\mathcal{P}(\mathbb{X})$, such that $X \in \mathcal{P}(\mathbb{X})$ and consequently likewise for any $X_{\:r}$.

\subsection{Definitions and examples}

Consider some statistic of interest $\mathsf{T}\colon\mathcal{P}(\mathbb{X})\to\mathbb{E}$, with $\mathbb{E}$ the range of $\mathsf{T}$. Typically $\mathsf{T}(X)\in\mathbb{R}^{d_{\mathsf{T}(X)}}$ for some $d_{\mathsf{T}(X)}\in\mathbb{N}^+$ which may depend on $X$ (if, for example, we let $\mathsf{T}$ be the set of order statistics). We call any $\mathsf{T}$ finite dimensional if there exists a fixed $d_{\:\mathsf{T}}<\infty$ such that $\mathsf{T}(X)\in\mathbb{R}^{d_{\:\mathsf{T}}}$ for any $X$.

Next, for an arbitrary partition $[I]$ with $\#[I]=R$ and the corresponding division in multi-sets $\big\{X_{\:1},\dots,X_{\:R}\big\}$, let $\mathsf{T}(X_{\:[I]})$ denote the collection $\big\{\mathsf{T}(X_{\:1}),\dots,\mathsf{T}(X_{\:R})\big\}\in\mathbb{E}^R$. Then we have following fundamental definitions.

\begin{definition}\label{eps:def:SEP_statistic}
We call any finite-dimensional $\mathsf{T}\colon\mathcal{P}(\mathbb{X})\to\mathbb{E}$ a Strongly Embarrassingly Parallel (SEP) statistic if for every partition $[I]$ there exists a function $\mathcal{F}_{\:[I]}\colon\mathbb{E}^{\#[I]}\to\mathbb{E}$, symmetric in its arguments, such that $\mathsf{T}(X)=\mathcal{F}_{\:[I]}\big(\mathsf{T}(X_{\:[I]})\big)$.
\end{definition}

\begin{definition}\label{eps:def:WEP_statistic}
We call a statistic $\mathsf{T}$ a Weakly Embarrassingly Parallel (WEP) statistic if $\mathsf{T}$ can be written as a function of finitely many SEP statistics.
\end{definition}

Thus an SEP statistic can be calculated in an embarrassingly parallel way with no inter-dependencies, through a function $\mathcal{F}_{\:[I]}$ applied after localized computation on $\mathsf{T}(X_{\:[I]})$. A WEP statistic can be expressed directly in terms of SEP primitives. 

These definitions enable a precise characterization of the set of functions for which exact calculations can be done via parallel computation, as the following examples illustrate.

\begin{example}[Sample mean vs.\ standard deviation]\label{eps:eg:sample_mean}
For the sample mean $\bar{X}$, observe that for any partition $[I]=\big\{I_{\:1},\dots,I_{\:R}\big\}$, we have: $\bar{X}$ $=$ $\sum_{r=1}^R\big(\nicefrac{|I_{\:r}|}{|I|}\big)\cdot \bar{X_{\:r}}$, which implies that $\bar{X}$ is SEP. Here, $\mathcal{F}_{\:[I]}$ is the weighted sum with weights $(\nicefrac{|I_{\:1}|}{|I|})$, $\dots$, $(\nicefrac{|I_{\:R}|}{|I|})$. For the sample standard deviation $\mathsf{S}(X)$, observe that: 
\begin{align*}
\mathsf{S}(X)=\sqrt{\sum_{r=1}^R\left(\frac{(|I_{\:r}|-1)}{|I|-1} {\mathsf{S}(X_{\:r})}^2+\frac{|I_{\:r}|}{|I|-1} (\bar{X_{\:r}}-\bar{X})^2\right)}.
\end{align*}
Therefore, while $\mathsf{S}(X)$ is not itself an SEP statistic, the 2-tuple $\mathsf{T}(X):=\big\{\bar{X},\mathsf{S}(X)\big\}$ \emph{is} SEP. Since $\mathsf{S}(X)$ is a function of $\mathsf{T}(X)$, it is WEP.
\end{example}

\begin{example}[Sample maximum vs.\ median]\label{eps:eg:sample_max}
Let $\mathsf{T}$ be the sample maximum (or minimum). Then $\mathsf{T}(X)=\mathcal{F}_{\:[I]}(\mathsf{T}\big(X_{\:[I]})\big)$ holds for every partition $[I]$ and corresponding collection of multi-sets $X_{\:[I]}$, with $\mathcal{F}_{\:[I]}$ being the maximum (or minimum) function. Hence $\mathsf{T}$ is SEP. However, if $\mathsf{T}$ instead denotes the sample median (i.e., fiftieth percentile by ordinal rank), then $\mathsf{T}$ is neither SEP nor WEP, because no corresponding map $\mathcal{F}_{\:[I]}$ exists, nor can $\mathsf{T}$ be expressed in terms of finitely many SEP statistics. Medians of local subsets $X_{\:1},\dots, X_{\:R}$ do not retain sufficient information to determine $\mathsf{T}$; rather, the median depends on $X$ through the entire ordered sample $\{x_{\:(1)},\ldots,x_{\:(|X|)}\}$, and so fails to respect our requirement of finite dimensionality.
\end{example}

\subsection{Minimally sufficient statistics and sums of transformations}\label{sec:eps:min_suff}

Embarrassingly parallel statistics as introduced in Definition \ref{eps:def:SEP_statistic} and Definition \ref{eps:def:WEP_statistic} above are precisely those whose structure guarantees straightforward distributed computation. They connect directly to the classical notion of minimally sufficient statistics.

\begin{theorem}\label{eps:thm:min_suff}
Suppose elements of $X$ are i.i.d.\ observations of a random variable from a family of probability distributions $P_{\:(\theta)}(.)\colon\sigma(\mathbb{X})\to(0,1)$ parameterized by $\theta$, absolutely continuous with respect to a common $\sigma$-finite measure. Then any finite-dimensional minimal sufficient statistic for $\theta$ is SEP.
\end{theorem}

\begin{proof}
See Appendix \ref{proof:eps:thm:min_suff}.
\end{proof}

\begin{example}[Normal parameter estimation]\label{eps:eg:Normal}
Suppose elements of the data $X$ are i.i.d.\ $\text{Normal}(\mu,\sigma^2)$. Then $\bar{X}$ is known to be minimally sufficient for $\mu$---and hence is SEP, just as we verified earlier in Example \ref{eps:eg:sample_mean}. If $\sigma$ is unknown but $\mu$ is known, then 
\begin{align*}
\mathsf{S'}(X):=\sqrt{\frac{\sum_{i\in I}(x_{\:i}-\mu)^2}{|X|-1}}
=
\sqrt{\sum_{r=1}^R\frac{(|I_{\:r}|-1)}{|I|-1}{\mathsf{S'}(X_{\:r})}^2}
\end{align*}
is minimally sufficient for $\sigma$---and hence $\mathsf{S'}(X)$ is SEP, which we can also observe directly by noting that $\mathsf{S'}(X)$ can be expressed entirely in terms of $\mathsf{S'}(X_{\:r})$. Finally, the 2-tuple $\big\{\bar{X},\mathsf{S}(X)\big\}$ comprising the sample mean and standard deviation is minimally sufficient for $(\mu,\sigma)$ in the case where both parameters are unknown---and hence the pair $\big\{\bar{X},\mathsf{S}(X)\big\}$ is SEP, exactly as we verified in Example \ref{eps:eg:sample_mean}.
\end{example}

Consider Theorem \ref{eps:thm:min_suff} and recall that a class of distributions is called an $L$-parameter exponential family if its densities take the general form $p_{\:(\theta)}(x)=h(x)\cdot\exp\big\{\sum_{l=1}^L\eta_{\:l}(\theta)\cdot\tau_{\:l}(x)-A(\theta)\big\}$. If elements of $X$ are i.i.d.\ observations of a random variable $x$ which admits density $p_{\:(\theta)}(x)$, and $\big\{\eta_{\:1}(\theta),\dots,\eta_{\:L}(\mathbf{\theta})\big\}$ is a linearly independent set, then the $L$-tuple $\big\{\sum_{i\in I}\tau_{\:1}(x_{\:i}),\dots,\sum_{i\in I}\tau_{\:L}(x_{\:i})\big\}$ is a minimal sufficient statistic for $\theta$. The Darmois--Koopman--Pitman theorem asserts that in such settings, exponential families are unique among distributions with fixed support in admitting minimal sufficient statistics of finite dimension. Motivated by the general form $\sum_{i\in I}\tau(x_{\:i})$, we introduce a subclass of SEP statistics that will in turn enable scalable estimation.

\begin{definition}\label{eps:def:SOT_statistic}
We call any $\mathsf{T}\colon\mathcal{P}(\mathbb{X})\to\mathbb{E}$ a Sum of Transformations (SOT) statistic if there exists a transformation $\tau\colon\mathbb{X}\to\mathbb{E}$ such that $\mathsf{T}(X)=\sum_{i\in I}\tau(x_{\:i})$.
\end{definition}

\begin{proposition}\label{eps:prop:SOT_SEP_statistic}
Fix a multi-set $X$ such that $1 < |X| < \infty$, and let $\mathcal{C}_{\:\text{SOT}}(X)$, $\mathcal{C}_{\:\text{SEP}}(X)$, and $\mathcal{C}_{\:\text{WEP}}(X)$ respectively denote the sets of all possible SOT, SEP, and WEP statistics generated by $X$. Then we have $\mathcal{C}_{\:\text{SOT}}(X)\subset\mathcal{C}_{\:\text{SEP}}(X)\subset\mathcal{C}_{\:\text{WEP}}(X)$, where both inclusions are proper.
\end{proposition}

\begin{proof}
See Appendix \ref{proof:eps:prop:SOT_SEP_statistic}.
\end{proof}

Definition \ref{eps:def:SOT_statistic} provides a tool to design distributed algorithms, because it guarantees scalability by way of Proposition \ref{eps:prop:SOT_SEP_statistic}. The key properties of this framework hold for more general algebraic structures and operations beyond the choice of $\mathbb{E}=\mathbb{R}$ endowed with addition; indeed, for any commutative and associative binary operation $\oplus$ defined on $\mathbb{E}$, every SOT statistic $\mathsf{T}(X)=\sum^{(\oplus)}_{\:i\in I}\tau(x_{\:i})$ will be SEP. This ensures the applicability of Proposition \ref{eps:prop:SOT_SEP_statistic} and what follows to broader types of summary statistics or features (e.g., graph adjacency structures, special classes of matrices, etc.).

\section{Parallel approximants for distributed inference}\label{sec:eps:func_opr}

We now extend the concepts introduced above to complex-valued functions, which will depend on a parameter $\theta \in \Theta$ that plays the role of an estimand in inference. We exhibit a class of such functions that guarantees embarrassingly parallel statistics, and then prove that a convergent sequence of approximations of this type exists for any function in $L^2(\mathbb{X}\times \Theta)$. This yields a generic approximate inference procedure that is well matched to the architecture of distributed computing systems.


To relate our data set of interest $X$ to some unknown $\theta$, where $\theta$ lies in some parameter space $\Theta$, consider complex-valued functions $\mathsf{G}(X,\theta)\colon\mathcal{P}(\mathbb{X})\times\:\Theta\to\mathbb{C}$. Given $\theta\in\Theta$ and a partition $[I]$ of $I$ with $\#[I]=R$, define $\mathsf{G}(X_{\:[I]},\theta)$ as the collection $\big\{\mathsf{G}(X_{\:1},\theta),\dots,\mathsf{G}(X_{\:R},\theta)\big\}$. 

In analogy to Definition \ref{eps:def:SEP_statistic} and Definition \ref{eps:def:SOT_statistic}, we then have the following.

\begin{definition}\label{eps:def:SEP_func}
We call any $\mathsf{G}(X,\theta) \colon \mathcal{P}(\mathbb{X})\times\:\Theta\to\mathbb{C}$ a Strongly Embarrassingly Parallel (SEP) function if for every partition $[I]$, there exists a function $\mathcal{F}_{\:[I]}^{(\theta)}\colon\mathbb{C}^{\#[I]}\to\mathbb{C}$, symmetric in its arguments, such that $\mathsf{G}(X,\theta)=\mathcal{F}_{\:[I]}^{(\theta)}\big(\mathsf{G}(X_{\:[I]},\theta)\big)$.
\end{definition}

\begin{definition}\label{eps:def:SOT_func}
We call any $\mathsf{G}(X,\theta) \colon\mathcal{P}(\mathbb{X})\times\:\Theta\to\mathbb{C}$ a Sum of Transformations (SOT) function if there exists a transformation $\Gamma\colon \mathbb{X}\times\:\Theta\to\mathbb{C}$ such that $\mathsf{G}(X,\theta)=\sum_{\:i\in I}\Gamma(x_{\:i},\theta)$. We say $\mathsf{G}(X,\theta)$ is generated by $\Gamma(x,\theta)$.
\end{definition}

In direct analogy to Proposition \ref{eps:prop:SOT_SEP_statistic}, any SOT function is an SEP function. The canonical setting of likelihood-based inference makes these concepts concrete.

\begin{example}[Data log-likelihood]\label{eps:eg:log_lhd}
Suppose elements of $X$ are i.i.d.\ observations of a random variable that admits some density $p(x,\theta)$. Letting $\Gamma(x,\theta)=\log{p(x,\theta)}$, the log-likelihood of the data $X$ is an SOT function: $\mathsf{G}(X,\theta) = \sum_{\:i\in I}\log{p(x_{\:i},\theta)}$.
\end{example}

It is natural to ask which inferential settings give rise naturally to embarrassingly parallel statistics.  More formally, consider the function $\mathfrak{g}_{\:X,\:\mathsf{G}}\colon\Theta\to\mathbb{C}$, defined as $\mathfrak{g}_{\:X,\:\mathsf{G}}(\theta):=\mathsf{G}(X,\mathbf{\theta})$ for $\theta\in\Theta$, which we see belongs to $\mathbb{C}^{\:\Theta}$, the set of functions from $\Theta$ to $\mathbb{C}$. Then, for a fixed $\mathsf{G}$, the set $\mathcal{S}_{\:\mathsf{G}}=\{\mathfrak{g}_{\:X,\:\mathsf{G}}\colon X\in\mathbb{X}\}$ is a function space. 

\begin{example}[Sample maximum vs.\ median revisited]\label{eps:eg:max_opr}
Take $\mathbb{X}=\Theta=\mathbb{R}$ and let $\Gamma(x,\theta)=\mathbf{1}(x>\theta)$, such that $\mathsf{G}(X,\theta)$ $=$ $\sum_{\:i\in I}\mathbf{1}(x_{\:i}>\theta)$. Then if we consider the following functional operator $\mathsf{T}$ on $\mathcal{S}_{\:\mathsf{G}}$: 
\begin{align*}
\mathsf{T}(\mathfrak{g}_{\:X,\:\mathsf{G}})=\inf_{\hat{\theta}\in\Theta}\big\{\hat{\theta}:\mathfrak{g}_{\:X,\:\mathsf{G}}(\hat{\theta})=\min_{\theta\in\Theta}\,\mathfrak{g}_{\:X,\:\mathsf{G}}(\theta)\big\},
\end{align*}
we see that $\mathsf{T}$ is the sample maximum, which we recall from Example \ref{eps:eg:sample_max} is SEP. However, if instead we let $\Gamma(x,\theta)=|x-\theta|$, so that $\mathsf{G}(X,\theta)$ $=$ $\sum_{\:i\in I}|x_{\:i}-\theta|$, then we recognize $\mathsf{T}$ from Example \ref{eps:eg:sample_max} as the sample median, which is neither SEP nor WEP. 
\end{example}

It is nevertheless possible to ensure that arbitrary functional operators on $\mathcal{S}_{\:\mathsf{G}}$ yield WEP statistics, by considering transformations $\Gamma(x,\theta)$ of the following form.

\begin{definition}\label{eps:def:FAS_func}
We call any $\Gamma\colon \mathbb{X}\times\:\Theta\to\mathbb{C}$ a Finitely Additively Separable (FAS) function if there exist $J \in \mathbb{N}$ pairs of maps $\{f_{\:1}(x)$, $g_{\:1}(\theta)\}$, $\dots$, $\{f_{\:J}(x)$, $g_{\:J}(\theta)\}$, with each $f_{\:\cdot}(x)\colon\mathbb{X}\to\mathbb{C}$ and $g_{\:\cdot}(\theta)\colon\Theta\to\mathbb{C}$, and constants $\eta_{\:1},\dots,\eta_{\:J}\in\mathbb{C}$, such that $\Gamma$ is of the form $\Gamma(x,\theta)=\sum_{j=1}^J\eta_{\:j}\cdot f_{\:j}(x)\cdot g_{\:j}(\theta)$.
\end{definition}

\begin{proposition}\label{eps:prop:SOT_FAS_WEP}
Any arbitrary functional operator on the function space $\mathcal{S}_{\:\mathsf{G}}$ for an SOT function $\mathsf{G}(X,\theta)$ generated by an FAS function $\Gamma(x,\theta)$ is a WEP statistic.
\end{proposition}

\begin{proof}
See Appendix \ref{proof:eps:prop:SOT-FAS-WEP}.
\end{proof}

\begin{example}[Sample moments]\label{eps:eg:FAS_func}
Take $\mathbb{X}=\Theta=\mathbb{R}$ and let $\Gamma(x,\theta):=(x^k-\theta)^2$ for some $k\in\mathbb{N}$, whence $\Gamma(x,\theta)=x^{2k}-2\theta x^k+\theta^2$. Let $f_{\:1}(x)=x^{2k}$, $f_{\:2}(x)=x^{k}$, $f_{\:3}(x)=1$; $g_{\:1}(\theta)=1$, $g_{\:2}(\theta)=\theta$, $g_{\:3}(\theta)=\theta^2$; and $\eta_{\:1}=1$, $\eta_{\:2}=-2$, $\eta_{\:3}=1$. Thus we may write $\Gamma(x,\theta)=\sum_{j=1}^3\eta_{\:j}\cdot f_{\:j}(x)\cdot g_{\:j}(\theta)$, verifying that $\Gamma$ is indeed an FAS function. If we consider the following functional operator $\mathsf{T}$ on $\mathcal{S}_{\:\mathsf{G}}$: 
\begin{align*}
\mathsf{T}(\mathfrak{g}_{\:X,\:\mathsf{G}})=\big\{\hat{\theta}\in\Theta:\mathfrak{g}_{\:X,\:\mathsf{G}}(\hat{\theta})=\min_{\theta\in\Theta}\,\mathfrak{g}_{\:X,\:\mathsf{G}}(\theta)\big\},
\end{align*}
we see that $\mathsf{T}$ is the $k$th sample moment, which is SEP. Any central sample moment is also a function of lower orders of raw sample moments, and so they are WEP.
\end{example}

Having seen several examples of SEP and WEP statistics, we now show how to approximate an arbitrary $\Gamma \in L^2(\mathbb{X}\times \Theta)$ by a sequence of FAS functions, in turn yielding embarrassingly parallel statistics as guaranteed by Proposition \ref{eps:prop:SOT_FAS_WEP}.

\begin{theorem}[Approximation by FAS functions]\label{eps:thm:SOT_FAS_L2}
For any $\Gamma \in L^2(\mathbb{X}\times \Theta)$, there exists a sequence $\Gamma_{\:1}, \Gamma_{\:2}, \ldots$ of FAS functions in $L^2(\mathbb{X}\times \Theta)$ for which $\lVert \Gamma - \Gamma_J \rVert_2 \overset{J\to\infty}{\longrightarrow} 0$. This sequence comprises the partial sums $\{\Gamma_J(x,\theta); J \in \mathbb{N} \}$ of $\Gamma_{\:\infty}$ defined as follows:
\begin{align*}
\Gamma_{\:\infty}(x,\theta) := \sum_{j=1}^\infty\eta_{\:j}\cdot u^*_{\:j}(x)\cdot v_{\:j}(\theta),
\end{align*}
where $\{u_{\:j}\}_{j=1}^\infty$ and $\{v_{\:j}\}_{j=1}^\infty$ are orthonormal systems in $L^2(\mathbb{X})$ and $L^2(\Theta)$ respectively, and the coefficients $\{\eta_{\:j}\}_{j=1}^\infty$ are non-negative reals, non-increasing in $j$. The $L^2$ approximation error of any $\Gamma_J$ is correspondingly given by $\lVert \Gamma - \Gamma_{J} \rVert_2^2 = \lVert \Gamma \rVert_2^2 - \sum_{j=1}^{J}\eta_{\:j}^2$. 

If furthermore the sequence $\{\eta_{\:j}\}_{j=1}^\infty $ is summable, $u_{\:j}(x)$ and $v_{\:j}(\theta)$ are uniformly bounded for all $j$, and $\Gamma = \Gamma_\infty$ everywhere on $\mathbb{X}\times\Theta$, then $\lVert \Gamma - \Gamma_J \rVert_\infty \overset{J\to\infty}{\longrightarrow} 0$.
\end{theorem}
\begin{proof}
See Appendix \ref{proof:eps:thm:SOT_FAS_L2}.
\end{proof}

Theorem \ref{eps:thm:SOT_FAS_L2}, together with Proposition \ref{eps:prop:SOT_FAS_WEP}, validates our choice of FAS functions as an appropriate approximating class for the purpose of enabling statistical scalability through parallelization. The orthonormal systems giving rise to the sequence of optimally approximating FAS functions depend on $\Gamma$. Guided by these results, we are free either to adapt our approximation approach to a particular choice of $\Gamma$ in a given inference problem, or to adopt fixed sets of orthonormal bases that are known to have good approximation properties for appropriately matched target function spaces. Two examples of this reasoning that we shall employ in the sequel are as follows.

\begin{example}[Approximating the modulus of a difference]\label{eps:eg:abs}
Recall from Example \ref{eps:eg:sample_max} and Example \ref{eps:eg:max_opr} the sample median, whence $\Gamma(x,\theta)=|x-\theta|$. This expression admits the following convergent Fourier expansion whenever $\mathbb{X}=\Theta=(0,1)$:
\begin{equation*}
|x-\theta|_{\:J}:=\frac{\pi}{2}-\frac{4}{\pi}\sum_{j=1}^J\frac{\cos\big((2j-1)\cdot(x-\theta)\big)}{(2j-1)^2}; \quad |x-\theta| < \pi, \quad J\in\mathbb{N}^+.
\end{equation*}
Since $\cos(A(x-\theta)) = \cos(Ax) \cdot \cos(A\theta) + \sin(Ax) \cdot \sin(A\theta)$, we deduce that $|x-\theta|_{\:J}$ is an FAS function and hence will give rise to WEP statistics.
\end{example}

\begin{example}[Approximating indicator functions]\label{eps:eg:ind}
Analogously to Example \ref{eps:eg:abs}, we may approximate the indicators $\mathsf{1}(x<\theta)$ or $\mathsf{1}(x\leq\theta)$ using the convergent FAS Fourier expansion 
\begin{equation*}
\mathsf{1}_{\:J}(x,\theta):=\frac{1}{2}-\frac{2}{\pi}\sum_{j=1}^J\frac{\sin\big((2j-1)\cdot(x-\theta)\big)}{(2j-1)}; \quad |x-\theta| < \pi, \,\, x \neq \theta, \quad J\in\mathbb{N}^+.
\end{equation*}
\end{example}

While families other that trigonometric functions may yield improved approximating properties for various target functions and function spaces, these choices lead immediately to a simple and flexible implementation of distributed estimation using parallel approximants.

\section{Distributed estimation through parallel approximants}\label{sec:algo}

We now employ the approach described in section \ref{sec:eps:func_opr} above to construct estimators for specific inferential settings. We first give a procedure to compute sample quantiles, and then we specify an approximate fitting procedure for local polynomial regression. In each case we exhibit, under appropriate technical conditions, a convergent sequence of weakly embarrassingly parallel functions suitable for implementation in a distributed computing environment.

\subsection{Sample quantile determination}\label{sec:EPF-quantile}

We first propose an approach for parallel approximation of sample quantiles, suitable whenever data sets are sufficiently large and distributed so as to preclude the typical brute-force approach of sorting a sample in its entirety. Recalling from Example \ref{eps:eg:sample_max} that the sample minimum and maximum are self-evidently SEP whereas the sample median fails even to be WEP, we shall by contrast exhibit a sequence of WEP statistics which (under suitable technical conditions) give rise to an estimator of any $p$th sample quantile, $p\in(0,1)$, with desirable theoretical properties. Not only is this approach straightforward to implement and to optimize numerically within modern distributed computing environments, as detailed in the next section of this article, but also it enables multiple sample quantiles to be computed within a single map-reduce step for various choices of $p$.

For the purposes of establishing the behavior of our approach in the large-sample limit, consider a probability triple $(\Omega,\mathcal{F},P)$ giving rise to independent and identically distributed realizations $\{x_i, i \in I\}$ of a random variable $\tilde{x}\colon\Omega\to(0,1)$. Let $\hat{q}_{\:p}(X)$ denote a $p$th sample quantile of $\tilde{x}$ based on observed data $X$, so that $\hat{q}_{\:p}(X):=\big\{\theta\in(0,1)\colon\hat{F}(X,\theta^-)\leq p\leq\hat{F}(X,\theta)\big\}$, where $\hat{F}(X,\theta):=\nicefrac{1}{|X|}\big(\sum_{i\in I}\mathsf{1}(x_{\:i}\leq\theta)\big)$ is the empirical cumulative distribution function of $\tilde{x}$ based on $X$, and $\hat{F}(X,\theta^-):=\nicefrac{1}{|X|}\big(\sum_{i\in I}\mathsf{1}(x_{\:i}<\theta)\big)$ is the left-continuous version of $\hat{F}(X,\theta)$. We then assume the following.

\begin{assumption}\label{Quantile:assumption:positive-density}
The probability measure $P$ is absolutely continuous with respect to Lebesgue measure $\lambda$ on $(0,1)$, and the Radon--Nikodym derivative $f^{\:(P)}$ of $P$ with respect to $\lambda$ is almost surely positive in the interval $(0,1)$.
\end{assumption}

\begin{assumption}\label{Quantile:assumption:uniform-convergence}
The sequence of functions $\zeta_{\:J}^{\:(P)}(\theta)\colon(0,1)\to\mathbb{R}$ converges uniformly to $f^{\:(P)}(\theta)$ for $\theta\in(0,1)$, where
\begin{align*}
\zeta_{\:J}^{\:(P)}(\theta):=\mathbb{E}_{\:P}\:\Big[\frac{1}{\pi}\cdot \frac{\sin\big(2J(\tilde{x}-\theta)\big)}{\sin(\tilde{x}-\theta)}\Big]\text{ for }J\in\mathbb{N}.
\end{align*}
\end{assumption}

Assumption \ref{Quantile:assumption:uniform-convergence} acts a smoothness condition at $0$ and $1$ relatable to the Dirichlet kernel as it arises in Fourier analysis. We explore this assumption (which can be restated in various equivalent ways) in the simulation study of section \ref{sec:simulations} below, noting that if we identify $P$ with the uniform distribution on $(0,1)$, then for any fixed $0<\delta<1$, it can be shown that $\zeta_{\:J}^{\:(P)}(\theta)$ uniformly converges to $f^{\:(P)}(\theta)$ for $\theta\in(\delta,1-\delta)$---but it cannot be concluded that $\zeta_{\:J}^{\:(P)}(\theta)$ converges uniformly to $f^{\:(P)}(\theta)$ for $\theta\in(0,1)$.

Equipped with these assumptions, we have the following result.

\begin{theorem}\label{Quantile:thm:_prop}
Consider an i.i.d.\ sample $X = \{x_i, i \in I\}$ of observations of a random variable on $(0,1)$, distributed according to some probability measure $P$. Let the setting of Assumption \ref{Quantile:assumption:positive-density} be in force and fix any $p\in(0,1)$. Then as $| X | \to \infty$, almost surely $[P]$ we have that $X$ admits a unique $p$th sample quantile $\hat{q}_{\:p}(X)$, which we identify with the well known $M$-estimator \citet{AOS_1997_Koltchinskii, Book_2005_Koenker}
\begin{align*}
\hat{q}_{\:p}(X)=\operatornamewithlimits{argmin}_{\theta\in[0,1]}\sum_{i\in I} \left[ \frac{1}{2} \left| x_{\:i}-\theta \right| + \left(p - \frac{1}{2} \right) \cdot \left( x_{\:i}-\theta \right) \right].
\end{align*}
For any $J\in\mathbb{N}$, the following approximant is a WEP statistic:
\begin{align*}
\hat{q}_{\:J,\:p}(X)=\operatornamewithlimits{argmin}_{\theta\in[0,1]}\sum_{i\in I} \left[ \frac{1}{2}| x_{\:i}-\theta|_{\:J} + \left(p - \frac{1}{2} \right) \cdot \left( x_{\:i}-\theta \right) \right],
\end{align*}
where $| x_{\:i}-\theta|_{\:J}$ is the FAS function of Example \ref{eps:eg:abs}. If Assumption \ref{Quantile:assumption:uniform-convergence} holds, then
\begin{align*}
\lim_{J\to\infty}\limsup_{|X|\to\infty}\hat{q}_{\:J,\:p}(X)\overset{a.s.\,[P]}{=}\lim_{|X|\to\infty}\hat{q}_{\:p}(X) \overset{a.s.\,[P]}{=} 
\lim_{J\to\infty}\liminf_{|X|\to\infty}\hat{q}_{\:J,\:p}(X).
\end{align*}
\end{theorem}

\begin{proof}
See Appendix \ref{proof:Quantile:thm:_prop}.
\end{proof}

Theorem \ref{Quantile:thm:_prop} thus exhibits a sequence of WEP statistics suitable to estimate the $p$th sample quantile, for any $p\in(0,1)$. In the next section we describe a map-reduce implementation of $\hat{q}_{\:J,\:p}(X)$, along with a simulation study exploring the accuracy and computational scalability of this approach as a function of $J$ for large $|X|$ ($3 \times 10^9$ simulated data points).

Comparisons are shown in the subsequent section for simulated Normal and Uniform data relative to sample quantile approximation via binning, an approach which can also be implemented in parallel for large, distributed data sets straightforwardly using a map-reduce programming model. These two comparisons also serve to highlight the importance of Assumption \ref{Quantile:assumption:uniform-convergence}, with boundary effects observed near $0$ and $1$ for data distributed uniformly on the unit interval. Such effects are seen to decrease at any fixed point near $0$ or $1$ as $J$ increases, consistent with the discussion following Assumption \ref{Quantile:assumption:uniform-convergence}.

\subsection{Local regression for distributed data}\label{sec:EPF-Lowess}

As a second example of a distributed estimation algorithm which preserves desirable large-sample properties, we consider the local regression model as introduced by \citet{JASA_1979_Cleveland}, with response variable $\tilde{y}$ and a single explanatory variable $\tilde{x}$. The random response variable $\tilde{y}$ is assumed to be stochastically related to $\tilde{x}$ as:
\begin{align*}
\tilde{y}=\mu(\tilde{x})+\epsilon,\quad \epsilon\sim\text{Normal}(0,1).
\end{align*}
In this setting, we have for every $i\in I$ the pair $(x_{\:i},y_{\:i})$. For a given fraction of observations $\alpha\in(0,1)$ and positive integer $K$, the fitted value $\hat{\mu}(x)$ at a given point $x\in(0,1)$ is obtained by fitting a $K$th-order weighted polynomial regression model using $n_{\:\alpha}=\lceil\alpha\times|X|\rceil$ points in the local neighborhood of $x$. 

For large, distributed data sets, it is computationally intensive to identify these local neighborhoods repeatedly at scale---both for initial model fitting as well as for prediction. However, as we now show, it is possible to determine these neighborhoods through the use of embarrassingly parallel statistics. To establish the theoretical properties of this approach, we assume a probability triple $(\Omega,\mathcal{F},P)$ giving rise to i.i.d.\ realizations $\{x_{\:i}, i \in I\}$ for the explanatory variable $\tilde{x}\colon\Omega\to(0,1)$. Let us therefore define $\hat{F}_{\:x}(X,h):=\big(\nicefrac{1}{|X|}\big)\cdot\sum_{i\in I}\mathsf{1}(x-h\leq x_{\:i}\leq x+h)$ for $h\in(0,1)$ as a generalization of the empirical distribution function, along with the left-continuous version thereof, $\hat{F}_{\:x}(X,h^{\text{-}}):=\big(\nicefrac{1}{|X|}\big)\cdot\sum_{i\in I}\mathsf{1}(x-h< x_{\:i}< x+h)$. 

Now define $\hat{h}_{\:\alpha,\:x}(X)$ in relation to the distance from $x$ to its $n_{\:\alpha}$th-nearest neighbor in $X$:
\begin{align*}
 h=\hat{h}_{\:\alpha,\:x}(X)\Leftrightarrow h\in(0,1) \text{ and }\hat{F}_{\:x}(X,h^{\text{-}})\leq\alpha\leq\hat{F}_{\:x}(X,h).
\end{align*}
Observe that $\hat{h}_{\:\alpha,\:x}(X)$, if it exists, need not be unique. It can be viewed as a generalization of the quantile function $\hat{q}_{\:p}(X)$ from section \ref{sec:EPF-quantile}, by comparing $\hat{q}_{\:p}(X)$ to $\hat{h}_{\:p,\:0}(X)$ for $x\in(0,1)$.

We now describe how to compute a WEP variant of $\hat{h}_{\:\alpha,\:x}(X)$. Recall from Example \ref{eps:eg:ind} that $\mathsf{1}_{\:J}(x,\theta)$ is a $J$-term Fourier approximation to $\mathsf{1}(x\leq\theta)$. Similarly we may define the following approximation to $\mathsf{1}(x-h\leq\tilde{x}\leq x+h)$, which is likewise an FAS function:
\begin{align*}
\mathsf{1}_{\:J,\:x}(\tilde{x},h):=\mathsf{1}_{\:J}(\tilde{x},x-h)-\mathsf{1}_{\:J}(\tilde{x},x+h).
\end{align*}
Thus equipped, define $\hat{F}_{\:J,\:x}(X,h):=\big(\nicefrac{1}{|X|}\big)\cdot\sum_{i\in I}\mathsf{1}_{\:J,\:x}(x_{\:i},h)$ for $h\in(0,1)$. Observe that $\hat{F}_{\:J,\:x}(X,h)$ is an SOT function generated by the FAS function $\mathsf{1}_{\:J,\:x}(\tilde{x},h)$. Hence, by Proposition \ref{eps:prop:SOT_FAS_WEP}, any solution of $\hat{F}_{\:J,\:x}(X,h)=\alpha$ for $h$ will be a WEP statistic.

As a final preparatory step, assume the following technical conditions.

\begin{assumption}\label{Lowess:assumption:positive-density}
The probability measure $P$ is absolutely continuous with respect to the Lebesgue measure $\lambda$ on $(0,1)$, and the Radon--Nikodym derivative $f^{\:(P)}$ of $P$ with respect to $\lambda$ is almost surely positive in the interval $(0,1)$.
\end{assumption}

\begin{assumption}\label{Lowess:assumption:uniform-convergence}
The sequence of functions $\eta_{\:J,\:x}^{\:(P)}(h)\colon(0,1)\to\mathbb{R}$ converges uniformly to the limit function $f_{\:x}^{\:(P)}(h)=f^{\:(P)}(x-h)+f^{\:(P)}(x+h)$ in $h\in(0,1)$ for any given $x\in(0,1)$, where
\begin{align*}
\eta_{\:J,\:x}^{\:(P)}(h):=\mathbb{E}_{\:P}\:\bigg[\frac{2}{\pi}\cdot\Big(\frac{\sin\big(2J(\tilde{x}-x+h)\big)}{\sin(\tilde{x}-x+h)}+\frac{\sin\big(2J(\tilde{x}-x-h)\big)}{\sin(\tilde{x}-x-h)}\Big)\bigg].
\end{align*}
\end{assumption}

Thus equipped, we may approximate $\hat{h}_{\:\alpha,\:x}(X)$ in the large-sample limit by choosing $J$ sufficiently large and solving $\hat{F}_{\:J,\:x}(X,h)=\alpha$ for $h$. Because we are assured that any solution will be WEP, this approach enables the identification of local neighborhoods at scale.

\begin{theorem}\label{Lowess:thm:prop}
Consider the local regression setting described above, and let Assumption \ref{Lowess:assumption:positive-density} be in force. Then for any $\alpha\in(0,1)$ and $x\in(0,1)$, $\hat{h}_{\:\alpha,\:x}(X)$ exists almost surely $[P]$. Eventually as $|X|$ and $J$ grow large, almost surely $[P]$ at least one solution $\hat{h}_{\:J,\:\alpha,\:x}(X)$ exists for $h\in(0,1)$ to the expression $\hat{F}_{\:J,\:x}(X,h)=\alpha$. If furthermore Assumption \ref{Lowess:assumption:uniform-convergence} holds, then 
\begin{align*}
\lim_{J\to\infty}\limsup_{|X|\to\infty}\hat{h}_{\:J,\:\alpha,\:x}(X)\overset{a.s.\,[P]}{=}\lim_{|X|\to\infty}\hat{h}_{\:\alpha,\:x}(X) \overset{a.s.\,[P]}{=} \lim_{J\to\infty}\liminf_{|X|\to\infty}\hat{h}_{\:J,\:\alpha,\:x}(X).
\end{align*}
\end{theorem}

\begin{proof}
See Appendix \ref{overview:proof:Lowess:thm:_prop}.
\end{proof}

We conclude this section by showing how the approximation provided by Theorem \ref{Lowess:thm:prop} is incorporated into the overall local regression procedure. To implement local regression at a chosen location $x$, we consider the exact neighborhood weight for a given data point $x_{\:i}$ to be $W_{\:i}\big(x,\hat{h}_{\:\alpha,\:x}(X)\big)$, where $W_{\:i}(x,h)=w\big(\nicefrac{|x_{\:i}-x|}{h}\big)$ and $w$ is Tukey's tri-weight function:
\begin{align*}
w(u) := 
\begin{cases}
(1-u^3)^3 & \text{if } 0 \leq u < 1,\\
0 & \text{if } u \geq 1.
\end{cases}
\end{align*}
Once the exact weights $W_{\:i}\big(x,\hat{h}_{\:\alpha,\:x}(X)\big)$ are known for all $x_{\:i}$, then a polynomial of degree $K$ can be fitted to the data by minimizing the residual sum of squares $\mathsf{RSS}\big(\boldsymbol{\beta},X,Y,x,\hat{h}_{\:\alpha,\:x}(X)\big)$ with respect to $\boldsymbol{\beta}=\big(\beta_{\:0},\dots,\beta_{\:\mathsf{K}}\big)$, where
\begin{align*}
\mathsf{RSS}(\boldsymbol{\beta},X,Y,x,h):=\sum_{i\in I}W_{\:i}(x,h)\big(y_{\:i}-\sum_{\mathsf{k}=0}^\mathsf{K}\beta_{\:k}(x_{\:i}-x)^\mathsf{k}\big)^2.
\end{align*}
It is straightforward to show that this quantity may be expressed as follows:
\begin{align*}
\mathsf{RSS}(\boldsymbol{\beta},X,Y,x,h)\!&\,=a(X,Y,x,h)-2\boldsymbol{\beta}'\underline{\boldsymbol{a}}(X,Y,x,h)+\boldsymbol{\beta}'\underline{\underline{\boldsymbol{A}}}(X,x,h)\boldsymbol{\beta},\text{ where} \\
a(X,Y,x,h)&:=\sum_{i\in I}W_{\:i}(x,h)\cdot y_{\:i}^2, \\
\underline{\boldsymbol{a}}(X,Y,x,h)&:=\Big\{\sum_{i\in I}W_{\:i}(x,h)\cdot y_{\:i}\cdot(x_{\:i}-x)^\mathsf{k}\Big\}_{\:0\leq \mathsf{k}\leq \mathsf{K}},\\
\underline{\underline{\boldsymbol{A}}}(X,x,h)&:=\Big\{\Big\{\sum_{i\in I}W_{\:i}(x,h)(x_{\:i}-x)^{\mathsf{k}+\mathsf{k}'}\Big\}\Big\}_{0\leq \mathsf{k},\mathsf{k}'\leq \mathsf{K}}.
\end{align*}
The function $\mathsf{RSS}\big(\boldsymbol{\beta},X,Y,x,\hat{h}_{\:\alpha,\:x}(X)\big)$ then achieves its minimum in $\boldsymbol{\beta}$ at the point
\begin{align*}
\hat{\boldsymbol{\beta}}(x):={\underline{\underline{\boldsymbol{A}}}\big(X,x,\hat{h}_{\:\alpha,\:x}(X)\big)}^{-1}\underline{\boldsymbol{a}}\big(X,Y,x,\hat{h}_{\:\alpha,\:x}(X)\big),
\end{align*}
so that the exact fitted value of local regression at the point $x$ is $\hat{\mu}(x)=\hat{\beta}_0(x)$. In the approximate approach to local regression fitting outlined here, we instead minimize the function $\mathsf{RSS}\big(\boldsymbol{\beta},X,Y,x,\hat{h}_{\:J,\:\alpha,\:x}(X)\big)$ for $\boldsymbol{\beta}$, yielding 
\begin{align*}
\hat{\boldsymbol{\beta}}_{\:J}(x):={\underline{\underline{\boldsymbol{A}}}\big(X,x,\hat{h}_{\:J,\:\alpha,\:x}(X)\big)}^{-1}\underline{\boldsymbol{a}}\big(X,Y,x,\hat{h}_{\:J,\:\alpha,\:x}(X)\big).
\end{align*}
Then, our approximate fitted value of the local regression at the point $x$ is $\hat{\mu}_{\:J}(x)=\hat{\beta}_{\:J,\:0}(x)$, which is a WEP statistic. In this way we have exhibited a scalable version of local regression suitable for large, distributed data sets.

\section{Implementation in a distributed setting and accompanying simulation study}\label{sec:sim}

In this section we implement the two estimators derived in section \ref{sec:algo} using a map-reduce programming model suitable for large, distributed data sets. We begin by describing how to calculate strongly and weakly embarrassingly parallel statistics in parallel, distributed settings. We then show how to reduce the computational burden of our estimators further, replacing serial computation of trigonometric terms with serial multiplication based on the algebra of orthogonal polynomials. This leads directly to a distributed algorithm which can be used to scale local regression efficiently to large data sets. Finally, we conclude this section with an illustrative end-to-end example comparing our method of sample quantile determination to a simple parallel approach based on binning.

\subsection{Implementation in a parallel, distributed computing environment}\label{sec:MapReduce}

By design, the framework we have presented above is straightforward to implement using programming models such as map-reduce \citet{CACM_2008_Dean}, commonly used for large-scale data processing in parallel, distributed computing environments. Given a finite, nonempty multi-set of interest $X$ and the set $I$ representing its indexing scheme, distributed environments work with a partition $[I] = \{I_{\:1}, \dots, I_{\:R}\}$ that decomposes $X$ into corresponding multi-sets $X_{\:r}:=\{x_{\:i}:i\in I_{\:r}\}$ for $r\in\{1,\ldots,R\}$. Depending on the environment, $[I]$ may be specified explicitly by the user (e.g., Hadoop \citep{MSST_2010_Shvachko}) or determined implicitly to optimize overall system performance (e.g., Spark \citep{CACM_2016_Zaharia}). We discuss both choices.

Map-reduce input thus takes the form of a set of key-value pairs, which we may label $(1,X_{\:1}), \ldots, (R,X_{\:R})$. Suppose SEP statistics $\mathsf{T}_{\:1}(X), \ldots, \mathsf{T}_{\:L}(X)$ are to be computed. The Map step for the $r$th key-value pair $(r,X_{\:r})$ will apply the functions $\mathsf{T}_{\:1}(\cdot), \ldots, \mathsf{T}_{\:L}(\cdot)$ to the value $X_{\:r}$ to generate subset statistics $\mathsf{T}_{\:1}(X_{\:r})$, $\dots$, $\mathsf{T}_{\:L}(X_{\:r})$, consequently yielding the $L$ key-value pairs $\big(1,\mathsf{T}_{\:1}(X_{\:r})\big), \ldots, \big(L,\mathsf{T}_{\:L}(X_{\:r})\big)$. After completion of the Map step, there are $L \cdot R$ intermediate key-value pairs: $\left\{ \big(l,\mathsf{T}_{\:l}(X_{\:r})\big); \,\, 1\leq l\leq L, \, 1\leq r\leq R \right\}$. 

Now, since each $\mathsf{T}_{\:l}(\cdot)$ is SEP, by Definition \ref{eps:def:SEP_statistic} there exists for each $l$ a function $\mathcal{F}_{\:[I],\:l}(\cdot)$ such that $\mathsf{T}_{\:l}(X)=\mathcal{F}_{\:[I],\:l}\big(\mathsf{T}_{\:l}(X_{\:1}),\dots,\mathsf{T}_{\:l}(X_{\:R})\big)$. The $l$th Reduce step therefore collects all intermediate key-value pairs with key $l$, using $\mathcal{F}_{[I],l}$ as a reducer function to convert these into the output key-value pair $\big(l,\mathsf{T}_{\:l}(X)\big)$. Map-reduce thus returns $\mathsf{T}_{\:1}(X), \dots, \mathsf{T}_{\:L}(X)$ as required.

In the case that each $\mathsf{T}_{\:l}(X)$ is SOT with associated transformation $\tau_{\:l}(x)$, then the Map step for the $r$th key-value pair $(r,X_{\:r})$ will compute for each $i\in I_{\:r}$ the $L$ terms $\tau_{\:1}(x), \ldots, \tau_{\:L}(x)$, and sum these terms over $i\in I_{\:r}$ to generate $L$ number of subset statistics $\mathsf{T}_{\:1}(X_{\:r}),\dots,\mathsf{T}_{\:L}(X_{\:r})$. The $l$th Reduce step will simply apply as $\mathcal{F}_{\:[I],\:l}(\cdot)$ the summation operator, summing the corresponding $R$ values across all intermediate key-value pairs with key $l$. Since $I$ is the disjoint union of $I_{\:1}, \ldots, I_{\:R}$, this yields output key-value pair $\big(l,\mathsf{T}_{\:l}(X)\big)$ as required. 
If Spark is used when each $\mathsf{T}_{\:l}(X)$ is SOT, then $X$ can be declared as a column of a resilient distributed data set (RDD). This initial RDD is then transformed to an intermediate RDD, using a flat-map transformation with the function $\lambda(x) = \big(\tau_{\:1}(x), \dots, \tau_{\:L}(x)\big)$. This intermediate RDD is again transformed, using the summation operator for a Reduce transformation.

Finally, suppose that WEP statistics are to be computed. Consider a family of such statistics $\{ \mathsf{W}_{\:p}(X), p \in \mathcal{P} \}$, parameterized in terms of a (potentially uncountably infinite) set $\mathcal{P}$. Then each $\mathsf{W}_{\:p}(X)$ is itself a function of finitely many SEP statistics. If this functional dependence takes the form $\mathsf{W}_{\:p}(X)=\mathsf{G}_{\:p}\big(\mathsf{T}_{\:1}(X),\dots,\mathsf{T}_{\:L}(X)\big)$ for every $p\in\mathcal{P}$, then any number of WEP statistics can be evaluated in a single map-reduce step. This is a powerful practical feature in problems that can be appropriately parameterized, as is the case for the choices of sample quantiles $p \in (0,1)$ in Theorem \ref{Quantile:thm:_prop} and fitted points $x \in (0,1)$ in Theorem \ref{Lowess:thm:prop}.

\subsection{Use of orthogonal polynomials to reduce computation}\label{sec:General:MapReduce}

In the settings of both Theorem \ref{Quantile:thm:_prop} and Theorem \ref{Lowess:thm:prop}, it is possible to reduce the computational burden further by exploiting the algebra of orthogonal polynomials.

\subsubsection{Sample quantile determination}\label{sec:Quantile:MapReduce}
First consider the case of Theorem \ref{Quantile:thm:_prop}. Here we approximate the $p$th sample quantile $\hat{q}_{\:p}(X)$ by the WEP statistic $\hat{q}_{\:J,\:p}(X)$, which is obtained as the minimizer of the objective function $\mathsf{G}_{\:p,\:J}(X,\theta)$ for $\theta\in[0,1]$, where
\begin{align*}
\mathsf{G}_{\:p,\:J}(X,\theta):=\sum_{i\in I} \left[ \frac{1}{2}| x_{\:i}-\theta|_{\:J} + \left(p - \frac{1}{2} \right) \cdot \left( x_{\:i}-\theta \right) \right].
\end{align*} 
In practice we work with the standardized objective function $\bar{\mathsf{G}}_{\:J,\:p}(X,\theta):=\nicefrac{1}{|X|}\cdot\mathsf{G}_{\:J,\:p}(X,\theta)$. For $j\in\mathbb{N}$ and $z\in\mathbb{R}$, let $\mathsf{c}_{\:2j-1}(z):=\cos\big((2j-1)z\big)$ and $\mathsf{c}_{\:2j}(z):=\sin\big((2j-1)z\big)$. Also for $j\in\mathbb{N}$, consider the SEP statistic $\mathsf{C}_{\:j}(X):=\sum_{i\in I}\mathsf{c}_{\:j}(x_{\:i})$ and its standardized counterpart $\bar{\mathsf{C}}_{\:j}(X):=\nicefrac{1}{|X|}\cdot\mathsf{C}_{\:j}(X)$. It then follows from Lemma \ref{Quantile:lemma:population-expression} in Appendix \ref{proof:Quantile:thm:_prop} that
\begin{multline*}
\bar{\mathsf{G}}_{\:p,\:J}(X,\theta)=\Big(\frac{\pi}{4}-\big(p-\frac{1}{2}\big)\cdot \theta\Big)+\big(p-\frac{1}{2}\big)\cdot\bar{X}\\
-\:\frac{2}{\pi}\cdot \sum_{j=1}^J\frac{\Big(\bar{\mathsf{C}}_{\:2j-1}(X)\cdot \mathsf{c}_{\:2j-1}(\theta)+\bar{\mathsf{C}}_{\:2j}(X)\cdot \mathsf{c}_{\:2j}(\theta)\Big)}{(2j-1)^{2}}.
\end{multline*}

Conceptually, we proceed as follows. Consider an arbitrary set $\mathcal{Q}\subseteq (0,1)$, where we wish to compute $\hat{q}_{\:J,\:p}(X)$ for all $p\in\mathcal{Q}$ and some fixed $J$. First, we transform elements of the input data $X$ to the interval $(0,1)$. We then compute the $2J+1$ SEP statistics $|X|,\mathsf{C}_{\:1}(X),\dots,\mathsf{C}_{\:2J}(X)$ in a map-reduce step as described in section \ref{sec:MapReduce}. Since $\mathsf{G}_{\:J,\:p}(X,\theta)$ is an SOT function generated by the FAS function $\nicefrac{1}{2} \cdot | x -\theta|_{\:J} + \left(p - \nicefrac{1}{2} \right) \cdot \left( x -\theta \right)$, Proposition \ref{eps:prop:SOT_FAS_WEP} applies, and so its minimizer will be a WEP statistic. By minimizing $\bar{\mathsf{G}}_{\:J,\:p}(X,\theta)$ (inverse transforming its minimizer if necessary), we obtain the approximate quantile $\hat{q}_{\:J,\:p}(X)$.

This conceptual approach can be improved by the use of orthogonal polynomials to replace serial computation of trigonometric functions by serial multiplication \citet{Thesis_2019_Chakravorty}. To do so we shall require $K$ integer parameters $J_{\:1},\dots,J_{\:K}$ such that $J=\prod_{k=1}^K J_{\:k}$. Let $\boldsymbol{J}:=\big(J_{\:1},\dots,J_{\:K}\big)$ and define integers $L_{\:k}:=\prod_{l=1}^k J_{\:l}$ for $1\leq k<K$. Also define the index set
\begin{align*}
\mathbb{N}_{\{\boldsymbol{J}\}}:=\{1,\dots,2 \cdot J_{\:1}\big\}\times\{1,\dots,J_{\:2}\big\}\times\dots\times\{1,\dots,J_{\:K}\big\}.
\end{align*}

Let $\mathsf{c}^{\:2j-1}(x):=\cos^{2j-1}(x)$ and $\mathsf{c}^{\:2j}(x):=\sin(x)\cos^{2j-2}(x)$ for $j\in\mathbb{N}^+$. Given $J'\in\mathbb{N}^+$, let us also define $\tilde{\mathsf{c}}^{\:j}(x,J')=\cos^{j-1}(2J'x)$ for $j\in\mathbb{N}^+$. Then, for the $K$-dimensional vector $\boldsymbol{j}=\big(j_{\:1},j_{\:2},\dots,j_{\:K}\big)\in\mathbb{N}_{\{\boldsymbol{J}\}}$, define 
\begin{align*}
\mathsf{c}^{\boldsymbol{j}}(x,\boldsymbol{J})=\mathsf{c}^{\:j_{\:1},\:j_{\:2},\:\dots,\:j_{\:K}}(x,\boldsymbol{J}):=\mathsf{c}^{\:j_{\:1}}(x)\cdot \prod_{k=2}^K\tilde{\mathsf{c}}^{\:j_{\:k}}(x,L_{\:k-1}).
\end{align*}
Now, define the $K$-dimensional array-valued function
\begin{align*}
\boldsymbol{\mathsf{c}}^{\:(1:2J_{\:1}),\:(1:J_{\:2}),\:\dots\:,\:(1:J_{\:K})}(x,\boldsymbol{J}):=\{\{\mathsf{c}^{\boldsymbol{j}}(x,\boldsymbol{J})\}\}_{\left\{\boldsymbol{j}\in\mathbb{N}_{\{\boldsymbol{J}\}}\right\}}.
\end{align*}
Finally, let us define the following $K$-dimensional array-valued SEP statistic, with dimensions $\{2J_{\:1},J_{\:2},\dots,J_{\:K}\}$, and its corresponding standardized version:
\begin{align*}
&\boldsymbol{\mathsf{C}}^{\:(1:2J_{\:1}),\:(1:J_{\:2}),\:\dots,\:(1:J_{\:K})}(X,\boldsymbol{J}):=\sum_{i\in I}\boldsymbol{\mathsf{c}}^{\:(1:2J_{\:1}),\:(1:J_{\:2}),\:\dots,\:(1:J_{\:K})}(x_{\:i},\boldsymbol{J}), \\
&\bar{\boldsymbol{\mathsf{C}}}^{\:(1:2J_{\:1}),\:(1:J_{\:2}),\:\dots,\:(1:J_{\:K})}(X,\boldsymbol{J}):=\nicefrac{1}{|X|}\cdot\boldsymbol{\mathsf{C}}^{\:(1:2J_{\:1}),\:(1:J_{\:2}),\:\dots,\:(1:J_{\:K})}(X,\boldsymbol{J}),
\end{align*}
where we observe that for $\boldsymbol{j}\in\mathbb{N}_{\{\boldsymbol{J}\}}$, the $\boldsymbol{j}$th element of $\bar{\boldsymbol{\mathsf{C}}}^{\:(1:2J_{\:1}),\:(1:J_{\:2}),\:\dots,\:(1:J_{\:K})}(X,\boldsymbol{J})$ is
\begin{align*}
\bar{\mathsf{C}}^{\:\boldsymbol{j}}(X,\boldsymbol{J})=\nicefrac{1}{|X|}\cdot\mathsf{C}^{\:\boldsymbol{j}}(X,\boldsymbol{J})=\nicefrac{1}{|X|}\cdot\sum_{i\in I}\mathsf{c}^{\:\boldsymbol{j}}(x_{\:i},\boldsymbol{J}).
\end{align*}

It can then be shown (\citet[Section 3.6]{Thesis_2019_Chakravorty}) that for $K\in\mathbb{N}^+$ and $\boldsymbol{J}\in{\mathbb{N}^+}^K$, there exists a linear transformation $\mathcal{T}^{(K)}_{\:\boldsymbol{J}}:\mathbb{R}^{2J_{\:1}}\times\dots\times\mathbb{R}^{J_{\:K}}\to\mathbb{R}^{2J}$ such that
\begin{align*}
\bar{\boldsymbol{\mathsf{C}}}_{\:(1:2J)}(X)=\mathcal{T}^{(K)}_{\:\boldsymbol{J}}\big(\bar{\boldsymbol{\mathsf{C}}}^{\:(1:2J_{\:1}),\:(1:J_{\:2}),\:\dots,\:(1:J_{\:K})}(X,\boldsymbol{J})\big).
\end{align*}

This result implies that instead of directly computing SEP statistics $\mathsf{C}_{\:j}(X)$ for $j=1,\dots,2J$ in a map-reduce step (along with $|X|$), we may instead compute SEP statistics $\mathsf{C}^{\:\boldsymbol{j}}(X)$ for $\boldsymbol{j}\in\mathbb{N}_{\{\boldsymbol{J}\}}$ (along with $|X|$). While $2J+1$ SEP statistics must still be computed, we now have only $K \ll J$ cosine terms to be computed for each observation $x_{\:i}$. The Reduce operation in turn still involves summation over subsets. Finally, after the entirety of this map-reduce step, the statistic $\mathsf{C}^{\:\boldsymbol{j}}(X,\boldsymbol{J})$ is standardized to obtain the SEP statistic $\bar{\mathsf{C}}^{\:\boldsymbol{j}}(X,\boldsymbol{J})$ for $\boldsymbol{j}\in\mathbb{N}_{\{\boldsymbol{J}\}}$, and $\bar{\mathsf{C}}_{\:j}(X)$ for $j\in\{1,2,\dots,2J\}$ is then obtained efficiently in aggregate by applying the linear transformation $\mathcal{T}^{(K)}_{\:\boldsymbol{J}}$. We then denote our final result by $\hat{q}_{\:\boldsymbol{J},\:p}(X)$, noting that if $\boldsymbol{J}=\big(J_{\:1},\dots,J_{\:K}\big)$, then $\hat{q}_{\:\boldsymbol{J},\:p}(X) = \hat{q}_{\:J,\:p}(X)$ for $J=\prod_{k=1}^K J_{\:k}$.

\subsubsection{Local regression for distributed data}

For local regression, we employ a two-step map-reduce approach based on Theorem \ref{Lowess:thm:prop}, with the aid of the transformation $\mathcal{T}^{(K)}_{\:\boldsymbol{J}}$ introduced in section \ref{sec:Quantile:MapReduce} above. We consider a distributed set $\{X,Y\}$ of training data, indexed by the usual index set $I$ such that we have the pair $\{x_{\:i},y_{\:i}\}$ for each $i\in I$, and a test set $X_{\:0}$ stored in local memory, for which it is desired to predict the response variable $\tilde{y}$ for each $x\in X_{\:0}$. In the first map-reduce step, we compute $2J$ SEP statistics: $\bar{\mathsf{C}}^{\:\boldsymbol{j}}(X)$ for $\boldsymbol{j}\in\mathbb{N}_{\{\boldsymbol{J}\}}$, as described in section \ref{sec:MapReduce}. We then obtain $\bar{\mathsf{C}}^{\:j}(X)$ for $j\in\{1,2,\dots,2J\}$ by the transformation $\mathcal{T}^{(K)}_{\:\boldsymbol{J}}$ outlined in section~\ref{sec:Quantile:MapReduce}. 
Next, from lemma \ref{Lowess:lemma:delta_expression} of the Supplementary Material,
\begin{align}
\ \hat{F}_{\:J,\:x}(X,h)=\frac{4}{\pi}\cdot \sum_{j=1}^J\big(\bar{\mathsf{C}}_{\:2j-1}(X)\cdot\mathsf{c}_{\:2j-1}(x)+\bar{\mathsf{C}}_{\:2j}(X)\cdot\mathsf{c}_{\:2j}(x)\big)\cdot\frac{\mathsf{c}_{\:2j}(h)}{2j-1}.
\end{align}

Then, for each $x\in X_{\:0}$, we solve the equation $\hat{F}_{\:J,\:x}(X,h) = \alpha$ as observed in Theorem \ref{Lowess:thm:prop}, (noting that such a solution exists for sufficiently small $\alpha$) to obtain the corresponding value of $\hat{h}_{\:J,\:\alpha,\:x}(X)$. In the second map-reduce step we compute the $\mathsf{K}\times \mathsf{K}$ matrix-valued SEP statistic $\underline{\underline{\boldsymbol{A}}}\big(X,x,\hat{h}_{\:J,\:\alpha,\:x}(X)\big)$ as well as the $\mathsf{K}\times 1$ array-valued SEP statistic $\underline{\boldsymbol{a}}\big(X,Y,x,\hat{h}_{\:J,\:\alpha,\:x}(X)\big)$ for $x\in X_{\:0}$, as described in section \ref{sec:MapReduce}. In the post-Reduce step, we minimize the function $\mathsf{RSS}\big(\boldsymbol{\beta},X,Y,x,\hat{h}_{\:J,\:\alpha,\:x}(X)\big)$ for $\boldsymbol{\beta}$ by computing for each $x\in X_{\:0}$
\begin{align*}
\hat{\boldsymbol{\beta}}_{\:J}(x)={\underline{\underline{\boldsymbol{A}}}\big(X,x,\hat{h}_{\:J,\:\alpha,\:x}(X)\big)}^{-1}\underline{\boldsymbol{a}}\big(X,Y,x,\hat{h}_{\:J,\:\alpha,\:x}(X)\big).
\end{align*}
Finally, our approximate fitted value of the local regression at the point $x$ is $\hat{\mu}_{\:J}(x)=\hat{\beta}_{\:J,\:0}(x)$, for each $x\in X_{\:0}$.

\subsection{Comparing to sample quantile approximation via binning in a distributed setting}\label{sec:simulations}

We now report the results of a simulation study comparing our method of determining $\hat{q}_{\:\boldsymbol{J},\:p}(X)$ to a standard method based on linearly interpolating histogram bin counts, implemented within the Hadoop distributed file system (HDFS) with $|X|$ on the order of three billion observations (22.4 gigabytes in HDFS). Specifying an invertible distribution function $F^{\:(P)}(x)$ on $(0,1)$ and letting $N = |X|+1$, we first take $X$ to be a uniform random permutation of the values ${F^{\:(P)}}^{\text{-}1}(\nicefrac{1}{N}),{F^{\:(P)}}^{\text{-}1}(\nicefrac{2}{N}),\dots,{F^{\:(P)}}^{\text{-}1}(\nicefrac{N-1}{N})$. These permuted values are then identified with the corresponding set $\{q_{\:p}^{\:(P)}:p=\nicefrac{1}{N},\nicefrac{2}{N},\dots,\nicefrac{N-1}{N}\}$ of population quantiles of ${F^{\:(P)}}$. This manner of simulation avoids the computation of exact sample quantiles $\hat{q}_{\:p}^{\:(P)}(X)$ by brute-force sorting when $|X|$ is large and distributed, since the difference $\big|\:q_{\:p}^{\:(P)}-\hat{q}_{\:p}^{\:(P)}(X)\:\big|$ can be expected to be negligible for sufficiently large $|X|$.

Approximating quantiles by linearly interpolating histogram bin counts is a natural comparison that can be implemented straightforwardly in parallel using map-reduce. Here the range of $X$ is partitioned into $B$ equi-spaced bins via the intervals $(b_{\:0},b_{\:1}]$, $(b_{\:1},b_{\:2}]$, $\dots$, $(b_{\:B-1},b_{\:B})$. The number of observations in each bin is then calculated, leading for $r=1,\dots,B$ to the set of SEP statistics $\mathsf{T}_{\:r}(X)=\sum_{i\in I}\mathsf{1}(b_{r-1}<x_{\:i}\leq b_r)$. From these cumulative frequency counts, quantile values can be approximated by linear interpolation. 

As a baseline computational comparison, using binary search to assign an $x_{\:i}$ to one of $B$ bins requires $\mathcal{O}(\log B)$ operations on average for each $i\in I$, whereas calculating $\hat{q}_{\:J,\:p}(X)$ using the linear transformation $\mathcal{T}^{(K)}_{\:\boldsymbol{J}}$ requires computation of $K \ll J$ cosine terms for each $i\in I$, with $J=\prod_{k=1}^K J_{\:k}$. The logarithm of the maximal coefficient in $\mathcal{T}^{(K)}_{\:\boldsymbol{J}}$ grows as $(2J_{\:1}-1)+\sum_{k=2}^K (J_{\:k}-1)$, however, and so $K$ should grow as $J$ grows in order to curtail the growth of round-off error in any practical implementation. In the setting considered here, with $J \in \{ 32, 108, 128, 256, 648 \}$, we found $K=3$ or $K=4$ to be  adequate. Minimizing $(2J_{\:1}-1)+\sum_{k=2}^K (J_{\:k}-1)$ for fixed $J$ and $K$ then leads naturally to the factorization choices $\boldsymbol{J} \in \{(2,4,4), (3,6,6), (2,4,4,4), (4,8,8), (3,6,6,6)\}$.

For this simulation study we took $N = 3 \times 10^{9}$ and considered two examples from the Beta family of distributions: a $\operatorname{Beta}(19,19)$ density, which behaves like a Normal density rescaled to the unit interval in a way that satisfies Assumptions \ref{Quantile:assumption:positive-density} and \ref{Quantile:assumption:uniform-convergence} (see section \ref{sec:verify_beta(19,19)} of the accompanying Supplementary Material), and a $\operatorname{Beta}(1,1)$ or $\operatorname{Uniform}(0,1)$ distribution, which satisfies Assumption \ref{Quantile:assumption:positive-density} but not Assumption \ref{Quantile:assumption:uniform-convergence} (see section \ref{sec:verify_uniform(0,1)} of the Supplementary Material). In each case we compared $\hat{q}_{\:\boldsymbol{J},\:p}(X)$, for the choices of $\boldsymbol{J}$ listed above, with $\hat{q}_{\:p}^{\:B}(X)$ based on linearly interpolating $B$ histogram bin counts, with $B=1\times 10^4,\dots,5\times 10^4$ bins.

\begin{figure}[t]
\begin{center}
\includegraphics[width=1\columnwidth]{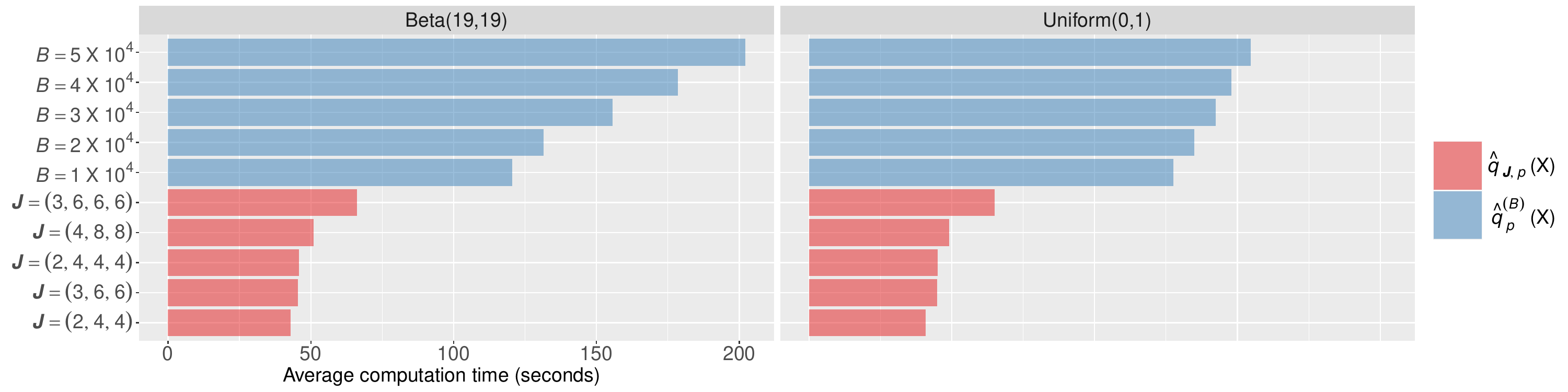}
\caption{\label{fig:runtime} Running times (averaged over 30 repetitions) for map-reduce implementations of quantile computations, based on order $J=\prod_{k=1}^K J_{\:k}$ parallel approximants ($\hat{q}_{\:\boldsymbol{J},\:p}(X)$, red) and on linearly interpolating $B$ histogram bin counts ($\hat{q}_{\:p}^{\:B}(X)$, blue), for near-Normal and uniform distributions, and with $|X| = 3 \times 10^{9}-1$ observations.}
\end{center}
\end{figure}

Figure~\ref{fig:runtime} compares the running times of these approaches on a $200$-core cluster capable of running Hadoop jobs, with each of $199$ cores assigned to run one process at a time. Simulated data were partitioned into $199$ approximately equi-sized blocks within HDFS, so that $198$ of these blocks contained $201$ subsets, with each subset having $7.5 \times 10^{4}$ observations, and the final block contained $202$ subsets. This procedure yields blocks of approximate size $115$ MB in HDFS, which is well within the recommended block size range for a Hadoop job.

It is immediately apparent from Fig.~\ref{fig:runtime} that for the largest value of $J=648$ considered here, $\hat{q}_{\:\boldsymbol{J},\:p}(X)$ is faster to calculate than $\hat{q}_{\:p}^{\:B}(X)$, even for the smallest value of $B = 1 \times 10^4$ considered here. Figure~\ref{fig:Beta_d_244} shows that, for the case of a $\operatorname{Beta}(19,19)$ density and the smallest value of $J$ considered here ($\boldsymbol{J}=(2,4,4)$, so that $J=2\times 4\times4 =32$), the median error in quantile computation is lower for $\hat{q}_{\:J,\:p}(X)$ than for any $\hat{q}_{\:p}^{\:B}(X)$, with $B=1,\dots,5\times 10^4$.

\begin{figure}[t]
\begin{center}
\includegraphics[width=\columnwidth]{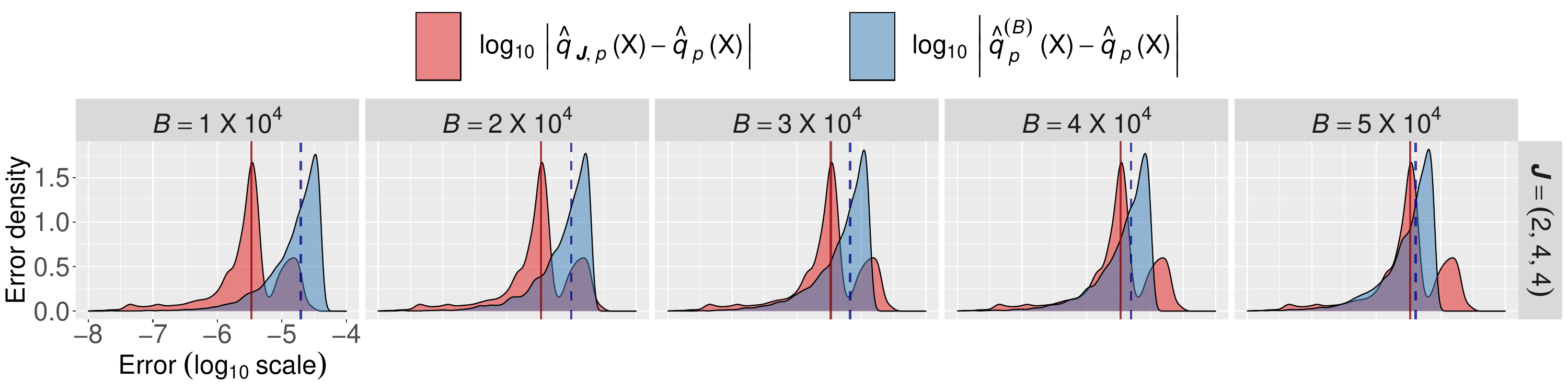}
\caption{\label{fig:Beta_d_244}Quantile computation errors for $\hat{q}_{\:(2,4,4),\:p}(X)$---the smallest value of $J$ considered---relative to $\hat{q}_{\:p}^{\:B}(X)$, for a $\operatorname{Beta(19,19)}$ density and $|X| = 3 \times 10^{9}-1$ observations. Increasing $B$ to achieve near-equal median error over all $p$ (vertical lines) requires over four times as much computation time as $\hat{q}_{\:(2,4,4),\:p}(X)$ does (cf.~Fig.~\ref{fig:runtime}).}
\end{center}
\end{figure}

\begin{figure}
\begin{center}
\includegraphics[width=1\columnwidth]{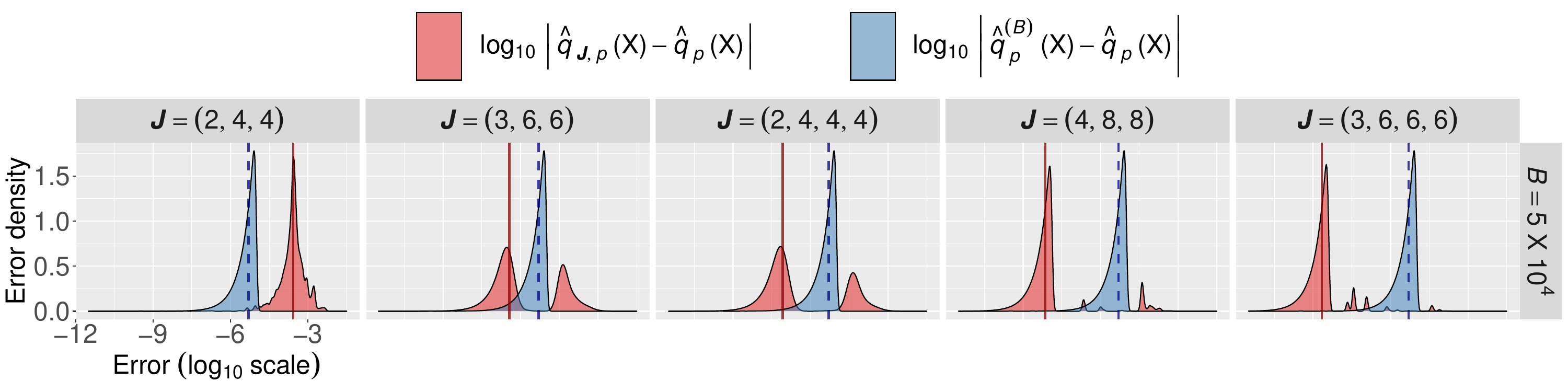}
\caption{\label{fig:uniform_b_5k} Quantile computation errors for $\hat{q}_{\:p}^{\:5\times 10^4}(X)$---the largest value of $B$ considered---relative to $\hat{q}_{\:\boldsymbol{J},\:p}(X)$,  with increasing values of $J$, for a $\operatorname{Uniform(0,1)}$ distribution and $|X| = 3 \times 10^{9}-1$ observations. Relative omputation times for $\hat{q}_{\:\boldsymbol{J},\:p}(X)$ remain uniformly lower (cf.~Fig.~\ref{fig:runtime}), and median errors quickly outperform $\hat{q}_{\:p}^{\:5\times 10^4}(X)$.}
\end{center}
\end{figure}

\begin{figure}
\begin{center}
\includegraphics[width=1\columnwidth]{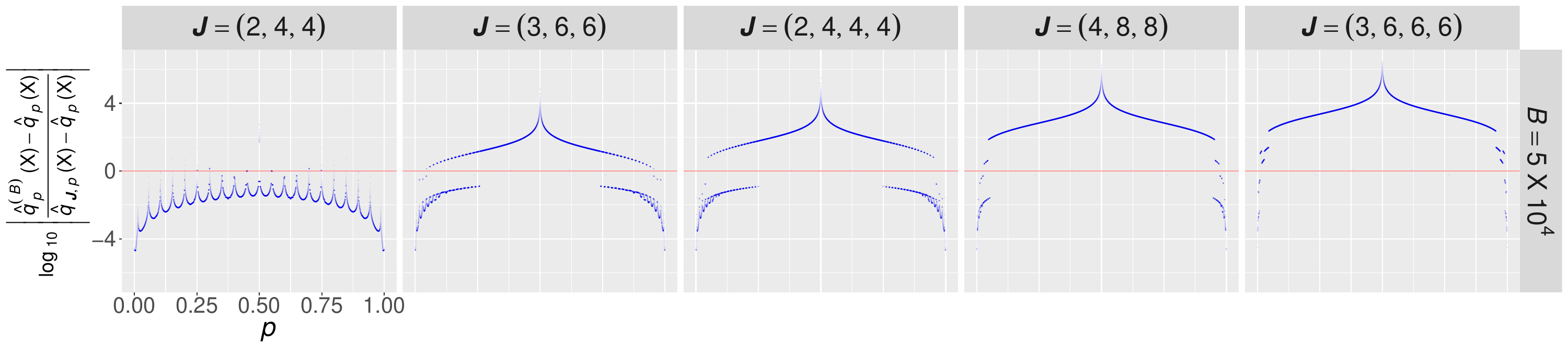}
\caption{\label{fig:uniform_xy_5k} Relative error of Fig.~\ref{fig:uniform_b_5k} as a function of $p$, showing (as $J$ increases) the effect of violating Assumption \ref{Quantile:assumption:uniform-convergence}.}
\end{center}
\end{figure}

Figure~\ref{fig:uniform_b_5k} shows a scenario similar to Fig.~\ref{fig:Beta_d_244} for a $\operatorname{Uniform(0,1)}$ distribution: superior relative performance at a lower computational cost. Finally, Fig.~\ref{fig:uniform_xy_5k} shows the implication of violating the smoothness condition of Assumption \ref{Quantile:assumption:uniform-convergence}: As discussed in section~\ref{sec:algo}, when $p$ is close to the boundary points $0$ and $1$, the guarantee of \emph{uniform} convergence is lost.

\section{Discussion}\label{sec:Disc}

In this article we have introduced the concept of embarrassingly parallel statistics---of both strong and weak type---in order to enable statistical scalability and approximate inference in distributed computing environments. By introducing appropriate approximations to functions which on parameters that play the role of estimands in inference, and then providing with guarantees on limiting behavior, we have demonstrated how to build scalable inference algorithms for very large data sets. We have provided two concrete examples, sample quantile approximation and local regression fitting, each of which comes with theoretical guarantees and admits straightforward implementation via programming models such as map-reduce.

\begin{appendices}

\section{Proof of Theorem \ref{eps:thm:min_suff}}\label{proof:eps:thm:min_suff}

Let the probability density function associated with $P_{(\theta)}(\cdot)$ be denoted $p_{(\theta)}(x)$. Also, let $\mathsf{T}(X)$ be a finite dimensional minimal sufficient statistic for $\theta$, if we assume that such a statistic exists (otherwise there is nothing to prove). Let $X_{\:[I]}=\big\{\:X_{\:1},\dots,X_{\:R}\:\big\}$ denote the collection of multi-subsets of $X$ for an arbitrary partition $[I]$. Since $\mathsf{T}(X)$ is sufficient, by Fisher--Neyman's factorization theorem, there exist non-negative functions $g_{(\theta)}\colon\Theta\to\mathbb{R}$ and $h\colon\mathbb{X}\to\mathbb{R}$, such that: $p_{(\theta)}(X_{\:r})=g_{(\theta)}\big(T(X_{\:r})\big)\cdot h(X_{\:r})$ for $r=1,\dots,R$. Now since the sub-samples $X_{\:1},\dots,X_{\:R}$ are mutually independent, we have:
\begin{align*}
p_{(\theta)}(X)&=\prod_{r=1}^Rp_{(\theta)}(X_{\:r})=\prod_{r=1}^Rg_{(\theta)}\big(\mathsf{T}(X_{\:r})\big)\cdot h(X_{\:r}) \\
&=\prod_{r=1}^Rg_{(\theta)}\big(T(X_{\:r})\big)\cdot \prod_{r=1}^Rh(X_{\:r}).
\end{align*}

Another application of the factorization theorem shows that the collection $\mathsf{T}(X_{\:[I]})$ is sufficient for $\theta$. From the definition of minimal sufficiency, it follows in turn that $\mathsf{T}(X)$ is a function of $\mathsf{T}\big(X_{\:[I]}\big)$. Thus there exists $\mathcal{F}_{\:[I]}$ such that $\mathsf{T}(X)=\mathcal{F}_{\:[I]}\big(\mathsf{T}(X_{\:[I]})\big)$, and hence in accordance with Definition \ref{eps:def:SEP_statistic}, $\mathsf{T}(X)$ is SEP as claimed.

\section{Proof of Proposition \ref{eps:prop:SOT_SEP_statistic}}\label{proof:eps:prop:SOT_SEP_statistic}

It follows directly from Definition \ref{eps:def:SEP_statistic} and Definition \ref{eps:def:WEP_statistic} that $\mathcal{C}_{\:\text{SEP}}(X)\subset\mathcal{C}_{\:\text{WEP}}(X)$ but that $\mathcal{C}_{\:\text{WEP}}(X)\not\subset\mathcal{C}_{\:\text{SEP}}(X)$, and thus $\mathcal{C}_{\:\text{SEP}}(X)$ is a proper subset of $\mathcal{C}_{\:\text{WEP}}(X)$ as claimed. 

To show that $\mathcal{C}_{\:\text{SOT}}(X)\subset\mathcal{C}_{\:\text{SEP}}(X)$, suppose $\mathsf{T}(X)$ is an SOT statistic, with associated transformation $\tau(x)$. Given any partition $[I]=\big\{\:I_{\:1},\dots,I_{\:R}\:\big\}$ of $I$, let $X_{\:[I]}=\big\{\:X_{\:1},\dots,X_{\:R\:}\big\}$ be the corresponding collection of multi-sets. Since $I$ is the disjoint union of $I_{\:1},\dots,I_{\:R}$, we have:
\begin{align*}
\mathsf{T}(X)& =\sum_{i\in I}\tau(x_{\:i})=\sum_{i\in\cup_{r=1}^RI_{\:r}}\tau(x_{\:i})=\sum_{r=1}^R\sum_{x_{\:i}\in X_{\:r}}\tau(x_{\:i}) \\
&=\sum_{r=1}^R\mathsf{T}(X_{\:r})=\mathcal{F}_{[I]}(\mathsf{T}(X_{[I]}),
\end{align*}
with $\mathcal{F}_{\:[I]}$ the summation operator. Hence in accordance with Definition \ref{eps:def:SEP_statistic}, $\mathsf{T}(X)$ is SEP, thereby establishing that $\mathcal{C}_{\:\text{SOT}}(X)\subset\mathcal{C}_{\:\text{SEP}}(X)$. 

To show that $\mathcal{C}_{\:\text{SEP}}(X)\not\subset\mathcal{C}_{\:\text{SOT}}(X)$, and hence that $\mathcal{C}_{\:\text{SOT}}(X)$ is a proper subset of $\mathcal{C}_{\:\text{SEP}}(X)$ as claimed, suppose elements of $X$ are i.i.d.\ observations of a random variable distributed uniformly over the interval $[0,\theta]$. Then $\mathsf{T}(X):=\max(X)$ is a minimal sufficient statistic for parameter $\theta$, and is furthermore of finite dimension. Hence by Theorem \ref{eps:thm:min_suff}, $\mathsf{T}(X)\in \mathcal{C}_{\:\text{SEP}}(X)$. However, since $\mathsf{T}(X)$ cannot be written in the form $\sum_{i\in I}\tau(x_{\:i})$, it follows from Definition \ref{eps:def:SOT_statistic} that $\mathsf{T}(X)\not\in \mathcal{C}_{\:\text{SOT}}(X)$.

\section{Proof of Proposition \ref{eps:prop:SOT_FAS_WEP}}\label{proof:eps:prop:SOT-FAS-WEP}

Let the SOT function $\mathsf{G}(X,\theta)$ be generated by $\Gamma(x,\theta)$, with  $\Gamma(x,\theta)=\sum_{j=1}^J\eta_{\:j \cdot }f_{\:j}(x) \cdot g_{\:j}(\theta)$. Thus
\begin{align*}
\mathsf{G}(X,\theta)=\sum_{i\in I}\Gamma(x_{\:i},\theta)=\sum_{i\in I}\sum_{j=0}^J\eta_{\:j}\cdot f_{\:j}(x_{\:i})\cdot g_{\:j}(\theta)=\sum_{j=0}^J\eta_{\:j}\cdot \mathsf{F}_{\:j}(X)\cdot g_{\:j}(\theta),
\end{align*}
where each $\mathsf{F}_{\:j}(X)$ is an SOT (and thus, by Proposition \ref{eps:prop:SOT_SEP_statistic}, an SEP) statistic with associated transformation $f_{\:j}(x)$. We see that for any fixed $\theta$, the SOT function $\mathsf{G}(X,\theta)$ depends on $X$ solely through the values of the SEP statistics $\mathsf{F}_{\:1}(X),\dots,\mathsf{F}_{\:J}(X)$. So, the set $\mathcal{S}_{\:\mathsf{G}}$ is characterized fully by the set $\big\{\mathsf{F}_{\:1}(X),\dots,\mathsf{F}_{\:J}(X)\big\}$, and thus any arbitrary operator on $\mathcal{S}_{\:\mathsf{G}}$ must be a function of these $J$ SEP statistics. Hence in accordance with Definition \ref{eps:def:WEP_statistic}, any such operator constitutes a WEP statistic as claimed.

\section{Proof of Theorem \ref{eps:thm:SOT_FAS_L2}}\label{proof:eps:thm:SOT_FAS_L2}

This result is a consequence of \citet[Theorem~4]{AM_1992_Simsa}, where the author derives a Hilbert--Schmidt-type decomposition $\Gamma_\infty$ of any function $\Gamma \in L^2(\mathbb{X}\times\Theta)$ by way of (in the language of Definition \ref{eps:def:FAS_func}) finitely additively separable functions, so that $\Gamma = \Gamma_\infty$ holds almost everywhere on $\mathbb{X}\times\Theta$. The number of non-zero coefficients $\eta_j$ in such a decomposition will in fact be precisely the number of non-zero eigenvalues of a certain positive semi-definite operator related to $\Gamma$, which will be finite if and only if $\Gamma$ lies in a weakly closed subset comprising a finite sum of products of uni-variate functions respectively in $L^2(\mathbb{X})$ and $L^2(\Theta)$. 

The extension from convergence in $L^2(\mathbb{X}\times\Theta)$ to uniform convergence proceeds as follows. First, with respect to the orthonormal systems $\{u_{\:j}\}_{j=1}^\infty$ and $\{v_{\:j}\}_{j=1}^\infty$, suppose there exist real numbers $M_{\:x} \geq \sup_{x\in\mathbb{X},j\in\mathbb{N}} |u_{\:j}(x)|$ and $M_{\:\theta} \geq \sup_{\theta\in\Theta,j\in\mathbb{N}} |v_{\:j}(\theta)|$. Then if furthermore $\Gamma(x,\theta) = \Gamma_\infty(x,\theta)$ holds for all $(x,\theta)\in\mathbb{X}\times\Theta$, we may write:
\begin{align*}
\lVert\Gamma-\Gamma_J\rVert_\infty 
&= \lVert\sum_{j=J+1}^{\infty} \! \eta_{\:j}u^*_{\:j}v_{\:j}\rVert_\infty \\
&\leq \sum_{j=J+1}^{\infty} \! \eta_{\:j} \sup_{x\in\mathbb{X}} |u_{\:j}(x)| \sup_{\theta\in\Theta} |v_{\:j}(\theta)| \leq M_{\:x}\cdot M_{\:\theta} \cdot \! \sum_{j=J+1}^{\infty} \! \eta_{\:j}.
\end{align*}
Finally, we see that if $\sum_{j=1}^\infty\eta_{\:j}<\infty$, then $\lVert \Gamma - \Gamma_J \rVert_\infty \overset{J\to\infty}{\longrightarrow} 0$ as claimed.

\section{Proof of Theorem \ref{Quantile:thm:_prop}}\label{proof:Quantile:thm:_prop}

Observe from the statement of Theorem \ref{Quantile:thm:_prop} that $\hat{q}_{\:J,\:p}(X)$ is defined to be the minimizer of a continuous function in the compact interval $[0,1]$, and is guaranteed to exist. However, due to the periodic nature of this function in $[0,1]$, it may fail to be unique. Indeed, Lemma \ref{Quantile:lemma:population-expression} shows that, $\hat{q}_{\:J,\:p}(X)$ almost surely $[P]$ lies in the interval $(0,1)$, and furthermore in Lemma \ref{Quantile:lemma:sample-asymptotic}, it is guaranteed to approach $\hat{q}_{\:J}(X)$, as $J \to \infty$.

We require several auxiliary results before we prove the main result in Theorem \ref{Quantile:thm:_prop}. Recall that for our given probability triple $\big(\Omega,\mathcal{F},P\big)$, we have the random variable $\tilde{x}:\Omega\to(0,1)$. We work within the setting of Theorem \ref{Quantile:thm:_prop}, such that for each $i\in I$, the $x_{\:i}$'s are i.i.d.\ observations of $\tilde{x}$. For an outcome $w\in\Omega$, we thus observe $x_{\:i}(w)$, with  $X(w,n)$ denoting the observed data $\big(x_{\:1}(w)$, $\dots$, $x_{\:n}(w)\big)$. Given the parameter $\theta\in\Theta$, with $\Theta=(0,1)$ being the parameter space, let $F^{\:(P)}(\theta)=P(\tilde{x}\leq \theta)$. If Assumption \ref{Quantile:assumption:positive-density} holds, then $F^{\:(P)}(\theta):=\int_0^\theta f^{\:(P)}(t)\,dt$ for $\theta\in(0,1)$, and observe that $F^{\:(P)}(\theta)$ has derivative $f^{\:(P)}(\theta)$.

First, we have the following three results:

\begin{lemma}\label{Quantile:lemma:inverse-continuity}
Suppose Assumption \ref{Quantile:assumption:positive-density} holds. Then: \\
(A)\ Given $\theta\in(0,1)$, let $\theta_{\:n}\in(0,1)$ for $n\in\mathbb{N}$, such that $\lim\limits_{n\to\infty}F^{\:(P)}(\theta_{\:n})\to F^{\:(P)}(\theta)$. Then $\lim\limits_{n\to\infty}\theta_{\:n}\to \theta$. \\
(B)\ There exists a unique solution $q_{\:p}^{\:(P)}\in(0,1)$ for $\theta$ in the expression $F^{\:(P)}(\theta)=p$.
\end{lemma}

We now use Lemma \ref{Quantile:lemma:inverse-continuity} to establish that $\hat{q}_{\:p}(X)$ converges almost surely to $q_{\:p}^{\:(P)}$ for any fixed $p\in(0,1)$:

\begin{lemma}\label{Quantile:lemma:exact-convergence}
Suppose Assumption \ref{Quantile:assumption:positive-density} holds. Then: \\
(A)\ Almost surely $[P]$, all elements of $X$ are distinct.\\
(B)\ For any fixed $p\in(0,1)$, $\lim_{|X|\to\infty}\hat{q}_{\:p}(X)\overset{a.s.}{=}q_{\:p}^{\:(P)}$.
\end{lemma}

Before getting into our next result, let us define the following sequence of functions with common domain $(-\pi,\pi)$:
\begin{align*}
\mathsf{1}_{\:J}(z):=\frac{1}{2}-\frac{2}{\pi}\cdot\sum_{j=1}^J\frac{\sin\big((2j-1)z\big)}{2j-1},\text{ for }z\in[-\pi,\pi]\text{ and } J\in\mathbb{N}.
\end{align*}
Observe that $\mathsf{1}_{\:J}(z)$ is the Fourier-series approximation to both of the indicator functions $\mathsf{1}(z<0)$ and $\mathsf{1}(z\leq 0)$. Recall from Example \ref{eps:eg:ind} that, we defined $\mathsf{1}_{\:J}(x,\theta)$ to be Fourier-series approximation to both $\mathsf{1}(x<\theta)$ and $\mathsf{1}(x\leq\theta)$, and $\mathsf{1}(z<0)$ can be identified as a primitive version if we note that: $\mathsf{1}_{\:J}(z)=\mathsf{1}_{\:J}(z,0)$. Now we have the following result:

\begin{lemma}\label{Quantile:lemma:indicator-uniform-convergence}
(A)\: The sequence of functions $\{ \mathsf{1}_{\:J}(z)\}_{J \in \mathbb{N}}$ converges uniformly to the limit $\mathsf{1}(z<0)$ or $\mathsf{1}(z\leq 0)$ in the interval $(-\pi,-\delta)\bigcup(\delta,\pi)$, for any $0<\delta<\pi$.

(B)\: The functions $\{ \mathsf{1}_{\:J}(z)\}_{J \in \mathbb{N}}$ are uniformly bounded in the interval $[-\pi,\pi]$.
\end{lemma}

Observe: $P(\tilde{x}\leq\theta)=\mathbb{E}_{\:P}\big(\mathsf{1}(\tilde{x}\leq\theta)\big)$; hence, by definition : $F^{\:(P)}(\theta)=\mathbb{E}_{\:P}\big(\mathsf{1}(\tilde{x}\leq\theta)\big)$. If we define the left-continuous version of the right-continuous function $F^{\:(P)}(\theta)$ as: $F^{\:(P)}(\theta^-):=P(\tilde{x}<\theta)$, the we also have: $F^{\:(P)}(\theta^-)=\mathbb{E}_{\:P}\big(\mathsf{1}(\tilde{x}<\theta)\big)$. 

Remember from section \ref{sec:Quantile:MapReduce}: $\mathsf{c}_{\:2j-1}(z)=\cos\big((2j-1)z\big)$ and $\mathsf{c}_{\:2j}(z)=\sin\big((2j-1)z\big)$, for $j\in\mathbb{N}$ and $z\in\mathbb{R}$. Then, we can write for $\tilde{x}\in(0,1)$ and $\theta\in(0,1)$:
\begin{align*}
\mathsf{1}_{\:J}(\tilde{x},\theta)=\frac{1}{2}-\frac{2}{\pi}\cdot \sum_{j=1}^J\frac{\sin\big((2j-1)\cdot (\tilde{x}-\theta)\big)}{2j-1}=\frac{1}{2}-\frac{2}{\pi}\cdot \sum_{j=1}^J\frac{\mathsf{c}_{\:2j}(\tilde{x}-\theta)}{2j-1}.
\end{align*}
Further note that: $\mathsf{c}_{\:2j}(\tilde{x}-\theta)=\mathsf{c}_{\:2j}(\tilde{x})\cdot \mathsf{c}_{\:2j-1}(\theta)-\mathsf{c}_{\:2j-1}(\tilde{x})\cdot \mathsf{c}_{\:2j}(\theta)$, hence we may write:
\begin{multline}\label{Quantile:proof:eq:b}
\mathbb{E}_{\:P}\big(\mathsf{1}_{\:J}(\tilde{x},\theta)\big) \\
=\frac{1}{2}-\frac{2}{\pi}\cdot\sum_{j=1}^J\left(\mathbb{E}_{\:P}\big(\mathsf{c}_{\:2j}(\tilde{x})\big)\cdot\frac{\mathsf{c}_{\:2j-1}(\theta)}{2j-1}-\mathbb{E}_{\:P}\big(\mathsf{c}_{\:2j-1}(\tilde{x})\big)\cdot\frac{\mathsf{c}_{\:2j}(\theta)}{2j-1}\right)\!.\!\!\!
\end{multline}

Given $J\in\mathbb{N}$ and $p\in(0,1)$, define $q_{\:J,\:p}^{\:(P)}\in(0,1)$ as a solution to $\mathbb{E}_{\:P}\big(\mathsf{1}_{\:J}(\tilde{x},\theta)\big)=p$ for $\theta$, provided such a solution exists. Let us then introduce the sequence of sets $\mathcal{Q}_{\:J}^{\:(P)}$ for $J\in\mathbb{N}$ as follows:
\begin{equation*}
    \mathcal{Q}_{\:J}^{\:(P)}:=\big\{p\colon\exists\ \theta\in\Theta\text{ such that }\mathbb{E}_{\:P}\big(\mathsf{1}_{\:J}(\tilde{x},\theta)\big)=p\big\}.
\end{equation*}
Note that if $q_{\:\alpha,\:J}^{\:(P)}$ exists, it need not be unique, since $\mathbb{E}_{\:P}\big(\mathsf{1}_{\:J}(\tilde{x},\theta)\big)$ is a weighted sum of periodic trigonometric functions in $\theta$, it is possible to have multiple solutions to $\mathbb{E}_{\:P}\big(\mathsf{1}_{\:J}(\tilde{x},\theta)\big)=p$. Therefore we shall define the sets:
\begin{equation*}
    \mathcal{G}_{\:J,\:p}^{\:(P)}:=\big\{\theta\in\Theta\colon\mathbb{E}_{\:P}\big(\mathsf{1}_{\:J}(\tilde{x},\theta)\big)=p\big\}.
\end{equation*}

Finally, for an arbitrary set $S\subset\mathbb{R}$, define $\bigtriangledown(S):=\max\big\{|a-b|\colon a\in S,b\in S\big\}$. 

We now show that under basic regulatory conditions, $q_{\:J,\:p}^{\:(P)}$ exists and converges to the limit $q_{\:p}^{\:(P)}$ when $J$ becomes large.

\begin{lemma}\label{Quantile:lemma:population-asymptotic}
Suppose Assumption \ref{Quantile:assumption:positive-density} holds. Then: \\
(A)\ $\lim_{J\to\infty}\mathbb{E}_{\:P}\:\big(\mathsf{1}_{\:J}(\tilde{x},\theta)\big)\to F^{\:(P)}(\theta)$ uniformly in $\theta\in(0,1)$. \\
(B)\ $\lim_{J\rightarrow\infty}\mathcal{Q}_{\:J}^{\:(P)}=(0,1)$. \\
(C)\ $\lim_{J\to\infty}q_{\:J,\:p}^{\:(P)}\to q_{\:p}^{\:(P)}$ and $\lim_{J\to\infty}\bigtriangledown\big(\mathcal{G}_{\:J,\:p}^{\:(P)}\big)\to 0$ for $p\in(0,1)$.
\end{lemma}

Similar to the argument preceding \citet[Definition 1.1]{AOS_1997_Koltchinskii}, a $p$th sample quantile $\hat{q}_{\:p}(X)$ can be identified as an optimizer to the optimization problem \citet{Book_2005_Koenker}: $\hat{q}_{\:p}(X)=\operatornamewithlimits{argmin}_{\theta\in(0,1)} \mathsf{G}_{\:p}(X,\theta)$, where $\mathsf{G}_{\:p}(X,\theta)=\sum_{i\in I}\rho_{\:p}(x_{\:i}-\theta)$ and $\rho_{\:p}(z)=\big(\nicefrac{1}{2}\big)\cdot|z|+\big(p-\nicefrac{1}{2}\big)\cdot z$ for $z\in\mathbb{R}$. Now from Example \ref{eps:eg:abs}, we know that when the data is scaled to $(0,1)$, the Fourier series of $|x_{\:i}-\theta|$ provides a convergent FAS approximation, and so we have the following expression for $x_{\:i}\in(0,1)$ and $\theta\in(0,1)$:
\begin{align}\label{Quantile:proof:eq:e}
\rho_{\:p}(x_{\:i}-\theta)&=\frac{1}{2}\cdot|x_{\:i}-\theta|+\big(p-\frac{1}{2}\big)\cdot(x_{\:i}-\theta) \\
&=\frac{\pi}{4}-\frac{2}{\pi}\cdot \sum_{j=1}^{\infty}\frac{\cos\big((2j-1)\cdot (x_{\:i}-\theta)\big)}{(2j-1)^2}+\big(p-\frac{1}{2}\big)\cdot(x_{\:i}-\theta). \nonumber
\end{align}
We denote the $J$th partial sum in the expression of $\rho_{\:p}(x_{\:i}-\theta)$ as $\rho_{\:J,\:p}(x_{\:i}-\theta)$ .From the statement of Theorem \ref{Quantile:thm:_prop} we have $\hat{q}_{\:J,\:p}(X)=\operatornamewithlimits{argmin}_{\theta\in(0,1)} \mathsf{G}_{\:J,\:p}(X,\theta)$, where $\mathsf{G}_{\:J,\:p}(X,\theta)=\sum_{i\in I}\rho_{\:J,\:p}(x_{\:i}-\theta)$. Recall that in section \ref{sec:Quantile:MapReduce}, we introduced the standardized function $\bar{\mathsf{G}}_{\:J,\:p}(X,\theta)=\nicefrac{1}{|X|}\cdot\mathsf{G}_{\:J,\:p}(X,\theta)$, for which it still holds that $\hat{q}_{\:J,\:p}(X)=\operatornamewithlimits{argmin}_{\theta\in[0,1]} \bar{\mathsf{G}}_{\:J,\:p}(X,\theta)$. We then have the following:

\begin{lemma}\label{Quantile:lemma:population-expression}
We have that \\
(A)\ \quad $\bar{\mathsf{G}}_{\:J,\:p}(X,\theta) =\frac{\pi}{4}-\frac{2}{\pi}\cdot \sum_{j=1}^J\frac{\bar{\mathsf{C}}_{\:2j-1}(X)\cdot \mathsf{c}_{\:2j-1}(\theta)+\bar{\mathsf{C}}_{\:2j}(X)\cdot \mathsf{c}_{\:2j}(\theta)}{(2j-1)^{2}}+\big(p-\frac{1}{2}\big)\cdot\big(\bar{X}-\theta\big)$ and \\
 \quad $\frac{\partial}{\partial\theta}\big(\bar{\mathsf{G}}_{\:J,\:p}(X,\theta)\big)=\hat{F}_{\:J}(X,\theta)-p\text{, where }\hat{F}_{\:J}(X,\theta)=\frac{1}{|X|}\cdot\big(\sum_{i\in I}\mathsf{1}_{\:J}(x_{\:i},\theta)\big)$. \\
 (B)\ \quad$\frac{1}{|X|}\cdot\big(\mathsf{G}_{\:p}(X,\theta)-\mathsf{G}_{\:J,\:p}(X,\theta)\big)=O(J^{-1})$ for arbitrary $X$ and $p\in(0,1)$. \\
 (C)\ \quad $\hat{q}_{\:J,\:p}(X)$ lies in the open interval $(0,1)$ for large $J$.
\end{lemma}

Part C of Lemma \ref{Quantile:lemma:population-expression} asserts that $\hat{q}_{\:J,\:p}(X)$ is a solution to the equation $\hat{F}_{\:J}(X,\theta)=p$ for $\theta\in(0,1)$ for sufficiently large $J$. We now characterize the different solution sets that are possible when considering the expression $\hat{F}_{\:J}(X,\theta)=p$. Fix $J\in\mathbb{N}$ and $p\in(0,1)$, and for a given data set $X$, we expand the definition of $\hat{q}_{\:J,\:p}(X)$ as a solution to the equation $\hat{F}_{\:J}(X,\theta)=p$ for $\theta\in(0,1)$, whenever such a solution exists. For $J\in\mathbb{N}$, define the set
\begin{align*}
\hat{\mathcal{Q}}_{\:J}(X):=\big\{p\colon\exists\ \theta\in(0,1)\text{ such that }\hat{F}_{\:J}(X,\theta)=p\big\},
\end{align*}
and for $p\in(0,1)$, $J\in\mathbb{N}$, $n\in\mathbb{N}$ and $w\in\Omega$ define the set:
\begin{align*}
\hat{\mathcal{G}}_{\:J,\:p,\:n}(w)=\Big\{\theta\in(0,1)\colon\hat{F}_{\:J}\big(X(w,n),\theta\big)=p\Big\}.
\end{align*}
Observe that for small $n$, it is possible that $\hat{\mathcal{G}}_{\:J,\:p,\:n}(w)=\emptyset$, meaning there is no solution to $\hat{F}_{\:J}\big(X(w,n),\theta\big)=\alpha$ for $\theta\in(0,1)$. It is also possible that there are multiple solutions to $\hat{F}_{\:J}\big(X(w,n),\theta\big)=p$. However, $\hat{\mathcal{G}}_{\:J,\:p,\:n}(w)$ is always a finite set, as the equation $\hat{F}_{\:J}\big(X(w,n),\theta\big)=p$ can't have infinitely many solutions for $\theta\in(0,1)$. From part A of Lemma \ref{Quantile:lemma:population-asymptotic}, we have 
\begin{align*}
\lim_{J\to\infty}\mathbb{E}_{\:P}\big(\mathsf{1}_{\:J}(\tilde{x},\theta)\big)=F^{\:(P)}(\theta)\text{ for }\theta\in(0,1).
\end{align*}
In perspective of \ref{Quantile:proof:eq:b}, we may then write
\begin{align}\label{Quantile:proof:eq:c}
\frac{1}{2}-\frac{2}{\pi}\cdot\sum_{j=1}^{\infty}\left(\mathbb{E}_{\:P}\big(\mathsf{c}_{\:2j}(\tilde{x})\big)\cdot\frac{\mathsf{c}_{\:2j-1}(\theta)}{2j-1}-\mathbb{E}_{\:P}\big(\mathsf{c}_{\:2j-1}(\tilde{x})\big)\cdot\frac{\mathsf{c}_{\:2j}(\theta)}{2j-1}\right)=F^{\:(P)}(\theta).
\end{align}

We also have that ${\mathsf{c}_{\:2j-1}}'(\theta)=-(2j-1)\cdot \mathsf{c}_{\:2j}(\theta)$ and
${\mathsf{c}_{\:2j}}'(\theta)=(2j-1)\cdot \mathsf{c}_{\:2j-1}(\theta)$, and so the derivative of the $j$th term of the left-hand side of equation \ref{Quantile:proof:eq:c} is $\nicefrac{2}{\pi}\Big(\mathbb{E}_{\:P}\big(\mathsf{c}_{\:2j-1}(\tilde{x})\big)\mathsf{c}_{\:2j-1}(\theta)+\mathbb{E}_{\:P}\big(\mathsf{c}_{\:2j}(\tilde{x})\big)\mathsf{c}_{\:2j}(\theta)\Big)$. The
corresponding partial in $J$ is precisely $\zeta_{\:J}^{\:(P)}(\theta)$ as defined in Assumption \ref{Quantile:assumption:uniform-convergence}:
\begin{align*}
\zeta_{\:J}^{\:(P)}(\theta)&=\mathbb{E}_{\:P}\left(\frac{1}{\pi}\cdot \frac{\sin\left(2J(\tilde{x}-\theta)\right)}{\sin(\tilde{x}-\theta)}\right)\\
&=\frac{2}{\pi}\sum_{j=1}^J\Big(\mathbb{E}_{\:P}\big(\mathsf{c}_{\:2j-1}(\tilde{x})\big)\mathsf{c}_{\:2j-1}(\theta)+\mathbb{E}_{\:P}\big(\mathsf{c}_{\:2j}(\tilde{x})\big)\mathsf{c}_{\:2j}(\theta)\Big).
\end{align*}
From \citet[Theorem 9.13]{Book_1957_Apostol}, we know if $\zeta_{\:J}^{\:(P)}(\theta)$ converges uniformly to a limit, then that limit will be the derivative of the right-hand side of equation \ref{Quantile:proof:eq:c}, which is $f^{\:(P)}(\theta)$. Therefore, for $J\in\mathbb{N}$, define $\rho_{\:J}^{\:(P)} \in \mathbb{R} \cup \infty $ as follows:
\begin{align*}
\rho_{\:J}^{\:(P)}:=\sup_{\theta\in(0,1)}\big|\:f^{\:(P)}(\theta)-\zeta_{\:J}^{\:(P)}(\theta)\:\big|.
\end{align*}
Below, to prove Theorem \ref{Quantile:thm:_prop}, we shall use Assumption \ref{Quantile:assumption:uniform-convergence} to force $\rho_{\:J}^{\:(P)}$ to zero in $J$. For the moment, however, we use $\rho_{\:J}^{\:(P)}$ simply to characterize limit points with respect to $\hat{\mathcal{G}}_{\:J,\:p}(w)$, defined for $p\in(0,1)$, $J\in\mathbb{N}$ and $w\in\Omega$ as the set
\begin{align*}
\hat{\mathcal{G}}_{\:J,\:p}(w):=\bigcup_{n\in\mathbb{N}}\hat{\mathcal{G}}_{\:J,\:p,\:n}(w).
\end{align*}

\begin{lemma}\label{Quantile:lemma:sample-asymptotic}
Suppose Assumption \ref{Quantile:assumption:positive-density} holds. Then: \\
(A)\ $\lim\limits_{|X|\to\infty}\sup\limits_{\theta\in\Theta}\big|\:\hat{F}_{\:J}(X,\theta)-\mathbb{E}_{\:P}\big(\mathsf{1}_{\:J}(\tilde{x},\theta)\big)\:\big|\overset{a.s.}{=} 0$. \\
(B)\ $\lim\limits_{|X|\rightarrow\infty}\lim\limits_{J\rightarrow\infty}\hat{\mathcal{Q}}_{\:J}(X)\overset{a.s.}{=}(0,1)$. \\
(C)\ Recalling that $q_{\:J,\:p}^{\:(P)}$ is any solution to $\mathbb{E}_{\:P}\big(\mathsf{1}_{\:J}(\tilde{x},\theta)\big)=p$, for $J$ sufficiently large, there exists a null set $N$ such that for $w\in N'$, if $\hat{q}_{\:J,\:p}\big(X(w,\cdot)\big)$ is a limit point of the set $\hat{\mathcal{G}}_{\:J,\:p}(w)$, then $\Big|\:F^{\:(P)}\Big(\hat{q}_{\:J,\:p}\big(X(w,\cdot)\big)\Big)-F^{\:(P)}\left(q_{\:J,\:p}^{\:(P)}\right)\:\Big|\leq\rho^{\:(P)}_{\:J}$.
\end{lemma}

\begin{proof}[Proof of Theorem \ref{Quantile:thm:_prop}]
Suppose Assumption \ref{Quantile:assumption:positive-density} and Assumption \ref{Quantile:assumption:uniform-convergence} hold. First of all, since Assumption \ref{Quantile:assumption:positive-density} is true, from part C of Lemma \ref{Quantile:lemma:population-asymptotic}, we know that $\lim\limits_{J\to\infty}q_{\:J,\:p}^{\:(P)}=q_{\:\alpha}^{\:(P)}$, and hence that $\lim\limits_{J\to\infty}F\big(q_{\:J,\:p}^{\:(P)}\big)=F\big(q_{\:p}^{\:(P)}\big)$. Secondly, since Assumption \ref{Quantile:assumption:uniform-convergence} is true, $\zeta_{\:J}^{\:(P)}(\theta)$ converges uniformly to $f^{\:(P)}(\theta)$ for $\theta\in(0,1)$, and so, for the set of real numbers $\{\rho_{\:J}^{\:(P)}\}_{J\in\mathbb{N}}$, we have $\lim\limits_{J\to\infty}\rho_{\:J}^{\:(P)}\to 0$. Now, fix $\epsilon>0$, and observe we can choose a $J_{\:1}$ large enough such that for $J>J_{\:1}$, both $\big|\:F\big(q_{\:J,\:p}^{\:(P)}\big)-F\big(q_{\:p}^{\:(P)}\big)\:\big|<\nicefrac{\epsilon}{2}$ and $\rho_{\:J}^{\:(P)}<\nicefrac{\epsilon}{2}$. 

Observe that $\limsup\limits_{n\to\infty}\hat{q}_{\:J,\:p}\big(X(w,n)\big)$ is a limit point of the set $\hat{\mathcal{G}}_{\:J,\:p}(w)$. Since Assumption \ref{Quantile:assumption:uniform-convergence} holds, from part C of Lemma \ref{Quantile:lemma:sample-asymptotic}, we know that we can choose a $J_{\:2}$ large enough such that for $J>J_{\:2}$, there exists some null set $N$ such that for $w\in N'$, we have the inequality:
\begin{align*}
\quad\left|\:F^{\:(P)}\Big(\limsup_{n\to\infty}\hat{q}_{\:J,\:p}\big(X(w,n)\big)\Big)-F^{\:(P)}\big(q_{\:J,\:p}^{\:(P)}\big)\:\right|\leq\rho^{\:(P)}_{\:J}.
\end{align*}
Then we have for any $J>\operatorname{max}(J_{\:1}, J_{\:2})$ that
\begin{align*}
&\hspace{-2em}\left|\:F^{\:(P)}\Big(\limsup_{n\to\infty}\hat{q}_{\:J,\:p}\big(X(w,n)\big)\Big)-F^{\:(P)}\big(q_{\:p}^{\:(P)}\big)\:\right| \\
&\leq\left|\:F^{\:(P)}\Big(\limsup_{n\to\infty}\hat{q}_{\:J,\:p}\big(X(w,n)\big)\Big)-F^{\:(P)}\big(q_{\:J,\:p}^{\:(P)}\big)\:\right| +\left|\:F\big(q_{\:J,\:p}^{\:(P)}\big)-F\big(q_{\:p}^{\:(P)}\big)\:\right| \\
&< \rho^{\:(P)}_{\:J} + \frac{\epsilon}{2},
\end{align*}
which is less than $\epsilon$ by the result of the preceding paragraph. Since $\epsilon$ is arbitrary, we conclude
\begin{align*}
\lim_{J\to\infty}F^{\:(P)}\Big(\limsup_{n\to\infty}\hat{q}_{\:J,\:p}\big(X(w,n)\big)\Big)=F^{\:(P)}\big(q_{\:p}^{\:(P)}(x)\big).
\end{align*}

Next, under Assumption \ref{Quantile:assumption:positive-density}, part A of Lemma \ref{Quantile:lemma:inverse-continuity} implies that
\begin{align*}
\lim_{J\to\infty}\limsup_{n\to\infty}\hat{q}_{\:J,\:p}\big(X(w,n)\big)=q_{\:p}^{\:(P)}.
\end{align*}
Since this is true for any $w\in N'$ relative to the null set $N$, we have that
\begin{align*}
\lim_{J\to\infty}\limsup_{|X|\to\infty}\hat{q}_{\:J,\:p}(X)\overset{a.s.}{=}q_{\:p}^{\:(P)}.
\end{align*}

Under Assumption \ref{Quantile:assumption:positive-density}, from part B of Lemma \ref{Quantile:lemma:exact-convergence}, we have the result that
\begin{align*}
\lim_{|X|\to\infty}\hat{q}_{\:p}(X)\overset{a.s.}{=}q_{\:p}^{\:(P)}.
\end{align*}
and so we must have that
\begin{align*}
\lim_{J\to\infty}\limsup_{|X|\to\infty}\hat{q}_{\:J,\:p}(X)\overset{a.s.}{=}\lim_{|X|\to\infty}\hat{q}_{\:p}(X).
\end{align*}

The result $\lim\limits_{J\to\infty}\liminf\limits_{|X|\to\infty}\hat{q}_{\:J,\:p}(X)\overset{a.s.}{=}\lim\limits_{|X|\to\infty}\hat{q}_{\:p}(X)$ follows analogously
\end{proof}

\section{Outline of Proof of Theorem \ref{Lowess:thm:prop}}\label{overview:proof:Lowess:thm:_prop}

Here we outline the proof of Theorem \ref{Lowess:thm:prop}, whose structure parallels that of the proof of Theorem \ref{Quantile:thm:_prop}. It and several auxiliary lemmata appear in full in the Supplementary Material.

Recall that the local regression setting of Theorem \ref{Lowess:thm:prop} assumes a probability triple $(\Omega,\mathcal{F},P)$ giving rise to a set of i.i.d.\ realizations $X = \{x_{\:i}, i \in I\}$ for the explanatory variable $\tilde{x}\colon\Omega\to(0,1)$. Let Assumption \ref{Lowess:assumption:positive-density} be in force throughout. 

Following the arguments of Lemma \ref{Quantile:lemma:inverse-continuity}, we show in Lemma \ref{Lowess:lemma:inverse_continuity} of the Supplementary Material that the inverse function of $F_{\:x}^{\:(P)}(h) := P(x-h\leq \tilde{x}\leq x+h)$ is continuous and that for fixed $\alpha\in(0,1)$, there exists a unique solution $h_{\:\alpha,\:x}^{\:(P)}\in(0,1)$ for $h$ to $F_{\:x}^{\:(P)}(h)=\alpha$. Recall its sample counterpart $\hat{h}_{\:\alpha,\;x}(X)$, defined in section \ref{sec:EPF-Lowess} with respect to $\hat{F}_{\:x}(X,h^{-}) = \big(\nicefrac{1}{|X|}\big)\cdot\sum_{i\in I}\mathsf{1}(x-h< x_{\:i}< x+h)$ and $\hat{F}_{\:x}(X,h) = \big(\nicefrac{1}{|X|}\big)\cdot\sum_{i\in I}\mathsf{1}(x-h\leq x_{\:i}\leq x+h)$ as
\begin{align*}
h=\hat{h}_{\:\alpha,\:x}(X) \Leftrightarrow h\in(0,1) \text{ and } \hat{F}_{\:x}(X,h^{-}) \leq \alpha \leq \hat{F}_{\:x}(X,h).
\end{align*}
Paralleling Lemma \ref{Quantile:lemma:exact-convergence}, we show in Lemma \ref{Lowess:lemma:exact_convergence} of the Supplementary Material first that $\hat{h}_{\:\alpha,\;x}(X)$ almost surely $[P]$ exists, and then that $\hat{h}_{\:\alpha,\;x}(X)$ almost surely $[P]$ converges to $h_{\:\alpha,\:x}^{\:(P)}$ as the sample size $|X|$ grows large.

Next, consider our $J$-term Fourier approximation $\hat{F}_{\:J,\:x}(X,h)=\big(\nicefrac{1}{|X|}\big)\cdot\sum_{i\in I}\mathsf{1}_{\:J,\:x}(x_{\:i},h)$. Recall from section \ref{sec:Quantile:MapReduce}
the notation $\mathsf{c}_{\:2j-1}(z)=\cos\big((2j-1)z\big)$ and $\mathsf{c}_{\:2j}(z)=\sin\big((2j-1)z\big)$, as well as $\mathsf{C}_{\:j}(X) =\sum_{i\in I}\mathsf{c}_{\:j}(x_{\:i})$ and $\bar{\mathsf{C}}_{\:j}(X)=\nicefrac{1}{|X|}\cdot\mathsf{C}_{\:j}(X)$. Lemma \ref{Lowess:lemma:delta_expression} of the Supplementary Material in turn yields the expressions
\begin{align}\label{Lowess:proof:eq:a}
\ \mathsf{1}_{\:J,\:x}(\tilde{x},h)=\frac{4}{\pi}\cdot \sum_{j=1}^J\big(\mathsf{c}_{\:2j-1}(\tilde{x}-x)\big)\frac{\mathsf{c}_{\:2j}(h)}{2j-1},
\end{align}
\begin{align}\label{Lowess:proof:eq:b}
\ \hat{F}_{\:J,\:x}(X,h)=\frac{4}{\pi}\cdot \sum_{j=1}^J\big(\bar{\mathsf{C}}_{\:2j-1}(X)\cdot\mathsf{c}_{\:2j-1}(x)+\bar{\mathsf{C}}_{\:2j}(X)\cdot\mathsf{c}_{\:2j}(x)\big)\cdot\frac{\mathsf{c}_{\:2j}(h)}{2j-1}.
\end{align}
Paralleling Lemma \ref{Quantile:lemma:population-asymptotic}, we show in Lemma \ref{Lowess:lemma:population_asymptotic} of the Supplementary Material via equation \ref{Lowess:proof:eq:a} that $\mathbb{E}_{\:P}\big(\mathsf{1}_{\:J,\:x}(\tilde{x},h)\big)$ converges uniformly to $F_{\:x}^{\:(P)}(h)$ as $J$ grows large. We also show that, given any $\alpha\in(0,1)$, for sufficiently large $J$, there exists at least one solution $h_{\:J,\:\alpha,\:x}^{\:(P)}$ to the equation $\mathbb{E}_{\:P}\big(\mathsf{1}_{\:J,\:x}(\tilde{x},h)\big)=\alpha$ for $h\in(0,1)$, and as $J\to\infty$, then $h_{\:J,\:\alpha,\:x}^{\:(P)}\to h_{\:\alpha,\:x}^{\:(P)}$.

Having established conditions under which $\hat{h}_{\:\alpha,\;x}(X) \to h_{\:\alpha,\:x}^{\:(P)}$ as $|X|$ grows large and $h_{\:J,\:\alpha,\:x}^{\:(P)}\to h_{\:\alpha,\:x}^{\:(P)}$ as $J$ grows large, we are now ready to treat the corresponding sample asymptotics jointly in $|X|$ and $J$. Paralleling Lemma \ref{Quantile:lemma:sample-asymptotic} and with the aid of equations \ref{Lowess:proof:eq:a} and \ref{Lowess:proof:eq:b}, we show in Lemma \ref{Lowess:lemma:sample_asymptotic} of the Supplementary Material that for any fixed $J$, $\big|\hat{F}_{\:J,\:x}(X,h)-\mathbb{E}_{\:P}\big(\mathsf{1}_{\:J,\:x}(\tilde{x},h)\big)\big|$ converges almost surely $[P]$ to $0$, uniformly for $\{x,h\}\in(0,1)$ as $|X|$ grows large. We also use equation \ref{Lowess:proof:eq:b} to show that, given any $\alpha\in(0,1)$, for sufficiently large $|X|$ and $J$, almost surely $[P]$ there exists at least one solution $\hat{h}_{\:J,\:\alpha,\:x}(X)$ to the equation $\hat{F}_{\:J,\:x}(X,h)=\alpha$ for $h\in(0,1)$. 

Now recall, with respect to our underlying probability triple $(\Omega,\mathcal{F},P)$, that for an outcome $w\in\Omega$, we observe $x_{\:i}(w)$, with  $X(w,n)$ denoting the observed data $\big(x_{\:1}(w)$, $\dots$, $x_{\:n}(w)\big)$. In the final portion of Lemma \ref{Lowess:lemma:sample_asymptotic} of the Supplementary Material, we establish that for $w\in\Omega$, if $\hat{h}_{\:J,\:\alpha,\:x}\big(X(w,\cdot)\big)$ is a limit point of the set $\{\hat{h}_{\:J,\:\alpha,\:x}\big(X(w,n)\big):n\in\mathbb{N}\}$, then the difference $\Big|F_{\:x}^{\:(P)}\big(h_{\:J,\:\alpha,\:x}^{\:(P)}\big)-F_{\:x}^{\:(P)}\Big(\hat{h}_{\:J,\:\alpha,\:x}\big(X(w,\cdot)\big)\Big)\Big|$ is bounded for a given $J$ and $x$. In this setting $\{\hat{h}_{\:J,\:\alpha,\:x}\big(X(w,n)\big):n\in\mathbb{N}\}$ must be treated as a tail set, such that at least one element $\hat{h}_{\:J,\:\alpha,\:x}\big(X(w,n)\big)$ exists when $|X|$ and  $J$ are sufficiently large.

By continuity of the inverse function of $F_{\:x}^{\:(P)}(h)$, it follows that squeezing this bound to zero as $J$ grows large will lead to the result of Theorem \ref{Lowess:thm:prop}. To do so, we first use equation \ref{Lowess:proof:eq:a} to establish that differentiating $\mathbb{E}_{\:P}\big(\mathsf{1}_{\:J,\:x}(\tilde{x},h)\big)$ with respect to $h$ yields $\eta_{\:J,\:x}^{\:(P)}(h)$ as defined in Assumption \ref{Lowess:assumption:uniform-convergence}. Recalling that $\mathbb{E}_{\:P}\big(\mathsf{1}_{\:J,\:x}(\tilde{x},h)\big) \to F_{\:x}^{\:(P)}(h)$ uniformly in $h$, it follows from \citet[Theorem 9.13]{Book_1957_Apostol} that if $\eta_{\:J,\:x}^{\:(P)}(h)$ converges uniformly to any function as $J$ grows large, then this function must be the derivative $f_{\:x}^{\:(P)}(h)$ of $F_{\:x}^{\:(P)}(h)$. The distance between $\eta_{\:J,\:x}^{\:(P)}(h)$ and $f_{\:x}^{\:(P)}(h)$ is in fact precisely the bound that we need to control as $J$ grows large, thereby motivating Assumption \ref{Lowess:assumption:uniform-convergence} and completing the proof outline.

\end{appendices}

\newpage

\begin{center}
{\LARGE Supplementary Material}
\end{center}
\vspace{4em}%

\section{Proof of the lemmata used in Appendix \ref{proof:eps:prop:SOT-FAS-WEP}}

\subsection{Proof of Lemma \ref{Quantile:lemma:inverse-continuity}}
\begin{proof}
(A)\ We have $F^{\:(P)}(\theta)=\int_0^\theta f^{\:(P)}(t)\,dt$, since by Assumption \ref{Quantile:assumption:positive-density}, $f>0\:[P]$ in $(0,1)$, and so we must have $F^{\:(P)}$ continuous and strictly monotone in $(0,1)$. Then part A follows from continuity of the inverse of a continuous function.

(B)\ Since $F^{\:(P)}(\theta)=P(x\leq\theta)$ and $0<\tilde{x}<1$, we have $F^{\:(P)}(0)=0$ and $F^{\:(P)}(1)=1$. Given $p\in(0,1)$, since $F^{\:(P)}$ is continuous, then by the intermediate-value theorem there exists at least one value $q^{\:(P)}_{\:p}\in(0,1)$, such that $F\big(q^{\:(P)}_{\:p}\big)=p$. To show unicity of $q^{\:(P)}_{\:p}$, assume to the contrary that there exist two distinct values $0<q^{\:(P)}_{\:p,\:1}<q^{\:(P)}_{\:p,\:2}<1$ satisfying $F^{\:(P)}(\theta)=p$ for $\theta$. Then we have
\begin{align*}
F\left(q^{\:(P)}_{\:p,\:1}\right)-F\left(q^{\:(P)}_{\:p,\:2}\right)=p-p=0\:\Leftrightarrow\:\int_{q^{\:(P)}_{\:p,\:1}}^{q^{\:(P)}_{\:p,\:2}}f(t)\,dt=0,
\end{align*}
which contradicts the assumption that $f^{\:(P)}>0\:[P]$.
\end{proof}

\subsection{Proof of Lemma \ref{Quantile:lemma:exact-convergence}}
\begin{proof}
(A)\ For $\{i,i'\}\in\mathbb{N}$ and $i\neq i'$, let us define sets:
\begin{align*}
\psi_{\:i,\:i'}:=\{w\in\Omega\colon x_{\:i}(w)=x_{\:i'}(w)\};\:\psi:=\bigcup_{1\leq i<i'<\infty}\psi_{\:i,\:i'}.
\end{align*}
Note that for $\{i,i'\}\in\mathbb{N}$ and $i\neq i'$, we have $P(\psi_{\:i,\:i'})=0$, by Assumption \ref{Quantile:assumption:positive-density}, $P$ admits a density and hence the random variable $z=x_{\:i}-x_{\:i'}$ cannot have mass at $0$. Clearly $\psi\subset\Omega$ and we have:
\begin{align*}
P(\psi)=P\bigg(\bigcup_{1\leq i<i'<\infty}\psi_{\:i,\:i'}\bigg)\leq\sum_{1\leq i<i'<\infty}P(\psi_{\:i,\:i'})=0.
\end{align*}
So, $\psi$ is a null set$\:[P]$. Now, if $w\notin\psi$ and $X(w,n)=\big(x_{\:1}(w),\dots,x_{\:n}(w)\big)$, then $X$ has all distinct elements for any $n\in\mathbb{N}$. Thus we conclude the result.

(B)\ Let $\epsilon>0$ be given. By the Glivenko--Cantelli theorem, we have the result that  $\sup_{\theta\in(0,1)}\big|\:\hat{F}(X,\theta)-F^{\:(P)}(\theta)\:\big|\overset{a.s.}{=}0$. In particular for $\theta=\hat{q}_{\:P}(X)$, we also have $\big|\:\hat{F}\big(X,\hat{q}_{\:P}(X)\big)-F^{\:(P)}\big(\hat{q}_{\:P}(X)\big)\:\big|\overset{a.s.}{=}0$. Hence, there exists a null set $N_{\:1}\subset\Omega$, such that for each $w\in{N_{\:1}}'$, there exists an $n_{\:w,\:1}$ such that for $n>n_{\:w,\:1}$ we have that $\big|\:\hat{F}\big(X(w,n),\hat{q}_{\:P}\big(X(w,n)\big)\big)-F^{\:(P)}\big(\hat{q}_{\:P}\big(X(w,n)\big)\big)\:\big|<\nicefrac{\epsilon}{2}$.

Now, from the definition of a quantile, we know that $\hat{q}_{\:P}(X)$ satisfies the expression $\hat{F}\big(X,\hat{q}_{\:P}(X)^-\big)\leq p\leq \hat{F}\big(X,\hat{q}_{\:P}(X)\big)$.
Observe that if the elements of $X$ are distinct, then for any $\theta\in(0,1)$, we have $\big|\hat{F}(X,\theta^-)-\hat{F}(X,\theta)\big|\leq\nicefrac{1}{|X|}$. We can take $\theta=\hat{q}_{\:P}(X)$ and consequently, if elements of $X$ are distinct, we have
$\big|\hat{F}\big(X,\hat{q}_{\:P}(X)\big)-p\big|\leq\nicefrac{1}{|X|}$. So, from part A we can conclude that there exists a null set $N_{\:2}\subset\Omega$, such that for each $w\in{N_{\:2}}'$ and any $n\in\mathbb{N}$, we have $\big|\hat{F}\big(X(w,n),\hat{q}_{\:P}\big(X(w,n)\big)\big)-p\big|\leq\nicefrac{1}{n}$. Clearly $N:=N_{\:1}\cup N_{\:2}$ is such a null set, and so if $w\in{N}'$, then by taking $n_{\:w}:=\operatorname{max}(n_{\:w,\:1},\nicefrac{2}{\epsilon})$ we have:
\begin{align*}
&\hspace{-2em}\Big|\:F^{\:(P)}\Big(\hat{q}_{\:P}\big(X(w,n)\big)\Big)-p\:\Big| \\
&\leq\Big|\:\hat{F}\Big(X(w,n),\hat{q}_{\:P}\big(X(w,n)\big)\Big)-F^{\:(P)}\Big(\hat{q}_{\:P}\big(X(w,n)\big)\Big)\:\Big| \\
&\hspace{2em}+\Big|\:\hat{F}\Big(X(w,n),\hat{q}_{\:P}\big(X(w,n)\big)\Big)-p\:\Big| \\
&<\nicefrac{\epsilon}{2}+\nicefrac{1}{n},
\end{align*}
with this sum being less than or equal to $\epsilon$ for $n$ sufficiently large. Now, from part B of Lemma \ref{Quantile:lemma:inverse-continuity}, we have a unique solution $q_{\:p}^{\:(P)}\in(0,1)$, to $F^{\:(P)}(\theta)=p$ for $\theta$. Then from the above inequality, for the null set $N\in\Omega$, if $w\in{N}'$ and $n>n_{\:w}$, we have
$\big|\:F^{\:(P)}\big(\hat{q}_{\:P}\big(X(w,n)\big)\big)-F^{\:(P)}\big(q_{\:p}^{\:(P)}\big)\big|<\epsilon$. Since $\epsilon$ is arbitrary, we must have $\lim\limits_{n\to\infty}F^{\:(P)}\big(\hat{q}_{\:P}\big(X(w,n)\big)\big)=F^{\:(P)}\big(q_{\:p}^{\:(P)}\big)$. From part A of Lemma \ref{Quantile:lemma:inverse-continuity}, we conclude 
$\lim\limits_{n\to\infty}\hat{q}_{\:P}\big(X(w,n)\big)=q_{\:p}^{\:(P)}$. This is true for any $w\in{N}'$, since $N\in\Omega$ is a null set, and hence we have proved the statement $\lim\limits_{|X|\to\infty}\hat{q}_{\:p}(X)\overset{a.s.}{=}q_{\:p}^{\:(P)}$.
\end{proof}

\subsection{Proof of Lemma \ref{Quantile:lemma:indicator-uniform-convergence}}
\begin{proof}
(A)~ First observe, for $z\in(-\pi,-\delta)\:\bigcup\:(\delta,\pi)$, we have: $\mathsf{1}(z<0)=\mathsf{1}(z\leq 0)$. Consider the sum:
\begin{align*}
E_{\:J}(z):=\sum_{j=1}^J\sin\big((2j-1)\cdot z\big)=\frac{\big(1-\cos(2Jz)\big)}{2\sin(z)}.
\end{align*}
For $z\in(-\pi,-\delta)\:\bigcup\:(\delta,\pi)$ and $J\in\mathbb{N}$, we have:
\begin{align*}
|E_{\:J}(z)|=\frac{|1-\cos(2Jz)|}{2|\sin(z)|}\leq\frac{1}{|\sin(z)|}\leq\frac{1}{|\sin(\delta)|}
\end{align*}
Now consider the Cauchy tail sum $e_{\:J,\:J'}(z)=\mathsf{1}_{\:J'}(z)-\mathsf{1}_{\:J}(z)$ for $J<J'$.
We have:
\begin{align*}
e_{\:J,\:J'}(z) &=\frac{2}{\pi}\cdot \sum_{j=J+1}^{J'}\frac{\sin\big((2j-1)\cdot z\big)}{2j-1}\\
&=\frac{2}{\pi}\cdot \sum_{j=J+1}^{J'}\frac{E_{\:j}(z)-E_{\:j-1}(z)}{2j-1} \\
&=\frac{2}{\pi}\cdot \left(\sum_{j=J+1}^{J'}\frac{E_{\:j}(z)}{2j-1}-\sum_{j=J+1}^{J'}\frac{E_{\:j}(z)}{2j+1}-\frac{E_{\:J}(z)}{2J+1}+\frac{E_{\:J'}(z)}{2J'+1}\right) \\
&=\frac{2}{\pi}\cdot \left(\sum_{j=J+1}^{J'}E_{\:j}(z)\cdot \Big(\frac{1}{2j-1}-\frac{1}{2j+1}\Big)-\frac{E_{\:J}(z)}{2J+1}+\frac{E_{\:J'}(z)}{2J'+1}\right).
\end{align*}
Thus we may obtain the following upper bound:
\begin{align*}
|e_{\:J,\:J'}(z)| &\leq\frac{2}{\pi|\sin(\delta)|}\cdot \left(\sum_{j=J+1}^{J'}\left(\frac{1}{2j-1}-\frac{1}{2j+1}\right)+\frac{1}{2J+1}+\frac{1}{2J'+1}\right) \\
&=\frac{2}{\pi|\sin(\delta)|}\left(\frac{1}{2J+1}-\frac{1}{2J'+1}+\frac{1}{2J+1}+\frac{1}{2J'+1}\right) \\
&=\frac{4}{(2J+1)\pi|\sin(\delta)|}.
\end{align*}
Given $\epsilon>0$, it follows that $|e_{\:J,\:J'}(z)|$ is in turn upper-bounded by $\epsilon$ eventually in $J$ and $J'>J$, for all $z\in(-\pi,-\delta)\:\cup\:(\delta,\pi)$. Thus by the Cauchy criterion, $\mathsf{1}_{\:J}(z)$ uniformly converges to its limit $\mathsf{1}(z<0)$ or $\mathsf{1}(z\leq 0)$.

(B) Since $\mathsf{1}_{\:J}(z)$ is an odd function, we have $|\mathsf{1}_{\:J}(-z)|=|\mathsf{1}_{\:J}(z)|$, and so it is sufficient to show that $\mathsf{1}_{\:J}(z)$ is uniformly bounded in the interval $(0,\pi)$. Now, for $z\in(0,\pi)$, we have the inequality:
\begin{align*}
\frac{2}{\pi}\leq\frac{\sin(z)}{z}\leq 1.
\end{align*}
The left-hand inequality comes from the fact that the function $f(z)=\nicefrac{\sin(z)}{z}$ attains a maximum at $z=\nicefrac{\pi}{2}$, where its value is $\nicefrac{2}{\pi}$, and the right-hand inequality is a standard trigonometric result that holds for any $z\in(0,\pi)$. Let $j_{\:0}$ be the largest integer such that $(2j_{\:0}-1)<\nicefrac{1}{z}$. Then $z<\nicefrac{1}{(2j_{\:0}-1)}$ and $z\geq\nicefrac{1}{(2j_{\:0}+1)}$. Observe that if $j\leq j_{\:0}$, then $(2j-1)\cdot z\leq(2j_{\:0}-1)\cdot z<1<\pi$, so that $\nicefrac{\sin\big((2j-1)\cdot z\big)}{\big((2j-1)\cdot z\big)}\leq 1$ or $\nicefrac{\sin\big((2j-1)\cdot z\big)}{(2j-1)}\leq z$. On the other hand, if $j>j_{\:0}$ then $\nicefrac{1}{(2j-1)}\leq\nicefrac{1}{(2j_{\:0}+1)}$. Now, we have:
\begin{align*}
\mathsf{1}_{\:J}(z)&=\frac{1}{2}+\frac{2}{\pi}\cdot \sum_{j=1}^J\frac{\sin\big((2j-1)\cdot z\big)}{2j-1} \\
&=\frac{1}{2}+\frac{2}{\pi}\cdot \left(\sum_{j=1}^{j_{\:0}}\frac{\sin\big((2j-1)\cdot z\big)}{2j-1}+\sum_{j=j_{\:0}+1}^J\frac{\sin\big((2j-1)\cdot z\big)}{2j-1}\right).
\end{align*}
Then, we have:
\begin{align*}
\left|\mathsf{1}_{\:J}(z)\right| &\leq\frac{1}{2}+\frac{2}{\pi}\cdot \left(|z|\sum_{j=1}^{j_{\:0}}\left|\frac{\sin\big((2j-1)\cdot z\big)}{(2j-1)\cdot z}\right|+\left|\sum_{j=j_{\:0}+1}^J\frac{\sin\big((2j-1)\cdot z\big)}{2j-1}\right|\right) \\
&<\frac{1}{2}+\frac{2}{\pi}\left(z\sum_{j=1}^{j_{\:0}}1+\frac{1}{(2j_{\:0}+1)}\left|\sum_{j=j_{\:0}+1}^{J}\sin\big((2j-1)\cdot z\big)\right|\right) \\
&=\frac{1}{2}+\frac{2\cdot j_{\:0}\cdot z}{\pi}+\frac{2}{\pi\cdot (2\cdot j_{\:0}+1)}\left|\sum_{j=j_{\:0}+1}^J
\sin\big((2j-1)\cdot z\big)\right| \\
&<\frac{1}{2}+\frac{2j_{\:0}}{\pi\cdot (2j_{\:0}-1)}+\frac{2z}{\pi}\cdot \left|E_{\:J}(z)-E_{\:j_{\:0}}(z)\right|\\
&<\frac{1}{2}+\frac{1}{\pi}+\frac{2z}{\pi}.\frac{2}{\sin(\frac{z}{2})} \\
&=\frac{1}{2}+\frac{1}{\pi}+\frac{4}{\pi}.\frac{2}{\frac{\sin(\frac{z}{2})}{\frac{z}{2}}}\\
&\leq\frac{1}{2}+\frac{1}{\pi}+\frac{4}{\pi}.\frac{2}{\frac{2}{\pi}}.
\end{align*}
So the sequence of functions $\mathsf{1}_{\:J}(z)$ is uniformly bounded by $M=9/2+1/\pi$.
\end{proof}

\subsection{Proof of Lemma \ref{Quantile:lemma:population-asymptotic}}
\begin{proof}
(A)\ Given $\epsilon>0$, fix an arbitrary $\theta\in(0,1)$. Since $\mathsf{1}_{\:J}(z)$ converges pointwise to $\mathsf{1}(0\leq z)$ for $z\in(-\pi,\pi)\setminus\{0\}$, we can conclude that $\mathsf{1}_{\:J}(\tilde{x},\theta)$ converges pointwise to $\mathsf{1}(\tilde{x}\leq\theta)$ for $(0,1)\setminus\{\theta\}$. 

Now, from part A of Lemma \ref{Quantile:lemma:indicator-uniform-convergence}, we know $\mathsf{1}_{\:J}(z)$ converges uniformly to $\mathsf{1}(z\leq 0)$ for $z\in(\-\pi,\pi)\setminus(-\delta,\delta)$ for any $\delta\in(0,\pi)$. So, there exists $J_{\:\epsilon,\:\delta}\in\mathbb{N}$ such that $\big|\mathsf{1}_{\:J}(z)-\mathsf{1}(z\leq 0)\big|\leq\nicefrac{\epsilon}{2}$ if $z\in(\-\pi,\pi)\setminus(-\delta,\delta)$ and $J>J_{\:\epsilon,\:\delta}$. Let us define the set 
\begin{align*}
    R_{\:\delta,\:\theta}=(0,1)\setminus(\theta-\delta,\theta+\delta).
\end{align*}
If $\tilde{x}\in R_{\:\delta,\:\theta}$, then $\tilde{x}-\theta\in(\-\pi,\pi)\setminus(-\delta,\delta)$, and in turn, we will have $\left|\mathsf{1}_{\:J}(\tilde{x},\theta)-\mathsf{1}(\tilde{x}\leq\theta)\right|\leq\nicefrac{\epsilon}{2}$ for $J>J_{\:\epsilon,\:\delta}$. 

Observe that the Lebesgue measure of $(0,1)\setminus R_{\:\delta,\:\theta}$ satisfies $\lambda\big((0,1)\setminus R_{\:\delta,\:\theta}\big)\leq 2\delta$, for any $\theta\in(0,1)$. Also, from part B of Lemma \ref{Quantile:lemma:indicator-uniform-convergence}, we realize that 
\begin{align*}
    \big|\:\mathsf{1}_{\:J}(\tilde{x},\theta)\:\big|\leq (9/2+1/\pi).
\end{align*}
Then, for $\tilde{x}\in(0,1)$, we have
\begin{align*}
\big|\:\mathsf{1}(\tilde{x},\theta)-\mathsf{1}_{\:J}(\tilde{x},\theta)\:\big|\leq\big|\:\mathsf{1}(\tilde{x},\theta)\:\big|+\big|\:\mathsf{1}_{\:J}(\tilde{x},\theta)\:\big|\leq 1+\:(9/2+1/\pi).
\end{align*}
For convenience, let us denote this bound by $M_{\:0}$. 

Since by Assumption \ref{Quantile:assumption:positive-density}, the probability measure $P(\cdot)$ is absolutely continuous with respect to Lebesgue measure $\lambda(\cdot)$, there exists a $\delta_{\:\epsilon}>0$ such that if $\lambda(A)<\delta_{\:\epsilon}$ for some $A\in\mathbb{B}(0,1)$, we have $P(A)<\nicefrac{\epsilon}{(2\cdot M_{\:0})}$. Take $A=R_{\:\nicefrac{\delta_{\:\epsilon}}{2},\:\theta}$; then for $J>J_{\:\epsilon,\:\nicefrac{\delta_{\:\epsilon}}{2}}$, we have:
\begin{align*}
&\hspace{-1em}\Big|\:\mathbb{E}_{\:P}\big(\mathsf{1}_{\:J}(\tilde{x},\theta)\big)-F^{\:(P)}(\theta)\:\Big| \\
&=\Big|\:\mathbb{E}_{\:P}\big(\mathsf{1}_{\:J}(\tilde{x},\theta)\big)-\mathbb{E}_{\:P}\big(\mathsf{1}(\tilde{x},\theta)\big)\:\Big| \\
&=\Big|\:\mathbb{E}_{\:P}\big(\mathsf{1}(\tilde{x},\theta)-\mathsf{1}_{\:J}(\tilde{x},\theta)\big)\:\Big| \\
&\leq\mathbb{E}_{\:P}\Big(\big|\:\mathsf{1}(\tilde{x},\theta)-\mathsf{1}_{\:J}(\tilde{x},\theta)\:\big|\Big) \\
&=\int_{\tilde{x}\in(0,1)}\big|\:\mathsf{1}(\tilde{x},\theta)-\mathsf{1}_{\:J}(\tilde{x},\theta)\:\big|\,dP \\
&=\int_{\tilde{x}\in R_{\:\nicefrac{\delta_{\:\epsilon}}{2},\:\theta}}\big|\:\mathsf{1}(\tilde{x},\theta)-\mathsf{1}_{\:J}(\tilde{x},\theta)\:\big|\,dP+\int_{\tilde{x}\in(0,1)\setminus R_{\:\nicefrac{\delta_{\:\epsilon}}{2},\:\theta}}\big|\:\mathsf{1}(\tilde{x},\theta)-\mathsf{1}_{\:J}(\tilde{x},\theta)\:\big|\,dP \\
&\leq\int_{\tilde{x}\in R_{\:\nicefrac{\delta_{\:\epsilon}}{2},\:\theta}}\frac{\epsilon}{2}\cdot dP+\int_{\tilde{x}\in(0,1)\setminus R_{\:\nicefrac{\delta_{\:\epsilon}}{2},\:\theta}}M_{\:0}\cdot dP \\
&=\frac{\epsilon}{2}\cdot P\left(R_{\:\nicefrac{\delta_{\:\epsilon}}{2},\:\theta}\right)+M_{\:0}\cdot P\left((0,1)\setminus R_{\:\nicefrac{\delta_{\:\epsilon}}{2},\:\theta}\right) \\
&<\frac{\epsilon}{2}\cdot 1+M_{\:0}\cdot \frac{\epsilon}{2\cdot M_{\:0}}=\epsilon.
\end{align*}
Note that $J_{\:\epsilon,\:\nicefrac{\delta_{\:\epsilon}}{2}}$ is independent of $\theta$. Since $\epsilon$ is an arbitrary positive number, we must have: $\lim\limits_{J\to\infty}\mathbb{E}_{\:P}\:\big(\mathsf{1}_{\:J}(\tilde{x},\theta)\big)\to F^{\:(P)}(\theta)$ uniformly for $\theta\in(0,1)$. \\

(B)\ Consider an arbitrary $p\in(0,1)$. We can pick an $\epsilon_{\:p}>0$ such that $\epsilon_{\:p}<\min(p,1-p)$. For the random variable $\tilde{x}$, we have $0<\tilde{x}< 1$; that is, $\mathsf{1}(\tilde{x}\leq 0)=0$ and $\mathsf{1}(\tilde{x}\leq 1)=1$, hence, $F^{\:P}(0)=0$ and $F^{\:P}(1)=1$. By part A, there exists $J_{\:0}\in\mathbb{N}$ such that, for $J>J_{\:0}$, we have:
\begin{align*}
&\quad\big|\:\mathbb{E}_{\:P}\big(\mathsf{1}_{\:J}(\tilde{x},0)\big)-F^{\:P}(0)\:\big|<\epsilon_{\:p}\quad \Leftrightarrow \quad \big|\:\mathbb{E}_{\:P}\big(\mathsf{1}_{\:J}(\tilde{x},0)\big)\:\big|<\epsilon_{\:p} \\
&\Rightarrow\mathbb{E}_{\:P}\big(\mathsf{1}_{\:J}(\tilde{x},0)\big)<\epsilon_{\:p}.
\end{align*}
Also, by part A, there exists $J_{\:1}\in\mathbb{N}$ such that, for $J>J_{\:1}$, we have:
\begin{align*}
&\quad\big|\:\mathbb{E}_{\:P}\big(\mathsf{1}_{\:J}(\tilde{x},1)\big)-F^{\:P}(1)\:\big|<\epsilon_{\:p} \quad \Leftrightarrow \quad \big|\:\mathbb{E}_{\:P}\big(\mathsf{1}_{\:J}(\tilde{x},0)\big)-1\:\big|<\epsilon_{\:p} \\
&\Rightarrow\mathbb{E}_{\:P}\big(\mathsf{1}_{\:J}(\tilde{x},1)\big)>1-\epsilon_{\:p}.
\end{align*}
Then, for $J>J_{\:p}=\operatorname{max}(J_{\:0},J_{\:1})$, we have:
\begin{align*}
\mathbb{E}_{\:P}\big(\mathsf{1}_{\:J}(\tilde{x},0)\big)<\epsilon_{\:p}<p<1-\epsilon_{\:p}<\mathbb{E}_{\:P}\big(\mathsf{1}_{\:J}(\tilde{x},1)\big).
\end{align*}
From the expression in equation \ref{Quantile:proof:eq:b}, we realize that $\mathbb{E}_{\:P}\big(\mathsf{1}_{\:J}(\tilde{x},\theta)\big)$ is a continuous function of $\theta$. So, by the intermediate value theorem, there is a number $\theta=q_{\:J,\:p}^{\:(P)}\in(0,1)$, satisfying $\mathbb{E}_{\:P}\big(\mathsf{1}_{\:J}(\tilde{x},\theta)\big)=p$. So, we conclude that $p\in\mathcal{Q}_{\:J}^{\:(P)}$, if $J>J_{\:p}$, and this is true for arbitrary $p\in(0,1)$. Hence $\lim\limits_{J\to\infty}\mathcal{Q}_{\:J}^{\:(P)}=(0,1)$. \\

(C)\ For $p\in(0,1)$, consider any sequence of numbers $q_{\:J,\:p}^{\:(P)}$ when such a sequence exists. From part B we know that $q_{\:J,\:p}^{\:(P)}$ exists when $J$ is large ($J>J_{\:p}$ say). For such a $J$, we have $\mathbb{E}_{\:P}\Big(\mathsf{1}_{\:J}\big(\tilde{x},q_{\:J,\:p}^{\:(P)}\big)\Big)=p$. Given, $\epsilon>0$, we know from part A that there exists an integer $J_{\:0}$ such that $\big|\:\mathbb{E}_{\:P}\big(\mathsf{1}_{\:J}(\tilde{x},\theta)\big)-F^{\:(P)}(\theta)\:\big|<\epsilon$, for any $\theta\in(0,1)$ if $J>J_{\:0}$. If we  replace $\theta$ with $q_{\:J,\:p}^{\:(P)}$ in this inequality, we obtain $\Big|\:p-F^{\:(P)}\big(q_{\:J,\:p}^{\:(P)}\big)\:\Big|<\epsilon$ when $J>J_{\:\epsilon}=\operatorname{max}(J_{\:p},J_{\:0})$.

Now, we know that $F^{\:(P)}\big(q_{\:p}^{\:(P)}\big)=p$. Then if $J>J_{\:\epsilon}$, it follows that
$\Big|\:F^{\:(P)}\big(q_{\:J,\:p}^{\:(P)}\big)-F^{\:(P)}\big(q_{\:p}^{\:(P)}\big)\:\Big|<\epsilon$. In other words, $\lim\limits_{J\to\infty}F^{\:(P)}\big(q_{\:J,\:p}^{\:(P)}\big)\to F^{\:(P)}\big(q_{\:p}^{\:(P)}\big)$, and hence by from part A of Lemma \ref{Quantile:lemma:inverse-continuity} we have that $\lim\limits_{J\to\infty}q_{\:J,\:p}^{\:(P)}\to q_{\:p}^{\:(P)}$. 

Next, we will show that $\lim\limits_{J\to\infty}\bigtriangledown\big(\mathcal{G}_{\:J,\:p}^{\:(P)}\big)\to 0$. Pick $\epsilon>0$. If $J>J_{\:\nicefrac{\epsilon}{2}}$, then
\begin{align*}
\big|\:q_{\:J,\:p}^{\:(P)}-q_{\:p}^{\:(P)}\:\big|<\nicefrac{\epsilon}{2}\text{ and }\big|\:{q'}_{\:J,\:p}^{\:(P)}-q_{\:p}^{\:(P)}\:\big|<\nicefrac{\epsilon}{2}
\end{align*}
for $\big\{q_{\:J,\:p}^{\:(P)},{q'}_{\:J,\:p}^{\:(P)}\big\}\in\mathcal{G}_{\:J,\:p}^{\:(P)}$ and $q_{\:J,\:p}^{\:(P)}\neq{q'}_{\:J,\:p}^{\:(P)}$. Then, we have:
\begin{align*}
\big|\:q_{\:J,\:p}^{\:(P)}-{q'}_{\:J,\:p}^{\:(P)}\:\big|&=\big|\:\big(q_{\:J,\:p}^{\:(P)}-q_{\:p}^{\:(P)}\big)-\big(\:{q'}_{\:J,\:p}^{\:(P)}-q_{\:p}^{\:(P)}\big)\:\big| \\
&\leq\big|\:\big(q_{\:J,\:p}^{\:(P)}-q_{\:p}^{\:(P)}\big)\:\big|+\big|\:\big(\:{q'}_{\:J,\:p}^{\:(P)}-q_{\:p}^{\:(P)}\big)\:\big| \\
&<\nicefrac{\epsilon}{2}+\nicefrac{\epsilon}{2}=\epsilon.
\end{align*}

Since $\epsilon$ is arbitrary, we have proved that $\lim\limits_{J\to\infty}\bigtriangledown\big(\mathcal{G}_{\:J,\:p}^{\:(P)}\big)\to 0$ for $p\in(0,1)$.
\end{proof}

\subsection{Proof of Lemma \ref{Quantile:lemma:population-expression}}
\begin{proof}
(A)\ We have:
\begin{align*}
&\hspace{-2em}\bar{\mathsf{G}}_{\:J,\:p}(X,\theta)=\frac{1}{|X|}\cdot\mathsf{G}_{\:J,\:p}(X,\theta)=\frac{1}{|X|}\Big(\sum_{i\in I}\rho_{\:J,\:p}(x_{\:i}-\theta)\Big) \\
&=\frac{1}{|X|}\bigg(\sum_{i\in I}\frac{\pi}{4}-\frac{2}{\pi}\sum_{j=1}^J\frac{\cos\big((2j-1)(x_{\:i}-\theta)\big)}{(2j-1)^2}+\big(p-\frac{1}{2}\big)(x_{\:i}-\theta)\bigg) \\
&=\frac{1}{|X|}\bigg(\sum_{i\in I}\frac{\pi}{4}-\frac{2}{\pi}\sum_{j=1}^J\frac{\big(\mathsf{c}_{\:2j-1}(x_{\:i})\mathsf{c}_{\:2j-1}(\theta)-\mathsf{c}_{\:2j-1}(x_{\:i})\mathsf{c}_{\:2j-1}(\theta)\big)}{(2j-1)^2} \\
&\hspace{5em}+\big(p-\frac{1}{2}\big)(x_{\:i}-\theta)\bigg) \\
&=\frac{\pi}{4}-\frac{2}{\pi}\sum_{j=1}^J\frac{\Big(\big(\nicefrac{1}{|X|}\sum_{i\in I}\mathsf{c}_{\:2j-1}(x_{\:i})\big)\mathsf{c}_{\:2j-1}(\theta)+\big(\nicefrac{1}{|X|}\sum_{i\in I}\mathsf{c}_{\:2j}(x_{\:i})\big)\mathsf{c}_{\:2j}(\theta)\Big)}{(2j-1)^2} \\
&\hspace{5.5em}+\big(p-\frac{1}{2}\big)\Big(\nicefrac{1}{|X|}\sum_{i\in I}x_{\:i}-\theta\Big) \\
&\hspace{-3.5em}\Leftrightarrow\bar{\mathsf{G}}_{\:J,\:p}(X,\theta)=\frac{\pi}{4}-\frac{2}{\pi}\cdot \sum_{j=1}^J\frac{\Big(\bar{\mathsf{C}}_{\:2j-1}(X)\cdot \mathsf{c}_{\:2j-1}(\theta)+\bar{\mathsf{C}}_{\:2j}(X)\cdot \mathsf{c}_{\:2j}(\theta)\Big)}{(2j-1)^2} \\
&\hspace{5em}+\big(p-\frac{1}{2}\big)\Big(\bar{X}-\theta\Big).
\end{align*}

Now, if we differentiate both sides with respect to $\theta$, we obtain
\begin{align*}
&\hspace{-2.5em}\frac{\partial}{\partial\theta}\big(\bar{\mathsf{G}}_{\:J,\:p}(X,\theta)\big)=-\frac{2}{\pi}\sum_{j=1}^J\frac{\Big(\bar{\mathsf{C}}_{\:2j-1}(X){\mathsf{c}_{\:2j-1}}'(\theta)+\bar{\mathsf{C}}_{\:2j}(X){\mathsf{c}_{\:2j}}'(\theta)\Big)}{(2j-1)^2}-\big(p-\frac{1}{2}\big) \\
&=\frac{1}{2}-\frac{2}{\pi}\sum_{j=1}^J\frac{\Big(\bar{\mathsf{C}}_{\:2j}(X)\mathsf{c}_{\:2j-1}(\theta)-\bar{\mathsf{C}}_{\:2j-1}(X)\mathsf{c}_{\:2j}(\theta)\Big)}{2j-1}-p \\
&=\frac{1}{2}-\frac{2}{\pi}\sum_{j=1}^J\frac{\Big(\big(\nicefrac{1}{|X|}\sum_{i\in I}\mathsf{c}_{\:2j}(x_{\:i})\big)\mathsf{c}_{\:2j-1}(\theta)-\big(\nicefrac{1}{|X|}\sum_{i\in I}\mathsf{c}_{\:2j-1}(x_{\:i})\big)\mathsf{c}_{\:2j}(\theta)\Big)}{2j-1}-p \\
&=\frac{1}{|X|}\sum_{i\in I}\Big(\frac{1}{2}-\frac{2}{\pi}\sum_{j=1}^J\frac{\big(\mathsf{c}_{\:2j}(x_{\:i})\mathsf{c}_{\:2j-1}(\theta)-\mathsf{c}_{\:2j-1}(x_{\:i})\mathsf{c}_{\:2j}(\theta)\big)}{2j-1}\Big)-p \\
&=\frac{1}{|X|}\sum_{i\in I}\Big(\frac{1}{2}-\frac{2}{\pi}\sum_{j=1}^J\frac{\sin\big((2j-1)(x_{\:i}-\theta)\big)}{2j-1}\Big)-p \\
&=\frac{1}{|X|}\sum_{i\in I}\mathsf{1}_{\:J}(x_{\:i},\theta)-p=\hat{F}_{\:J}(X,\theta)-p.
\end{align*}

(B) Let us define $\bar{\mathsf{E}}_{\:J,\:p}(X,\theta)=\nicefrac{1}{|X|}\cdot\big(\mathsf{G}_{\:p}(X,\theta)-\mathsf{G}_{\:J,\:p}(X,\theta)\big)$, then, we have
\begin{align*}
\hspace{-2em}&\big|\:\bar{\mathsf{E}}_{\:J,\:p}(X,\theta)\:\big|=\nicefrac{1}{|X|}\bigg|\:\sum_{i\in I}\big(\rho_{\:p}(x_{\:i},\theta)-\rho_{\:J,\:p}(x_{\:i},\theta)\big)\:\bigg|\\
&=\nicefrac{1}{|X|}\bigg|\:\sum_{i\in I}\frac{2}{\pi}\sum_{j=J+1}^{\infty}\frac{\cos\big((2j-1)(x_{\:i}-\beta)\big)}{(2j-1)^2}\:\bigg| \\
&\leq\nicefrac{1}{|X|}\sum_{i\in I}\frac{2}{\pi}\sum_{j=J+1}^{\infty}\bigg|\:\frac{\cos\big((2j-1)(x_{\:i}-\beta)\big)}{(2j-1)^2}\:\bigg|\\
&\leq\nicefrac{1}{|X|}\sum_{i\in I}\frac{2}{\pi}\sum_{j=J+1}^{\infty}\frac{1}{(2j-1)^2} \\
&=\frac{2}{\pi}\sum_{j=J+1}^{\infty}\frac{1}{(2j-1)^2}=O\big(\nicefrac{1}{J}\big).
\end{align*}
Hence, $\big|\:\bar{\mathsf{E}}_{\:J,\:p}(X,\theta)\:\big|=O(J^{-1})$ for arbitrary $X$ and $p\in(0,1)$ as claimed. \\

(C) We will show that, for large $J$, $\bar{\mathsf{G}}_{\:J,\:p}(X,\theta)$ cannot have its minimum at $0$ or $1$. This proves that $\hat{q}_{\:J,\:p}(X)$ lies in the open interval $(0,1)$ for large $J$, since the minimizer $\hat{q}_{\:J,\:p}(X)$ is guaranteed to exist in the compact interval $[0,1]$. We know $\bar{\mathsf{G}}_{\:p}(X,\theta)$ is minimized for $\theta=\hat{q}_{\:p}(X)\in(0,1)$. Let us pick $\epsilon_{\:p,\:X}>0$ such that \\
$\epsilon_{\:p,\:X}<\operatorname{min}\Big(\bar{\mathsf{G}}_{\:p}(X,0)-\bar{\mathsf{G}}_{\:p}\big(X,\hat{q}_{\:p}(X)\big),\bar{\mathsf{G}}_{\:p}(X,1)-\bar{\mathsf{G}}_{\:p}\big(X,\hat{q}_{\:p}(X)\big)\Big)$. From part B we know that there exists $J_{\:p,\:X}\in\mathbb{N}$ such that $\big|\:\bar{\mathsf{G}}_{\:p,\:J}(X,\theta)-\bar{\mathsf{G}}_{\:p}(X,\theta)\:\big|<\nicefrac{\epsilon_{\:p,\:X}}{3}$ for $\theta\in(0,1)\text{ and }J>J_{\:p,\:X}$.
In particular we have $\bar{\mathsf{G}}_{\:p,\:J}(X,0)-\bar{\mathsf{G}}_{\:p}(X,0)>-\nicefrac{\epsilon_{\:p,\:X}}{3}$ and $\bar{\mathsf{G}}_{\:p}\big(X,\hat{q}_{\:p}(X)\big)-\bar{\mathsf{G}}_{\:p,\:J}\big(X,\hat{q}_{\:p}(X)\big)>-\nicefrac{\epsilon_{\:p,\:X}}{3}$ for $J>J_{\:p,\:X}$. Then, for $J>J_{\:p,\:X}$ we have
\begin{align*}
&\hspace{-2em}\bar{\mathsf{G}}_{\:p,\:J}(X,0)-\bar{\mathsf{G}}_{\:p}\big(X,\hat{q}_{\:p,\:J}(X)\big) \\
&= \big(\bar{\mathsf{G}}_{\:p,\:J}(X,0)-\bar{\mathsf{G}}_{\:p}(X,0)\big)+\Big(\bar{\mathsf{G}}_{\:p}\big(X,\hat{q}_{\:p}(X)\big)-\bar{\mathsf{G}}_{\:p,\:J}\big(X,\hat{q}_{\:p}(X)\big)\Big) \\
&\hspace{2em}+\Big(\bar{\mathsf{G}}_{\:p}(X,0)-\bar{\mathsf{G}}_{\:p}\big(X,\hat{q}_{\:p}(X)\big) \\
&>-\frac{\epsilon_{\:p,\:X}}{3}-\frac{\epsilon_{\:p,\:X}}{3}+\epsilon_{\:p,\:X}=\frac{\epsilon_{\:p,\:X}}{3}>0.
\end{align*}
In other words, we have $\bar{\mathsf{G}}_{\:p,\:J}(X,0)>\bar{\mathsf{G}}_{\:p}\big(X,\hat{q}_{\:p,\:J}(X)\big)$ for $J>J_{\:p,\:X}$. Similarly, we can show $\bar{\mathsf{G}}_{\:p,\:J}(X,1)>\bar{\mathsf{G}}_{\:p}\big(X,\hat{q}_{\:p,\:J}(X)\big)$ for $J>J_{\:p,\:X}$. 

Thus we have demonstrated a value $\theta = \hat{q}_{\:p,\:J}(X)$, for which $\bar{\mathsf{G}}_{\:p,\:J}(X,0)>\bar{\mathsf{G}}_{\:p}(X,\theta)$ and $\bar{\mathsf{G}}_{\:p,\:J}(X,1)>\bar{\mathsf{G}}_{\:p}(X,\theta)$ for $J>J_{\:p,\:X}$. In conclusion, for sufficiently large $J$, $\bar{\mathsf{G}}_{\:J,\:p}(X,\theta)$ cannot attain a minimum at $0$ or $1$.
\end{proof}

\subsection{Proof of Lemma \ref{Quantile:lemma:sample-asymptotic}}
\begin{proof}(A)\ 
Suppose Assumption \ref{Quantile:assumption:positive-density} holds. Let $\theta\in(0,1)$, and consider arbitrary $\epsilon>0$. First, observe that, for $z\in(0,1)$ and $j\in\mathbb{N}$, we have:
\begin{align*}
\mathsf{1}_{\:j}(z)=\mathsf{1}_{\:j}(z,0)=\frac{1}{2}-\frac{2}{\pi}\cdot\sum_{k=1}^j\frac{\mathsf{c}_{\:2k}(z)}{2k-1}
&\quad\Leftrightarrow\quad\frac{2}{\pi}\cdot\sum_{k=1}^j\frac{\mathsf{c}_{\:2k}(z)}{2k-1}=\frac{1}{2}-\mathsf{1}_{\:j}(z).
\end{align*}
For $j\in\mathbb{N}$, let us define the SEP statistic $\mathsf{I}_{\:j}(X):=\sum_{i\in I}\mathsf{1}_{\:j}(x_{\:i})$ and a standardized version of $\mathsf{I}_{\:j}(X)$, which is the SEP statistic $\bar{\mathsf{I}}_{\:j}(X):=\frac{\mathsf{I}_{\:i}(X)}{|X|}$. Observe that
\begin{align*}
\bar{\mathsf{I}}_{\:j}(X)=\frac{1}{2}-\frac{2}{\pi}\cdot\sum_{k=1}^j\frac{\bar{\mathsf{C}}_{\:2k}(X)}{2k-1}
\end{align*}
Let $r_{\:j}(X):=\bar{\mathsf{C}}_{\:j}(X)-\mathbb{E}_{\:P}\big(\mathsf{c}_{\:j}(\tilde{x})\big)$ for $j\in\mathbb{N}$. Then:
\begin{align*}
\frac{2}{\pi}\cdot\sum_{k=1}^j\frac{r_{\:2k}(X)}{2k-1}&=\frac{2}{\pi}\cdot\sum_{k=1}^j\frac{\Big(\bar{\mathsf{C}}_{\:2k}(X)-\mathbb{E}_{\:P}\big(\mathsf{c}_{\:2k}(\tilde{x})\big)\Big)}{2k-1} \\
&=-\Big(\frac{1}{2}-\frac{2}{\pi}\cdot\sum_{k=1}^j\frac{\bar{\mathsf{C}}_{\:2k}(X)}{2k-1}\Big)+\mathbb{E}_{\:P}\Big(\frac{1}{2}-\frac{2}{\pi}\cdot\sum_{k=1}^j\frac{\mathsf{c}_{\:2k}(\tilde{x})}{2k-1}\Big) \\
&=-\bar{\mathsf{I}}_{\:j}(X)+\mathbb{E}_{\:P}\big(\mathsf{1}_{\:j}(\tilde{x})\big) \\
&\Leftrightarrow-\frac{2}{\pi}\cdot\sum_{k=1}^j\frac{r_{\:2k}(X)}{2k-1}=\bar{\mathsf{I}}_{\:j}(X)-\mathbb{E}_{\:P}\big(\mathsf{1}_{\:j}(\tilde{x})\big).
\end{align*}

Since $|\mathsf{c}_{\:2j-1}(\tilde{x})|\leq 1$, we have that $|\mathbb{E}_{\:P}\big(\mathsf{c}_{\:2j-1}(\tilde{x})\big)|\leq 1$ for  $j\in\mathbb{N}$. By Kolmogorov's strong law of large numbers, we have $\bar{\mathsf{C}}_{\:2j-1}(X)\overset{a.s.}{\longrightarrow}\mathbb{E}_{\:P}\big(\mathsf{c}_{\:2j-1}(\tilde{x})\big)$. So, for each $j$, there exist a null set $N_{\:j}\subset \Omega$, such that $\bar{\mathsf{C}}_{\:2j-1}\left(X(w,n)\right)\to\mathbb{E}_{\:P}\big(\mathsf{c}_{\:2j-1}(\tilde{x})\big)$ as $n\to\infty$, if $w\in {N_{\:j}}'$. 

From part B of Lemma \ref{Quantile:lemma:indicator-uniform-convergence}, we know $|\mathsf{1}_{\:j}(\tilde{x})|\leq\nicefrac{9}{2}+\nicefrac{1}{\pi}$, and so $|\mathbb{E}_{\:P}\big(\mathsf{1}_{\:j}(\tilde{x})\big)|\leq\nicefrac{9}{2}+\nicefrac{1}{\pi}$ for  $j\in\mathbb{N}$. Again, by Kolmogorov's strong law of large numbers, we have $\bar{\mathsf{I}}_{\:j}(X)\overset{a.s.}{\longrightarrow}\mathbb{E}_{\:P}\big(\mathsf{1}_{\:j}(\tilde{x})\big)$. So, for each $j$ there exists a null set $M_{\:j}\subset \Omega$, such that $\bar{\mathsf{I}}_{\:j}\left(X(w,n)\right)\to\mathbb{E}_{\:P}\big(\mathsf{1}_{\:j}(\tilde{x})\big)$ as $n\to\infty$, if $w\in {M_{\:j}}'$. Take $N=\big(\bigcup_{j\in\mathbb{N}}N_{\:j}\big)\bigcup\big(\bigcup_{j\in\mathbb{N}}M_{\:j}\big)$; since $N$ is a countable union of null sets, it is also a null set. Now if $w\in N'$, then $\bar{\mathsf{C}}_{\:2j-1}\big(X(w,n)\big)\to\mathbb{E}_{\:P}\big(\mathsf{c}_{\:2j-1}(\tilde{x})\big)$ for all $j\in\mathbb{N}$ and $\bar{\mathsf{I}}_{\:j}\big(X(w,n)\big)\to\mathbb{E}_{\:P}\big(\mathsf{1}_{\:j}(\tilde{x})\big)$ for all $j\in\mathbb{N}$. 

If we consider arbitrary $w\in N'$; then, from part A of Lemma \ref{Quantile:lemma:population-expression} and equation {Quantile:lemma:inverse-continuity}, we have: 
\begin{align*}
&\hspace{-2em}\hat{F}_{\:J}\big(X(w,n),\theta\big)-\mathbb{E}_{\:P}\big(\mathsf{1}_{\:J}(\tilde{x},\theta)\big) \\
&=\frac{1}{2}-\frac{2}{\pi}\cdot\sum_{j=1}^J\left(\bar{\mathsf{C}}_{\:2j}\big(X(w,n)\big)\cdot\frac{\mathsf{c}_{\:2j-1}(\theta)}{2j-1}-\bar{\mathsf{C}}_{\:2j-1}\big(X(w,n)\big)\cdot\frac{\mathsf{c}_{\:2j}(\theta)}{2j-1}\right) \\
&\hspace{2cm}-\frac{1}{2}+\frac{2}{\pi}\cdot\sum_{j=1}^J\left(\mathbb{E}_{\:P}\big(\mathsf{c}_{\:2j}(\tilde{x})\big)\cdot\frac{\mathsf{c}_{\:2j-1}(\theta)}{2j-1}-\mathbb{E}_{\:P}\big(\mathsf{c}_{\:2j-1}(\tilde{x})\big)\cdot\frac{\mathsf{c}_{\:2j}(\theta)}{2j-1}\right) \\
&=\frac{2}{\pi}\cdot\sum_{j=1}^J\Big(\bar{\mathsf{C}}_{\:2j-1}\big(X(w,n)\big)-\mathbb{E}_{\:P}\big(\mathsf{c}_{\:2j-1}(\tilde{x})\big)\Big)\cdot\frac{\mathsf{c}_{\:2j}(\theta)}{2j-1} \\
&\hspace{2cm}-\frac{2}{\pi}\cdot\sum_{j=1}^J\mathsf{c}_{\:2j-1}(\theta)\cdot\frac{\Big(\bar{\mathsf{C}}_{\:2j}\big(X(w,n)\big)-\mathbb{E}_{\:P}\big(\mathsf{c}_{\:2j}(\tilde{x})\big)\Big)}{2j-1}.
\end{align*}
From our previous notation,  $r_{\:j}\big(X(w,n)\big)=\bar{\mathsf{C}}_{\:j}\big(X(w,n)\big)-\mathbb{E}_{\:P}\big(\mathsf{c}_{\:j}(\tilde{x})\big)$ for $j\in\mathbb{N}$, and hence
\begin{align*}
&\hspace{-2em}\hat{F}_{\:J}\big(X(w,n),\theta\big)-\mathbb{E}_{\:P}\big(\mathsf{1}_{\:J}(\tilde{x},\theta)\big) \\
&=\frac{2}{\pi}\cdot\sum_{j=1}^Jr_{\:2j-1}\big(X(w,n)\big)\cdot\frac{\mathsf{c}_{\:2j}(\theta)}{2j-1}-\frac{2}{\pi}\cdot\sum_{j=1}^J\mathsf{c}_{\:2j-1}(\theta)\cdot\frac{r_{\:2j}\big(X(w,n)\big)}{2j-1}.
\end{align*}
Now, from Abel's identity for partial sums applied to the first summation, we have:
\begin{align*}
&\hspace{-2em}\frac{2}{\pi}\cdot\sum_{j=1}^Jr_{\:2j-1}\big(X(w,n)\big)\cdot\frac{\mathsf{c}_{\:2j}(\theta)}{2j-1} \\
&=r_{\:2J-1}\big(X(w,n)\big)\cdot\left(\frac{2}{\pi}\cdot\sum_{k=1}^J\frac{\mathsf{c}_{\:2k}(\theta)}{2k-1}\right) \\
&\hspace{2em}+\sum_{j=1}^{J-1}\Big(r_{\:2j-1}\big(X(w,n)\big)-r_{\:2j+1}\big(X(w,n)\big)\Big)\cdot\left(\frac{2}{\pi}\cdot\sum_{k=1}^j\frac{\mathsf{c}_{\:2k}(\theta)}{2k-1}\right) \\
&=r_{\:2J-1}\big(X(w,n)\big)\cdot\left(\frac{1}{2}-\mathsf{1}_{\:J}(\theta)\right) \\
&\hspace{2em}+\sum_{j=1}^{J-1}\Big(r_{\:2j-1}\big(X(w,n)\big)-r_{\:2j+1}\big(X(w,n)\big)\Big)\cdot\left(\frac{1}{2}-\mathsf{1}_{\:j}(\theta)\right).
\end{align*}
Proceeding analogously for the second summation, we obtain:
\begin{align*}
&\hspace{-2em}-\frac{2}{\pi}\cdot\sum_{j=1}^J\mathsf{c}_{\:2j-1}(\theta)\cdot\frac{r_{\:2j}\big(X(w,n)\big)}{2j-1} \\
&=\mathsf{c}_{\:2J-1}(\theta)\cdot\left(-\frac{2}{\pi}\cdot\sum_{k=1}^J\frac{r_{\:2k}\big(X(w,n)\big)}{2k-1}\right) \\
&\hspace{2em}+\sum_{j=1}^{J-1}\big(\mathsf{c}_{\:2j-1}(\theta)-\mathsf{c}_{\:2j+1}(\theta)\big)\cdot\left(-\frac{2}{\pi}\cdot\sum_{k=1}^j\frac{r_{\:2k}\big(X(w,n)\big)}{2k-1}\right) \\
&=\mathsf{c}_{\:2J-1}(\theta)\cdot\Big(\bar{\mathsf{I}}_{\:J}\big(X(w,n)\big)-\mathbb{E}_{\:P}\big(\mathsf{1}_{\:J}(\tilde{x})\big)\Big) \\
&\hspace{2em}+\sum_{j=1}^{J-1}\big(\mathsf{c}_{\:2j-1}(\theta)-\mathsf{c}_{\:2j+1}(\theta)\big)\cdot\Big(\bar{\mathsf{I}}_{\:j}\big(X(w,n)\big)-\mathbb{E}_{\:P}\big(\mathsf{1}_{\:j}(\tilde{x})\big)\Big).
\end{align*}
Thus, substituting these expressions of partial-sums in the previous equation, we may write
\begin{align*}
&\hspace{-2em}\hat{F}_{\:J}\big(X(w,n),\theta\big)-\mathbb{E}_{\:P}\big(\mathsf{1}_{\:J}(\tilde{x},\theta)\big) \\
&=r_{\:2J-1}\big(X(w,n)\big)\cdot\left(\frac{1}{2}-\mathsf{1}_{\:J}(\theta)\right) \\
&\hspace{2em}+\sum_{j=1}^{J-1}\Big(r_{\:2j-1}\big(X(w,n)\big)-r_{\:2j+1}\big(X(w,n)\big)\Big)\cdot\left(\frac{1}{2}-\mathsf{1}_{\:j}(\theta)\right) \\
&\hspace{4em}+\mathsf{c}_{\:2J-1}(\theta)\cdot\Big(\bar{\mathsf{I}}_{\:J}\big(X(w,n)\big)-\mathbb{E}_{\:P}\big(\mathsf{1}_{\:J}(\tilde{x})\big)\Big) \\
&\hspace{6em}+\sum_{j=1}^{J-1}\big(\mathsf{c}_{\:2j-1}(\theta)-\mathsf{c}_{\:2j+1}(\theta)\big)\cdot\Big(\bar{\mathsf{I}}_{\:j}\big(X(w,n)\big)-\mathbb{E}_{\:P}\big(\mathsf{1}_{\:j}(\tilde{x})\big)\Big).
\end{align*}

It follows from part B of Lemma \ref{Quantile:lemma:indicator-uniform-convergence} that $\mathsf{1}_{\:j}(\theta)$ is uniformly bounded for $j\in\mathbb{N}$ and $\theta\in(0,1)$. So, there exist $M\in\mathbb{R}$ such that $\big|\:\nicefrac{1}{2}-\mathsf{1}_{\:j}(\theta)\:\big|\leq M$ for $j\in\mathbb{N}$ and $\theta\in(0,1)$. Then, we have:
\begin{align*}
&\hspace{-2em}\Big|\:\hat{F}_{\:J}\big(X(w,n),\theta\big)-\mathbb{E}_{\:P}\big(\mathsf{1}_{\:J}(\tilde{x},\theta)\big)\:\Big| \\
&\leq\left|\:r_{\:2J-1}\big(X(w,n)\big)\:\right|\cdot\left|\:\frac{1}{2}-\mathsf{1}_{\:J}(\theta)\:\right| \\
&\hspace{2em}+\sum_{j=1}^{J-1}\Big|\:r_{\:2j-1}\big(X(w,n)\big)-r_{\:2j+1}\big(X(w,n)\big)\:\Big|\cdot\bigg|\:\frac{1}{2}-\mathsf{1}_{\:j}(\theta)\:\bigg| \\
&\hspace{4em}+\big|\:\mathsf{c}_{\:2J-1}(\theta)\:\big|\cdot\Big|\:\bar{\mathsf{I}}_{\:J}\big(X(w,n)\big)-\mathbb{E}_{\:P}\big(\mathsf{1}_{\:J}(\tilde{x})\big)\:\Big| \\
&\hspace{6em}+\sum_{j=1}^{J-1}\big|\:\mathsf{c}_{\:2j-1}(\theta)-\mathsf{c}_{\:2j+1}(\theta)\:\big|\cdot\Big|\:\bar{\mathsf{I}}_{\:j}\big(X(w,n)\big)-\mathbb{E}_{\:P}\big(\mathsf{1}_{\:j}(\tilde{x})\big)\:\Big| \\
&\leq\big|\:r_{\:2J-1}\big(X(w,n)\big)\:\big|\cdot M \\
&\hspace{2em}+\sum_{j=1}^{J-1}\Big(\big|\:r_{\:2j-1}\big(X(w,n)\big)\:\big|+\big|\:r_{\:2j+1}\big(X(w,n)\big)\:\big|\Big)\cdot M \\
&\hspace{4em}+\Big|\:\bar{\mathsf{I}}_{\:J}\big(X(w,n)\big)-\mathbb{E}_{\:P}\big(\mathsf{1}_{\:J}(\tilde{x})\big)\:\Big|+\sum_{j=1}^{J-1}2\cdot\Big|\:\bar{\mathsf{I}}_{\:j}\big(X(w,n)\big)-\mathbb{E}_{\:P}\big(\mathsf{1}_{\:j}(\tilde{x})\big)\:\Big| \\
&<\sum_{j=1}^J 2M\cdot\big|\:r_{\:2j-1}\big(X(w,n)\big)\:\big|+\sum_{j=1}^J2\cdot\Big|\:\bar{\mathsf{I}}_{\:j}\big(X(w,n)\big)-\mathbb{E}_{\:P}\big(\mathsf{1}_{\:j}(\tilde{x})\big)\:\Big| \\
&=\sum_{j=1}^J2M\cdot\Big|\:\bar{\mathsf{C}}_{\:2j-1}\big(X(w,n)\big)-\mathbb{E}_{\:P}\big(\mathsf{c}_{\:2j-1}(\tilde{x})\big)\:\Big|+\sum_{j=1}^J2\cdot\Big|\:\bar{\mathsf{I}}_{\:j}\big(X(w,n)\big)-\mathbb{E}_{\:P}\big(\mathsf{1}_{\:j}(\tilde{x})\big)\:\Big|.
\end{align*}
with, the second inequality follows as $|\mathsf{c}_{\:2J-1}(\theta)|\leq 1$ for $J\in\mathbb{N}$ and $|\mathsf{c}_{\:2j-1}(\theta)-\mathsf{c}_{\:2j+1}(\theta)|\leq|\mathsf{c}_{\:2j-1}(\theta)|+|\mathsf{c}_{\:2j+1}(\theta)|\leq 2$ for $j\in\mathbb{N}$.

Now, for $j\in\mathbb{N}$, pick $n_{\:w,\:2j-1}\in\mathbb{N}$ large enough so that for $n>n_{\:w,\:2j-1}$, we have
\begin{align*}
\Big|\:\bar{\mathsf{C}}_{\:2j-1}\big(X(w,n)\big)-\mathbb{E}_{\:P}\big(\mathsf{c}_{\:2j-1}(\tilde{x})\big)\:\Big|<\frac{\epsilon}{2M}\cdot 2^{-(2j-1)};
\end{align*}
also, pick $n_{\:w,\:2j}\in\mathbb{N}$ large enough so that for $n>n_{\:w,\:2j}$, we have 
\begin{align*}
\Big|\:\bar{\mathsf{I}}_{\:j}\big(X(w,n)\big)-\mathbb{E}_{\:P}\big(\mathsf{1}_{\:2j-1}(\tilde{x})\big)\:\Big|<\frac{\epsilon}{2}\cdot 2^{-(2j)}.
\end{align*}
Define $n_{\:w,\:(1:2J)}=\max(n_{\:w,\:1},n_{\:w,\:2},\dots,n_{\:w,\:2J})$. Then, for $n>n_{\:w,\:(1:2J)}$, we have
\begin{align*}
\Big|\:\hat{F}_{\:J}\big(X(w,n),\theta\big)-\mathbb{E}_{\:P}\big(\mathsf{1}_{\:J}(\tilde{x})\big)\:\Big|<\sum_{j=1}^{2J}\nicefrac{\epsilon}{2^j}<\epsilon.
\end{align*}
Observe that $n_{w,(1:2J)}$ does not depend on $\theta$, and so the above inequality holds uniformly for any $\theta\in(0,1)$. In other words, for $w\in N'$, if $n>n_{\:w,\:(1:2J)}$, then:
\begin{align*}
\underset{\:\theta\in(0,1)}{\sup}\Big|\:\hat{F}_{\:J}\big(X(w,n),\theta\big)-\mathbb{E}_{\:P}\big(\mathsf{1}_{\:J}(\tilde{x})\big)\:\Big|<\epsilon.
\end{align*}
Since $P(N)=0$, we conclude that $\underset{|X|\to\infty}{\lim}\underset{\theta\in(0,1)}{\sup}\Big|\:\hat{F}_{\:J}(X,\theta)-\mathbb{E}_{\:P}\big(\mathsf{1}_{\:J}(\tilde{x})\big)\:\Big|\overset{a.s.}{=}0$. \\

(B)\ Let $p\in(0,1)$, and observe that we may choose an $\epsilon_{\:p}>0$ such that $2*\epsilon_{\:p}<\operatorname{min}\big(p,1-p\big)$. Take $p^{\:(1)}\in(0,\epsilon_{\:p})$ and $p^{\:(2)}\in(1-\epsilon_{\:p},1)$. Recall from part B of Lemma \ref{Quantile:lemma:population-asymptotic}, that $\lim\limits_{J\rightarrow\infty}\mathcal{Q}_{\:J}^{\:(P)}=(0,1)$, where $\mathcal{Q}_{\:J}^{\:(P)}=\big\{p\colon\exists\ \theta\in(0,1)\text{ such that }\mathbb{E}_{\:P}\big(\mathsf{1}_{\:J}(\tilde{x},\theta)\big)=p\big\}$. So, there exists a $J_{\:p}$ such that if $J>J_{\:p}$, then we can find $q_{\:J,\:p^{\:(1)}}^{\:(P)}\in(0,1)$ such that $\mathbb{E}_{\:P}\Big(\mathsf{1}_{\:J}\big(\tilde{x},q_{\:J,\:p^{\:(1)}}^{\:(P)}\big)\Big)=p^{\:(1)}$, and $q_{\:J,\:p^{\:(2)}}^{\:(P)}\in(0,1)$ such that $\mathbb{E}_{\:P}\Big(\mathsf{1}_{\:J}\big(\tilde{x},q_{\:J,\:p^{\:(2)}}^{\:(P)}\big)\Big)=p^{\:(2)}$. 

Now by part A we know that there exists a null set $N$, such that if $w\in{N}'$, we can find an $n_{\:w,\:J}\in\mathbb{N}$ such that for $n>n_{\:w,\:J}$ and any $\theta\in(0,1)$, we have
\begin{align*}
\Big|\:\hat{F}_{\:J}\big(X(w,n),\theta\big)-\mathbb{E}_{\:P}\big(\mathsf{1}_{\:J}(\tilde{x},\theta)\big)\:\Big|\leq\epsilon_{\:p}.
\end{align*}
In particular, for $\theta=q_{\:J,\:p^{\:(1)}}^{\:(P)}$, if $n>n_{\:w,\:J}$, we have
\begin{align*}
&\quad\;\Big|\:\hat{F}_{\:J}\big(X(w,n),q_{\:J,\:p^{\:(1)}}^{\:(P)}\big)-\mathbb{E}_{\:P}\Big(\mathsf{1}_{\:J}\big(\tilde{x},q_{\:J,\:p^{\:(1)}}^{\:(P)}\big)\Big)\:\Big|\leq\epsilon_{\:p} \\
&\Leftrightarrow\Big|\:\hat{F}_{\:J}\big(X(w,n),q_{\:J,\:p^{\:(1)}}^{\:(P)}\big)-p^{\:(1)}\:\Big|\leq\epsilon_{\:p} \\
&\Rightarrow\:\hat{F}_{\:J}\big(X(w,n),q_{\:J,\:p^{\:(1)}}^{\:(P)}\big)\leq p^{\:(1)}+\epsilon_{\:p}<2\:\epsilon_{\:p}.
\end{align*}
Similarly, for $\theta=q_{\:J,\:p^{\:(2)}}^{\:(P)}$, if $n>n_{\:w,\:J}$, we have
\begin{align*}
&\quad\;\Big|\:\hat{F}_{\:J}\big(X(w,n),q_{\:J,\:p^{\:(2)}}^{\:(P)}\big)-\mathbb{E}_{\:P}\Big(\mathsf{1}_{\:J}\big(\tilde{x},q_{\:J,\:p^{\:(2)}}^{\:(P)}\big)\Big)\:\Big|\leq\epsilon_{\:p} \\
&\Leftrightarrow\Big|\:\hat{F}_{\:J}\big(X(w,n),q_{\:J,\:p^{\:(2)}}^{\:(P)}\big)-p^{\:(2)}\:\Big|\leq\epsilon_{\:p} \\
&\Rightarrow\:\hat{F}_{\:J}\big(X(w,n),q_{\:J,\:p^{\:(1)}}^{\:(P)}\big)\geq p^{\:(2)}-\epsilon_{\:p}>1-2\:\epsilon_{\:p}.
\end{align*}
It follows that if $n>n_{\:w,\:J}$, then
\begin{align*}
\hspace{-1em}0<\hat{F}_{\:J}\big(X(w,n),q_{\:J,\:p^{\:(1)}}^{\:(P)}\big)<2\:\epsilon_{\:p}<p<1-2\:\epsilon_{\:p}<\hat{F}_{\:J}\big(X(w,n),q_{\:J,\:p^{\:(2)}}^{\:(P)}\big)<1.
\end{align*}
Observe that $\hat{F}_{\:J}\big(X(w,n),\theta\big)$ is a continuous function of $\theta$. So, by the intermediate value theorem, we can find $\theta$ such that $\hat{F}_{\:J}\big(X(w,n),\theta\big)=p$. Recalling that $p\in(0,1)$ is arbitrary, we thus conclude that $p\in\hat{\mathcal{Q}}_{\:J}\big(X(w,n)\big)$ for any $p\in(0,1)$, as long as $n>n_{\:w,\:J}$, $w\in N'$, and $J>J_{\:p}$. Therefore, since $P(N)=0$, we conclude that $\lim\limits_{J\rightarrow\infty}\lim\limits_{|X|\rightarrow\infty}\hat{\mathcal{Q}}_{\:J}(X)\overset{a.s.}{=}(0,1)$. \\

(C)\ Given $p\in(0,1)$, let us consider arbitrary $\epsilon>0$. We know from part B of Lemma \ref{Quantile:lemma:population-asymptotic} that there exists $J_{\:1}\in\mathbb{N}$ such that $q_{\:J,\:p}^{\:(P)}$ exists for $J>J_{\:1}$, which implies that $\mathbb{E}_{\:P}\big(\mathsf{1}_{\:J}(\tilde{x},q_{\:J,\:p}^{\:(P)})\big)=p$. 

We also know from part B that there exists some $J_{\:2}\in\mathbb{N}$ such that, for $J>J_{\:2}$, we have a null set $N_{\:1,\:J}\subset \Omega$ such that, for $w\in N'_{\:1,\:J}$, there exists a $n_{\:w,\:J,\:1}\in\mathbb{N}$ such that $\hat{q}_{\:J,\:p}\big((X(w,n)\big)$ exists and $\hat{F}_{\:J}\Big(X(w,n),\hat{q}_{\:J,\:p}\big(X(w,n)\big)\Big)=p$. 

Next, part A implies that for any $J\in\mathbb{N}$, there exists a null set $N_{\:2,\:J}$ such that if $w\in N'_{\:2,\:J}$, then there exists $n_{\:w,\:J,\:2}$ such that for $n>n_{\:w,\:J,\:2}$, we will have that $\big|\:\hat{F}_{\:J}\big(X(w,n),\theta\big)-\mathbb{E}_{\:P}\big(\mathsf{1}_{\:J}(\tilde{x},\theta)\big)\:\big|<\nicefrac{\epsilon}{2}$ for any $\theta\in(0,1)$. 

Set $J_{\:0}=\max(J_{\:1},J_{\:2})$ and $N=\bigcup_{\:J>J_{\:0}}\big(N_{\:1,\:J}\cup N_{\:2,\:J}\big)$. Clearly $N$ is a null set. Now, for $w\in N'$ and $J>J_{\:0}$, define $n_{\:w,\:J}=\operatorname{max}(n_{\:w,\:J,\:1},n_{\:w,\:J,\:2})$. Then, for $n>n_{\:w,\:J}$, $w\in N'$ and $J>J_{\:0}$, and  we have for $\theta\in(0,1)$:
\begin{align*}
\big|\:\hat{F}_{\:J}\big(X(w,n),\theta\big)-\mathbb{E}_{\:P}\big(\mathsf{1}_{\:J}(\tilde{x},\theta)\big)\:\big|<\frac{\epsilon}{2}.
\end{align*}

So, we can take $\theta=\hat{q}_{\:J,\:p}\big(X(w,n),x)\big)$ and we will have that
\begin{align*}
\bigg|\:\hat{F}_{\:J}\Big(X(w,n),\hat{q}_{\:J,\:p}\big(X(w,n),x\big)\Big)-\mathbb{E}_{\:P}\bigg(\mathsf{1}_{\:J}\Big(\tilde{x},\hat{q}_{\:J,\:p}\big(X(w,n)\big)\Big)\bigg)\:\bigg|<\frac{\epsilon}{2}.
\end{align*}
Since $\hat{F}_{\:J}\Big(X(w,n),\hat{q}_{\:J,\:p}\big(X(w,n)\big)\Big)=p=\mathbb{E}_{\:P}\Big(\mathsf{1}_{\:J}\big(\tilde{x},q_{\:J,\:p}^{\:(P)}\big)\Big)$, we then have:
\begin{align*}
\Big|\:\mathbb{E}_{\:P}\Big(\mathsf{1}_{\:J}\big(\tilde{x},q_{\:J,\:p}^{\:(P)}\big)\Big)-\mathbb{E}_{\:P}\left(\mathsf{1}_{\:J}\Big(\tilde{x},\hat{q}_{\:J,\:p}\big(X(w,n)\big)\Big)\right)\:\Big|<\frac{\epsilon}{2}.
\end{align*}

Recall from the discussion following equation \ref{Quantile:proof:eq:c} that $\zeta^{\:(P)}_{\:J}(\theta)$ is the derivative of $\mathbb{E}_{\:P}\big(\mathsf{1}_{\:J}(\tilde{x},\theta)\big)$ with respect to $\theta$. Hence, by the fundamental theorem of calculus, we have:
\begin{align*}
\hspace{-1em}\mathbb{E}_{\:P}\Big(\mathsf{1}_{\:J}\big(\tilde{x},q_{\:J,\:p}^{\:(P)}\big)\Big)-\mathbb{E}_{\:P}\left(\mathsf{1}_{\:J}\Big(\tilde{x},\hat{q}_{\:J,\:p}\big(X(w,n)\big)\Big)\right)=\int_{\hat{q}_{\:J,\:p}\big(X(w,n)\big)}^{q_{\:J,\:p}^{\:(P)}}\zeta^{\:(P)}_{\:J}(\theta)\,d\theta.
\end{align*}
It thus follows that for $n>n_{\:w,\:J}$, $w\in N'$ and $J>J_{\:0}$, $\Big|\:\int_{\hat{q}_{\:J,\:\alpha}\big(X(w,n)\big)}^{q_{\:J,\:p}^{\:(P)}}\zeta^{\:(P)}_{\:J}(x,h)\,dh\:\big|<\frac{\epsilon}{2}$.

Recall that by Assumption \ref{Quantile:assumption:positive-density}, $F^{\:(P)}(\theta)$ has the derivative $f^{\:(P)}(\theta)$ for $\theta\in(0,1)$. Therefore, for $n>n_{\:w,\:J}$, $w\in N'$ and $J>J_{\:0}$, by the fundamental theorem of calculus, we also have:
\begin{align*}
&\hspace{-1em}\Big|\:F^{\:(P)}\big(q_{\:J,\:p}^{\:(P)}\big)-F^{\:(P)}\Big(\hat{q}_{\:J,\:p}\big(X(w,n)\big)\Big)\:\Big| \\
&=\Big|\:\int_{\hat{q}_{\:J,\:p}\big(X(w,n)\big)}^{q_{\:J,\:p}^{\:(P)}}f^{\:(P)}(\theta)\,d\theta\:\Big| \\
&\leq\Big|\:\int_{\hat{q}_{\:J,\:p}\big(X(w,n)\big)}^{q_{\:J,\:p}^{\:(P)}}\big(f^{\:(P)}(\theta)-\zeta^{\:(P)}_{\:J}(\theta)\big)\,d\theta\:\Big|+\Big|\:\int_{\hat{q}_{\:J,\:p}\big(X(w,n)\big)}^{q_{\:J,\:p}^{\:(P)}}\zeta^{\:(P)}_{\:J}(\theta)\,d\theta\:\Big| \\
&\leq\int_{\hat{q}_{\:J,\:p}\big(X(w,n)\big)}^{q_{\:J,\:p}^{\:(P)}}\big|\:f^{\:(P)}(\theta)-\zeta^{\:(P)}_{\:J}(\theta)\:\big|\,d\theta+\frac{\epsilon}{2} \\
&\leq\int_{\hat{q}_{\:J,\:p}\big(X(w,n)\big)}^{q_{\:J,\:p}^{\:(P)}}\rho^{\:(P)}_{\:J}\,d\theta+\frac{\epsilon}{2}=\rho_{\:J}^{\:(P)}\cdot\big|\:\hat{q}_{\:J,\:p}\big(X(w,n)\big)-q_{\:J,\:p}^{\:(P)}\:\big|+\frac{\epsilon}{2}.
\end{align*}
Since $\big\{\hat{q}_{\:J,\:p}\big(X(w,n)\big),q_{\:J,\:p}^{\:(P)}\big\}\in(0,1)$, we conclude from this expression that
\begin{align*}
\Big|\:F^{\:(P)}\big(q_{\:J,\:p}^{\:(P)}\big)-F^{\:(P)}\Big(\hat{q}_{\:J,\:p}\big(X(w,n)\big)\Big)\:\Big|<\rho_{\:J}^{\:(P)}+\frac{\epsilon}{2}.
\end{align*}

Now, suppose $\hat{q}_{\:J,\:p}\big(X(w,\cdot)\big)$ is a limit point with respect to the set $\hat{\mathcal{G}}_{\:J,\:p}(w)$. Then there exists a sub-sequence $\{n_{\:k}\}_{k\in\mathbb{N}}$ such that
\begin{align*} 
\lim\limits_{k\to\infty}\hat{q}_{\:J,\:p}\big(X(w,n_{\:k})\big)\to\hat{q}_{\:J,\:p}\big(X(w,\cdot)\big).
\end{align*}
By continuity of probability, it then follows directly that
\begin{align*}
F^{\:(P)}\Big(\hat{q}_{\:J,\:p}\big(X(w,n_{\:k})\big)\Big)\to F^{\:(P)}\Big(\hat{q}_{\:J,\:p}\big(X(w,\cdot)\big)\Big).
\end{align*}
and moreover there exists some $k_{\:w,\:J}\in\mathbb{N}$ such that for $k>k_{\:w,\:J}$ we have
\begin{align*}
\left|\:F^{\:(P)}\Big(\hat{q}_{\:J,\:p}\big(X(w,\cdot)\big)\Big)-F^{\:(P)}\Big(\hat{q}_{\:J,\:p}\big(X(w,n_{\:k})\big)\Big)\:\right|<\frac{\epsilon}{2}.
\end{align*}

Hence if we choose $k$ such that $k>k_{\:w,\:J}$ and $n_k>\max(n_{\:w,\:J,\:1},n_{\:w,\:J,\:2})$, then
\begin{align*}
&\hspace{-1em}\Big|\:F^{\:(P)}\Big(\hat{q}_{\:J,\:p}\big(X(w,\cdot)\big)\Big)-F^{\:(P)}\big(q_{\:J,\:p}^{\:(P)}\big)\:\Big| \\
&\leq\Big|\:F^{\:(P)}\Big(\hat{q}_{\:J,\:p}\big(X(w,\cdot)\big)\Big)-F^{\:(P)}\Big(\hat{q}_{\:J,\:p}\big(X(w,n_{\:k})\big)\Big)\:\Big| \\
&\hspace{-1em}+\Big|\:F^{\:(P)}\big(q_{\:J,\:p}^{\:(P)}\big)-F^{\:(P)}\Big(\hat{q}_{\:J,\:p}\big(X(w,n_{\:k})\big)\Big)\:\Big| \\
&<\ \Big(\frac{\epsilon}{2}+\rho_{\:J}^{\:(P)}\Big)+\frac{\epsilon}{2}.
\end{align*}
Finally, since $\epsilon$ is arbitrary, we must then have the claimed result that for any $w\in{N}'$
\begin{align*}
\Big|\:F^{\:(P)}\Big(\hat{q}_{\:J,\:p}\big(X(w,\cdot)\big)\Big)-F^{\:(P)}\big(q_{\:J,\:p}^{\:(P)}\big)\:\Big|\leq\rho_{\:J}^{\:(P)}.
\end{align*}
\end{proof}

\section{Supplementary material to complement section 4.1.2:}
As a first guess, we may be interested to know how our algorithm perform against Binning for a Normally distributed data. We consider Normal random-variable $\tilde{x}$ with mean $0$ and variance $1$, having distribution function $\Phi$. In this case, the data is random-permutation of $\{\Phi^{-1}(\nicefrac{1}{N}),\Phi^{-1}(\nicefrac{2}{N}),\dots,\Phi^{-1}(1-\nicefrac{1}{N})\}$. Remember a $\operatorname{Normal}(0,1)$ distribution has support $(-\infty,\infty)$. So, we can't directly apply Theorem \ref{Quantile:thm:_prop} for this data, because Assumptions \ref{Quantile:assumption:positive-density} and \ref{Quantile:assumption:uniform-convergence} requires the underlying random-variable to have a support of $(0,1)$.

We may consider to scale the data to $(0,1)$, first we pick a constant $M>\Phi^{-1}(1-\nicefrac{1}{N})\approx 6.174$, also ensuring $-M<\Phi^{-1}(\nicefrac{1}{N})$. Then, with the transformation: $\tilde{x}'\rightarrow\nicefrac{(\tilde{x}+M)}{2M}$, we can ensure that all the observations are inside $(0,1)$. The transformed variable can be identified as the Normal random-variable $\tilde{x}'$ with mean $0.5$ and variance $\nicefrac{1}{2M}$. Now Beta distribution is a natural choice for a distribution function with support $(0,1)$, and we may justify picking the parameters of a Beta distribution that resembles a $\operatorname{Normal}(0.5,\nicefrac{1}{2M})$ distribution, in fact, $\operatorname{Beta}(19,19)$ distribution has close first two moments ( mean = $0.5$, std = $\nicefrac{1}{2\sqrt{39}}$) if we restrict our choice of Beta parameters to integers. In this case, our generated data are a random permutation of the numbers ${F^{\:(P)}}^{\text{-}1}(\nicefrac{1}{N})$, ${F^{\:(P)}}^{\text{-}1}(\nicefrac{2}{N})$, $\dots$, ${F^{\:(P)}}^{\text{-}1}(\nicefrac{N-1}{N})$. Here ${F^{\:(P)}}$ is the cumulative distribution function of a $\operatorname{Beta(19,19)}$ random-variable with domain $(0,1)$. 

\subsection{Assumption \ref{Quantile:assumption:uniform-convergence} for \texttt{Beta(19,19)} distribution:}\label{sec:verify_beta(19,19)}
The data $X$ is a random permutation of $\operatorname{Beta}(19,19)$ population quantiles: $F^{-1}(\nicefrac{1}{N}),\dots,F^{-1}(\nicefrac{N-1}{N})$, here $F(t)=\int_0^tf(s)ds$, where:
\begin{align*}
f(t)=\frac{1}{\mathcal{B}(19,19)}\cdot t^{18}\cdot(1-t)^{18}\text{ for }t\in(0,1).
\end{align*}
For a random variable $\tilde{x}\sim\operatorname{Beta}(19,19)$ and arbitrary $x\in(0,1)$ we have the following:
\begin{align*}
&\hspace{-2em}\:\mathbb{E}_{\:P}\Big(\cos\big((2j-1)(\tilde{x}-x)\big)\Big)\\
&=\mathbb{E}_{\:P}\Big(\cos\big((2j-1)\tilde{x}\big)\Big)\cdot\cos\big((2j-1)x\big) \\ &\hspace{2em}+\mathbb{E}_{\:P}\Big(\sin\big((2j-1)\tilde{x}\big)\Big)\cdot\sin\big((2j-1)x\big) \\
&=\frac{\cos\big((2j-1)x\big)}{\mathcal{B}(19,19)}\cdot\int_0^1t^{18}\cdot(1-t)^{18}\cdot\cos\big((2j-1)t\big)dt \\
&\hspace{2em}+\frac{\sin\big((2j-1)x\big)}{\mathcal{B}(19,19)}\cdot\int_0^1t^{18}\cdot(1-t)^{18}\cdot\sin\big((2j-1)t\big)dt.
\end{align*}
By Cauchy-Schwartz inequality on $L^2$, we have:
\begin{align*}
&\hspace{-2em}\Big|\:\int_0^1t^{18}\cdot(1-t)^{18}\cdot e^{-\mathrm{i}(2j-1)t}\:dt\:\Big|^2 \\
&=\Big|\:\int_0^1\big(t^{17}\cdot(1-t)^{17}\Big)\cdot \Big(\overline{t\cdot(1-t)\cdot e^{\mathrm{i}(2j-1)t}}\Big)\:dt\:\Big|^2 \\
&\leq\Big|\:\int_0^1t^{34}\cdot(1-t)^{34}dt\:\Big|\cdot\Big|\:\int_0^1t^2\cdot(1-t)^2\cdot e^{\mathrm{i}\big(2(2j-1)t\big)}\:dt\:\Big| \\
&=\Big|\:\int_0^1t^{34}\cdot(1-t)^{34}dt\:\Big|\cdot\Big|\:\int_0^1t^2\cdot(1-t)^2\cdot \cos\big((4j-2)t\big)\:dt \\
&\hspace{2.5in}+\mathrm{i}\int_0^1t^2\cdot(1-t)^2\cdot \sin\big((4j-2)t\big)\:dt\:\Big|
\end{align*}
By AM-GM inequality, we have $t\cdot(1-t)\leq\nicefrac{1}{4}$ for $t\in(0,1)$. Hence:
\begin{align*}
\Big|\:\int_0^1t^{34}\cdot(1-t)^{34}dt\:\Big|\leq\nicefrac{1}{4^{34}}.
\end{align*}
Also we can easily verify the following identities:
\begin{align*}
&\int_0^1t^2\cdot(1-t)^2\cdot\cos\big((4j-2)t\big)dt=-\frac{\sin(4j-2)}{4(2j-1)^3}-\frac{3\big(1+\cos(4j-2)\big)}{4(2j-1)^4}+\frac{3\sin(4j-2)}{4(2j-1)^5}. \\
&\int_0^1t^2\cdot(1-t)^2\cdot\sin\big((4j-2)t\big)dt=\frac{1-\cos(4j-2)}{4(2j-1)^3}-\frac{3\sin(4j-2)}{4(2j-1)^4}-\frac{3\big(1-\cos(4j-2)\big)}{4(2j-1)^5}.
\end{align*}
In other words:
\begin{align*}
&\int_0^1t^2\cdot(1-t)^2\cdot\cos\big((4j-2)t\big)dt=O\Big(\nicefrac{1}{(2j-1)^3}\Big),\text{ and} \\ &\int_0^1t^2\cdot(1-t)^2\cdot\sin\big((4j-2)t\big)dt=O\Big(\nicefrac{1}{(2j-1)^3}\Big).
\end{align*}
Then, from the previous equation, we have a constant $k_{\:0}$, such that:
\begin{align*}
&\quad\Big|\:\int_0^1t^{18}\cdot(1-t)^{18}\cdot e^{-\mathrm{i}(2j-1)t}\:dt\:\Big|^2\leq \frac{k_{\:0}}{(2j-1)^3} \\
&\Leftrightarrow \Big|\:\int_0^1t^{18}\cdot(1-t)^{18}\cdot e^{-\mathrm{i}(2j-1)t}\:dt\:\Big|\leq \frac{\sqrt{k_{\:0}}}{(2j-1)^{\nicefrac{3}{2}}} \\
&\Leftrightarrow \Big|\:\int_0^1t^{18}\cdot(1-t)^{18}\cdot \cos\big((2j-1)t\big)\:dt\:\Big|\leq \frac{\sqrt{k_{\:0}}}{(2j-1)^{\nicefrac{3}{2}}} \text{ and,}\\
&\quad\: \Big|\:\int_0^1t^{18}\cdot(1-t)^{18}\cdot \sin\big((2j-1)t\big)\:dt\:\Big|\leq \frac{\sqrt{k_{\:0}}}{(2j-1)^{\nicefrac{3}{2}}}.
\end{align*}
which, in turn, implies:
\begin{align*}
&\hspace{-2em}\:\Big|\:\mathbb{E}_{\:P}\Big(\cos\big((2j-1)(\tilde{x}-x)\big)\Big)\:\Big|\\
&\leq\frac{\big|\:\cos\big((2j-1)x\big)\:\big|}{\mathcal{B}(19,19)}\cdot\Big|\:\int_0^1t^{18}\cdot(1-t)^{18}\cdot\cos\big((2j-1)t\big)dt\:\Big| \\
&\hspace{2em}+\frac{\big|\:\sin\big((2j-1)x\big)\:\big|}{\mathcal{B}(19,19)}\cdot\Big|\:\int_0^1t^{18}\cdot(1-t)^{18}\cdot\sin\big((2j-1)t\big)dt\:\Big| \\
&\leq\frac{1}{\mathcal{B}(19,19)}\cdot\frac{\sqrt{k_{\:0}}}{(2j-1)^{\nicefrac{3}{2}}}+\frac{1}{\mathcal{B}(19,19)}\cdot\frac{\sqrt{k_{\:0}}}{(2j-1)^{\nicefrac{3}{2}}} \\
&\Leftrightarrow \Big|\:\mathbb{E}_{\:P}\Big(\cos\big((2j-1)(\tilde{x}-x)\big)\Big)\:\Big|=O\Big(\nicefrac{1}{(2j-1)^{\nicefrac{3}{2}}}\Big).
\end{align*}
By an application of Weierstrass's M-test, we can conclude that the sequence of partial-sums: 

\begin{align*}
\mathbb{E}_P\Big(\frac{\sin\big(2J(\tilde{x}-x)\big)}{\pi\sin(\tilde{x}-x)}\Big)=\sum_{j=1}^J\mathbb{E}_P\Big(\cos\big((2j-1)(\tilde{x}-x)\big)\Big)
\end{align*}
uniformly converges to a limit function $g(x)$. From \citet[Theorem~9.13]{Book_1957_Apostol}, we must have $g(x)=f(x)$: the density of $\operatorname{Beta}(19,19)$ distribution. And so Assumption \ref{Quantile:assumption:positive-density} and Assumption \ref{Quantile:assumption:uniform-convergence} are satisfied, and we can expect $\hat{q}_{\:J,\:p}(X)$ to approach  $\hat{q}_{\:p}(X)$ as $|X|$ and then $J$ grow large.

\subsection{Assumption \ref{Quantile:assumption:uniform-convergence} for \texttt{Uniform(0,1)} distribution:}\label{sec:verify_uniform(0,1)}
Here data $X$ is a permutation of the numbers $\nicefrac{1}{N}$, $\nicefrac{2}{N}$, $\dots$, $\nicefrac{N-1}{N}$. For the Uniform distribution on the unit interval, we have the expression for the uniform density function: $f(x)=\mathsf{1}(0<x<1)=\mathsf{1}(0<x)-\mathsf{1}(1\leq x)$. Now, observe that:
\begin{align*}
\mathbb{E}_P\Big(\cos\big((2j-1)(\tilde{x}-x)\big)\Big)&=\int_{0}^{1}\cos\big((2j-1)(t-x)\big)~dt \\
&=\frac{\sin\big((2j-1)(1-x)\big)}{2j-1}-\frac{\sin\big((2j-1)(0-x)\big)}{2j-1}
\end{align*}
Then, we have:
\begin{align*}
\mathbb{E}_P\Big(\frac{2\sin\big(2J(\tilde{x}-x)\big)}{\pi\sin(\tilde{x}-x)}\Big)&=\sum_{j=1}^J\mathbb{E}_P\Big(\cos\big((2j-1)(\tilde{x}-x)\big)\Big) \\
&=\sum_{j=1}^J\Big(\frac{\sin\big((2j-1)(1-x)\big)}{2j-1}-\frac{\sin\big((2j-1)(0-x)\big)}{2j-1}\Big).
\end{align*}
Now, the uniform density function $f(x)=\mathsf{1}(0<x)-\mathsf{1}(1\leq x)$ admits the Fourier series expansion $\sum_{j=1}^{\infty}\Big(\frac{\sin\big((2j-1)(1-x)\big)}{2j-1}-\frac{\sin\big((2j-1)(0-x)\big)}{2j-1}\Big)$, point-wise converging to its limit in $(0,1)$. It also follows that $\mathbb{E}_P\Big(\frac{2\sin\big(2J(\tilde{x}-x)\big)}{\pi\sin(\tilde{x}-x)}\Big)$ converges uniformly to uniform-density function for any fixed interior region of the unit interval (see Lemma \ref{Quantile:lemma:indicator-uniform-convergence} part A in Appendix \ref{proof:eps:prop:SOT-FAS-WEP}).

\section{Proof of Theorem \ref{Lowess:thm:prop}}\label{proof:Lowess:thm:strong_consistency}

We begin by recalling that in Theorem \ref{Lowess:thm:prop}, $\hat{h}_{\:J,\:\alpha,\:x}$ is defined to be a solution for $h$ to the implicit equation given by $\hat{F}_{\:J,\:x}(X,h)=\alpha$ (in contrast to the setting employed in Theorem \ref{Quantile:thm:_prop}). While such a solution may fail to exist for arbitrary values of $\alpha\in(0,1)$ and $J$, note that $\hat{F}_{\:J,\:x}(X,0)=0$, and that in practice $\alpha$ will typically be chosen close to $0$. Therefore, since $\hat{F}_{\:J,\:x}(X,h)$ is continuous in $h$, we can expect that $\hat{h}_{\:J,\:\alpha,\:x}(X)$ may exist even for small values of $J$; indeed, Lemma \ref{Lowess:lemma:sample_asymptotic} below guarantees that $\hat{h}_{\:J,\:\alpha,\:x}(X)$ both exists and is eventually unique almost surely $[P]$ for any fixed $\alpha\in(0,1)$ as $J \to \infty$.

\subsection{Auxiliary results to prove Theorem \ref{Lowess:thm:prop}}\label{sec:Lowess:proof:aux}

Recall that for our given probability triple $\big(\Omega,\mathcal{F},P\big)$, we have the random variable $\tilde{x}:\Omega\to(0,1)$. Given $x\in(0,1)$ and the parameter $h\in\Theta$, with $\Theta=(0,1)$ being the parameter space, let $F_{\:x}^{\:(P)}(h):=P(x-h\leq\tilde{x}\leq x+h)$. If Assumption \ref{Lowess:assumption:positive-density} holds, then $F_{\:x}^{\:(P)}(h)=\int_{x-h}^{x+h}f^{\:(P)}(t)\,dt$ for $h\in(0,1)$, and observe that $F_{\:x}^{\:(P)}(h)$ has derivative $f_{\:x}^{\:(P)}(h)=f^{\:(P)}(x-h)+f^{\:(P)}(x+h)$. 

First, we have the following two results:

\begin{lemma}\label{Lowess:lemma:inverse_continuity}
Suppose Assumption \ref{Lowess:assumption:positive-density} holds. Then: \\
(A)\ Given $x$, let $h_{\:n}\in(0,1)$ for $n\in\mathbb{N}$, and $h\in(0,1)$, such that $\lim\limits_{n\to\infty}F_{\:x}^{\:(P)}(h_{\:n})\to F_{\:x}^{\:(P)}(h)$. Then $\lim\limits_{n\to\infty}h_{\:n}\to h$. \\
(B)\ For any fixed $\alpha \in (0,1)$, there exists a unique solution $h^{\:(P)}_{\:\alpha,\:x}\in(0,1)$ for $h$ to the implicit equation $F_{\:x}^{\:(P)}(h)=\alpha$.
\end{lemma}

\begin{proof}
(A)\ We have $F_{\:x}^{\:(P)}(h)=\int_{x-h}^{x+h}f^{\:(P)}(t)\,dt$, since by Assumption \ref{Lowess:assumption:positive-density}, $f^{\:(P)}>0\:[P]$ in $(0,1)$, and so we must have $F_{\:x}^{\:(P)}(h)$ continuous and strictly monotone for $h\in(0,1)$. Then part A follows from the continuity of the inverse of a continuous function.

(B)\ Since $F_{\:x}^{\:(P)}(h)=P(x-h\leq\tilde{x}\leq x+h)$, we have $\lim\limits_{h\to 0}F_{\:x}^{\:(P)}(h)=0$. Also, since $(0,1)\subseteq(x-1,x+1)$ for $x\in\mathbb{X}$, we have $F_{\:x}^{\:(P)}(1)=1$. Given $\alpha\in(0,1)$, since $F_{\:x}^{\:(P)}$ is continuous, by the intermediate value theorem there exists at least one value $h^{\:(P)}_{\:\alpha,\:x}\in(0,1)$, such that $F_{\:x}\big(h^{\:(P)}_{\:\alpha,\:x}\big)=\alpha$. To show the unicity of $h^{\:(P)}_{\:\alpha,\:x}$, assume to the contrary that there exist two distinct values $h^{\:(P)}_{\:\alpha,\:x,\:1} < h^{\:(P)}_{\:\alpha,\:x,\:2}$ satisfying $F_{\:x}^{\:(P)}(h)=\alpha$. Then we have:
\begin{align*}
F_{\:x}\big(h^{\:(P)}_{\:\alpha,\:x,\:1}\big)-F_{\:x}\big(h^{\:(P)}_{\:\alpha,\:x,\:2}\big)=\alpha-\alpha=0\:\Leftrightarrow\:\int_{h^{\:(P)}_{\:\alpha,\:x,\:1}}^{h^{\:(P)}_{\:\alpha,\:x,\:2}}f_{\:x}^{\:(P)}(t)\,dt=0.
\end{align*}
This contradicts our assumption that $f^{\:(P)}>0\:[P]$, since $f^{\:(P)}>0\:[P]\Leftrightarrow f_{\:x}^{\:(P)}>0\:[P]$ for $x\in(0,1)$.
\end{proof}

\begin{lemma}\label{Lowess:lemma:exact_convergence}
Suppose Assumption \ref{Lowess:assumption:positive-density} holds. Then: \\
(A)\  $\hat{h}_{\:\alpha,\:x}(X)$ exists almost surely $[P]$. \\
(B)\  Given $x\in(0,1)$, for any fixed $\alpha\in(0,1)$, $\lim\limits_{|X|\to\infty}\hat{h}_{\:\alpha,\:x}(X)\overset{a.s.}{=}h^{\:(P)}_{\:\alpha,\:x}$.
\end{lemma}

\begin{proof}
(A)\ If Assumption \ref{Lowess:assumption:positive-density} holds, by part A of Lemma \ref{Quantile:lemma:exact-convergence}, we know that $X$ has all distinct elements almost surely $[P]$. So we know in turn that $X$ has at least $\big\lceil\:\alpha\cdot|X|\:\big\rceil+1$ distinct elements almost surely $[P]$. So, given $x\in(0,1)$, we will be able to find at least $\big\lceil\:\alpha\cdot|X|\:\big\rceil$ elements in $(0,1)$ different than $x$ almost surely $[P]$. Thus we will be able to find $h\in(0,1)$ such that $\hat{F}_{\:x}(X,h^{\text{-}})\leq\alpha\leq\hat{F}_{\:x}(X,h)$.

(B)\ Let $\epsilon>0$ be given. By the Glivenko--Cantelli theorem, we have the result that  $\sup_{\theta\in(0,1)}\big|\:\hat{F}_{\:x}(X,h)-F_{\:x}^{\:(P)}(h)\:\big|\overset{a.s.}{=}0$. In particular for $h=\hat{h}_{\:\alpha,\:x}(X)$, we also have $\big|\:\hat{F}_{\:x}\big(X,\hat{h}_{\:\alpha,\:x}(X)\big)-F_{\:x}^{\:(P)}\big(\hat{h}_{\:\alpha,\:x}(X)\big)\:\big|\overset{a.s.}{=}0$. Hence, there exists a null set $N_{\:1}\subset\Omega$, such that for each $w\in{N_{\:1}}'$, there exists an $n_{\:w,\:1}$ such that for $n>n_{\:w,\:1}$ we have that $\Big|\:\hat{F}_{\:x}\Big(X(w,n),\hat{h}_{\:\alpha,\:x}\big(X(w,n)\big)\Big)-F^{\:(P)}\Big(\hat{h}_{\:\alpha,\:x}\big(X(w,n)\big)\Big)\:\Big|<\nicefrac{\epsilon}{2}$.

Now, from our definition of $\hat{h}_{\:\alpha,\:x}(X)$, we know that $\hat{h}_{\:\alpha,\:x}(X)$ satisfies the expression $\hat{F}_{\:x}\big(X,\hat{h}_{\:\alpha,\:x}(X)^{\text{-}}\big)\leq \alpha\leq \hat{F}_{\:x}\big(X,\hat{h}_{\:\alpha,\:x}(X)\big)$.
Observe that if the elements of $X$ are distinct, then for any $h\in(0,1)$, we have $\big|\hat{F}_{\:x}(X,h^{\text{-}})-\hat{F}_{\:x}(X,h)\big|\leq\nicefrac{2}{|X|}$. We can take $\theta=\hat{h}_{\:\alpha,\:x}(X)$ and consequently, if elements of $X$ are distinct, we have $\big|\hat{F}_{\:x}\big(X,\hat{h}_{\:\alpha,\:x}(X)\big)-\alpha\big|\leq\nicefrac{2}{|X|}$.
So, from part A we can conclude that there exists a null set $N_{\:2}\subset\Omega$, such that for each $w\in{N_{\:2}}'$ and any $n\in\mathbb{N}$, we have $\Big|\hat{F}_{\:x}\Big(X(w,n),\hat{h}_{\:\alpha,\:x}\big(X(w,n)\big)\Big)-\alpha\Big|\leq\nicefrac{2}{n}$. Clearly $N:=N_{\:1}\cup N_{\:2}$ is such a null set, and so if $w\in{N}'$, then by taking $n_{\:w}>\max(n_{\:w,\:1},\nicefrac{4}{\epsilon})$ we have:
\begin{align*}
&\hspace{-2em}\Big|\:F_{\:x}^{\:(P)}\Big(\hat{h}_{\:\alpha,\:x}\big(X(w,n)\big)\Big)-\alpha\:\Big| \\
&\leq\Big|\:F_{\:x}^{\:(P)}\Big(\hat{h}_{\:\alpha,\:x}\big(X(w,n)\big)\Big)-\hat{F}_{\:x}\Big(X(w,n),\hat{h}_{\:\alpha,\:x}\big(X(w,n)\big)\Big)\:\Big| \\
&\hspace{2em}+\Big|\:\hat{F}_{\:x}\Big(X(w,n),\hat{h}_{\:\alpha,\:x}\big(X(w,n)\big)\Big)-\alpha\:\Big| \\
&<\nicefrac{\epsilon}{2}+\nicefrac{2}{n},
\end{align*}
which sum is less than or equal to $\epsilon$ for $n$ sufficiently large; i.e., $n>n_{\:w}$. Now, from part B of  Lemma \ref{Lowess:lemma:inverse_continuity}, we have a unique solution $h^{\:(P)}_{\:\alpha,\:x}\in(0,1)$, to $F_{\:x}^{\:(P)}(h)=\alpha$ for $h$. Then from the above inequality, for the null set $N\in\Omega$, if $w\in{N}'$ and $n>n_{\:w}$, we have that $\Big|\:F_{\:x}^{\:(P)}\Big(\hat{h}_{\:\alpha,\:x}\big(X(w,n)\big)\Big)-F_{\:x}^{\:(P)}\big(h^{\:(P)}_{\:\alpha,\:x}\big)\big|<\epsilon$. Since $\epsilon$ is arbitrary, we must have $\lim\limits_{n\to\infty}F_{\:x}^{\:(P)}\Big(\hat{h}_{\:\alpha,\:x}\big(X(w,n)\big)\Big)=F_{\:x}^{\:(P)}\big(h^{\:(P)}_{\:\alpha,\:x}\big)$. From part A of Lemma \ref{Lowess:lemma:inverse_continuity}, we conclude: $\lim\limits_{n\to\infty}\hat{h}_{\:\alpha,\:x}\big(X(w,n)\big)=h^{\:(P)}_{\:\alpha,\:x}$. This is true for any $w\in{N}'$, since $N\in\Omega$ is a null set, and hence we have proved the statement $\lim\limits_{|X|\to\infty}\hat{h}_{\:\alpha,\:x}(X)\overset{a.s.}{=}h^{\:(P)}_{\:\alpha,\:x}$.
\end{proof}

Now, recall that just prior to the statement of Assumption \ref{Lowess:assumption:positive-density}, we defined $\mathsf{1}_{\:J,\:x}(\tilde{x},h)$ to be a $J$-term Fourier approximation to $\mathsf{1}(x-h\leq\tilde{x}\leq x+h)$. Noting that
\begin{align*}
&\mathsf{1}(x-h<\tilde{x}<x+h)=\mathsf{1}(x-h<\tilde{x})-\mathsf{1}(x+h\leq\tilde{x}), \\
&\mathsf{1}(x-h\leq\tilde{x}\leq x+h)=\mathsf{1}(x-h\leq\tilde{x})-\mathsf{1}(x+h<\tilde{x});
\end{align*}
we see that $\mathsf{1}_{\:J,\:x}(\tilde{x},h)$ can be expressed in terms of the Fourier approximation $\mathsf{1}_{\:J}(\cdot,\cdot)$ defined in Example \ref{eps:eg:ind}:
\begin{align*}
\mathsf{1}_{\:J,\:x}(\tilde{x},h)=\mathsf{1}_{\:J}(\tilde{x},x-h)-\mathsf{1}_{\:J}(\tilde{x},x+h).
\end{align*}

Recalling the notation introduced at the beginning of section \ref{sec:General:MapReduce}, we have the following.

\begin{lemma}\label{Lowess:lemma:delta_expression}
For $J\in\mathbb{N}$, the following expressions hold:
\begin{align*}
&(A)\ \mathsf{1}_{\:J,\:x}(\tilde{x},h)=\frac{4}{\pi}\cdot \sum_{j=1}^J\big(\mathsf{c}_{\:2j-1}(\tilde{x}-x)\big)\frac{\mathsf{c}_{\:2j}(h)}{2j-1}, \\
&(B)\ \hat{F}_{\:J,\:x}(X,h)=\frac{4}{\pi}\cdot \sum_{j=1}^J\big(\bar{\mathsf{C}}_{\:2j-1}(X)\cdot\mathsf{c}_{\:2j-1}(x)+\bar{\mathsf{C}}_{\:2j}(X)\cdot\mathsf{c}_{\:2j}(x)\big)\cdot\frac{\mathsf{c}_{\:2j}(h)}{2j-1}.
\end{align*}
\end{lemma}

\begin{proof}
(A)\ Recall that $\mathsf{1}_{\:J,\:x}(\tilde{x},h)=\mathsf{1}_{\:J}(\tilde{x},x+h)-\mathsf{1}_{\:J}(\tilde{x},x-h)$. From Example \ref{eps:eg:ind}, we have that
\begin{align*}
\mathsf{1}_{\:J}(\tilde{x},\theta)=\frac{1}{2}-\frac{2}{\pi}\cdot \sum_{j=1}^J\frac{\sin\big((2j-1)\cdot(\tilde{x}-\theta)\big)}{2j-1}.
\end{align*}
Substituting $x+h$ and $x-h$ respectively for $\theta$, we may then write:
\begin{align*}
&\hspace{-2em}\mathsf{1}_{\:J}(\tilde{x},x+h)-\mathsf{1}_{\:J}(\tilde{x},x-h) \\
&=\mathsf{1}_{\:J}(\tilde{x}-x-h)-\mathsf{1}_{\:J}(\tilde{x}-x+h) \\
&=\frac{2}{\pi}\cdot \sum_{j=1}^J\frac{-\sin\big((2j-1)(\tilde{x}-x-h)\big)+\sin\big((2j-1)(\tilde{x}-x+h)\big)}{2j-1} \\
&=\frac{2}{\pi}\cdot \sum_{j=1}^J\frac{2\cos\big((2j-1)(\tilde{x}-x)\big)\cdot\sin\big((2j-1)h\big)}{2j-1} \\
&=\frac{4}{\pi}\cdot \sum_{j=1}^J\mathsf{c}_{\:2j-1}(x_{\:i}-x)\cdot\frac{\mathsf{c}_{\:2j}(h)}{2j-1},
\end{align*}
which proves our statement. \\

(B)\ If we evaluate $\mathsf{1}_{\:J,\:x}(\tilde{x},h)$ from part A at $\tilde{x}=x_{\:i}$ and then sum over $i\in I$, we obtain:
\begin{align*}
\sum_{i\in I}\mathsf{1}_{\:J,\:x}(x_{\:i},h)&=\sum_{i\in I}\bigg(\frac{4}{\pi}\cdot \sum_{j=1}^J\mathsf{c}_{\:2j-1}(x_{\:i}-x)\cdot \frac{\mathsf{c}_{\:2j}(h)}{2j-1}\bigg) \\
&=\frac{4}{\pi}\cdot \sum_{j=1}^J\sum_{i\in I}\bigg(\mathsf{c}_{\:2j-1}(x_{\:i}-x)\cdot \frac{\mathsf{c}_{\:2j}(h)}{2j-1}\bigg) \\
&=\frac{4}{\pi}\cdot \sum_{j=1}^J\bigg(\Big(\sum_{i\in I}\mathsf{c}_{\:2j-1}(x_{\:i})\Big)\cdot \mathsf{c}_{\:2j-1}(x)+\Big(\sum_{i\in I}\mathsf{c}_{\:2j}(x_{\:i})\Big)\cdot \mathsf{c}_{\:2j}(x)\bigg)\cdot\frac{\mathsf{c}_{\:2j}(h)}{2j-1}.
\end{align*}
If we then divide both sides of this expression by $|X|$, we obtain:
\begin{align*}
&\hspace{-2em}\frac{1}{|X|}\cdot \sum_{i\in I}\mathsf{1}_{\:J,\:x}(x_{\:i},h) \\
&=\frac{4}{\pi}\cdot \sum_{j=1}^J\left(\frac{\sum_{i\in I}\mathsf{c}_{\:2j-1}(x_{\:i})}{|X|}\cdot \mathsf{c}_{\:2j-1}(x)+\frac{\sum_{i\in I}\mathsf{c}_{\:2j}(x_{\:i})}{|X|}\cdot \mathsf{c}_{\:2j}(x)\right)\cdot \frac{\mathsf{c}_{\:2j}(h)}{2j-1} \\
&\Leftrightarrow\hat{F}_{\:J,\:x}(X,h)=\frac{4}{\pi}\cdot \sum_{j=1}^J\Big(\bar{\mathsf{C}}_{\:2j-1}(X)\cdot \mathsf{c}_{\:2j-1}(x)+\bar{\mathsf{C}}_{\:2j}(X)\cdot \mathsf{c}_{\:2j}(x)\Big)\cdot \frac{\mathsf{c}_{\:2j}(h)}{2j-1},
\end{align*}
which proves our statement.
\end{proof}

Given $x\in\mathbb{X}$, $J\in\mathbb{N}$, and $\alpha\in(0,1)$, define $h_{\:J,\:\alpha,\:x}^{\:(P)}$ as a solution to $\mathbb{E}_{\:P}\big(\mathsf{1}_{\:J,\:x}(\tilde{x},h)\big)=\alpha$ for $h\in(0,1)$, provided such a solution exists. Let us then introduce the sequence of sets $\mathcal{A}_{\:J,\:x}^{\:(P)}$ for $J\in\mathbb{N}$ as follows:
\begin{align*}
\mathcal{A}_{\:J,\:x}^{\:(P)}:=\big\{\alpha\colon\exists\ h\in(0,1)\text{ such that }\mathbb{E}_{\:P}\big(\mathsf{1}_{\:J,\:x}(\tilde{x},h)\big)=\alpha\big\}. 
\end{align*}
Note that if $h_{\:J,\:\alpha,\:x}^{\:(P)}$ exists, it need not be unique; since $\mathbb{E}_{\:P}\big(\mathsf{1}_{\:J,\:x}(\tilde{x},h)\big)$ is weighted sum of periodic trigonometric functions in $h$, it is possible to have multiple solutions to $\mathbb{E}_{\:P}\big(\mathsf{1}_{\:J,\:x}(\tilde{x},h)\big)=\alpha$. Therefore we shall define the sets
\begin{align*}
\mathcal{H}_{\:J,\:\alpha,\:x}^{\:(P)}:=\big\{h\in(0,1)\colon\mathbb{E}_{\:P}\big(\mathsf{1}_{\:J,\:x}(\tilde{x},h)\big)=\alpha\big\}.
\end{align*}
Finally, for an arbitrary set $S\subset\mathbb{R}$, define $\bigtriangledown(S):=\max\big\{|h_{\:1}-h_{\:2}|\colon h_{\:1}\in S,h_{\:2}\in S\big\}$. 

We now show that under basic regulatory conditions, $h_{\:J,\:\alpha,\:x}^{\:(P)}$ exists and converges to the limit $h^{\:(P)}_{\:\alpha,\:x}$ when $J$ becomes large.

\begin{lemma}\label{Lowess:lemma:population_asymptotic}
Suppose Assumption \ref{Lowess:assumption:positive-density} holds. Then: \\
(A)\ $\lim\limits_{J\to\infty}\mathbb{E}_{\:P}\:\big(\mathsf{1}_{\:J,\:x}(\tilde{x},h)\big)\to F_{\:x}^{\:(P)}(h)$ uniformly in $x\in(0,1)$ and $h\in(0,1)$. \\
(B)\ $\lim\limits_{J\rightarrow\infty}\mathcal{A}_{\:J,\:x}^{\:(P)}=(0,1)$. \\
(C)\ $\lim\limits_{J\to\infty}h_{\:J,\:\alpha,\:x}^{\:(P)}\to h^{\:(P)}_{\:\alpha,\:x}$ and $\lim\limits_{J\to\infty}\bigtriangledown\big(\mathcal{H}_{\:J,\:\alpha,\:x}^{\:(P)}\big)\to 0$ for $\alpha\in(0,1)$ and $x\in\mathbb{X}$.
\end{lemma}
\begin{proof}
(A)\ Let $\epsilon>0$ and, given $x\in(0,1)$, fix an arbitrary $h\in(0,1)$. Since $\mathsf{1}_{\:J}(z)$ converges pointwise to both $\mathsf{1}(0<z)$ and $\mathsf{1}(0\leq z)$, for $z\in(-\pi,\pi)\setminus\{0\}$, we can conclude that $\mathsf{1}_{\:J,\:x}(\tilde{x},h)=\mathsf{1}_{\:J}(\tilde{x},x-h)-\mathsf{1}_{\:J}(\tilde{x},x+h)$ converges pointwise to $\mathsf{1}(x-h\leq\tilde{x}\leq x+h)=\mathsf{1}(x-h\leq\tilde{x})-\mathsf{1}(x+h<\tilde{x})$ for $\tilde{x}\in(0,1)\setminus\{x-h,x+h\}$. 

Now, from part A of Lemma \ref{Quantile:lemma:indicator-uniform-convergence}, we know that $\mathsf{1}_{\:J}(z)$ converges uniformly to both $\mathsf{1}(z<0)$ and $\mathsf{1}(z\leq 0)$ for $z\in(\-\pi,\pi)\setminus(-\delta,\delta)$ for any $\delta\in(0,\pi)$. So, there exists $J_{\:\epsilon,\:\delta}\in\mathbb{N}$ such that $\big|\mathsf{1}_{\:J}(z)-\mathsf{1}(z<0)\big|\leq\nicefrac{\epsilon}{2}$, and $\big|\mathsf{1}_{\:J}(z)-\mathsf{1}(z\leq 0)\big|\leq\nicefrac{\epsilon}{2}$ if $z\in(\-\pi,\pi)\setminus(-\delta,\delta)$ and $J>J_{\:\epsilon,\:\delta}$. Let us define the set
\begin{align*}
R_{\:\delta,\:h,\:x}:=(0,1)\setminus\big((x-h-\delta,x+h+\delta)\cup(x+h-\delta,x+h+\delta)\big).
\end{align*}
If $\tilde{x}\in R_{\:\delta,\:h,\:x}$, then $\tilde{x}-x+h\in(\-\pi,\pi)\setminus(-\delta,\delta)$ and $\tilde{x}-x-h\in(\-\pi,\pi)\setminus(-\delta,\delta)$, and in turn, we will have $\big|\mathsf{1}_{\:J,\:x}(\tilde{x},h)-\mathsf{1}(x-h\leq\tilde{x}\leq x+h)\big|\leq\nicefrac{\epsilon}{2}$ for $J>J_{\:\epsilon,\:\delta}$. 

Observe that the Lebesgue measure of $(0,1)\setminus R_{\:\delta,\:h,\:x}$ satisfies $\lambda\big((0,1)\setminus R_{\:\delta,\:h,\:x}\big)\leq 4\delta$, for any $h\in(0,1)$. Also, from part B of Lemma \ref{Quantile:lemma:indicator-uniform-convergence}, we realize that
\begin{align*}
\big|\:\mathsf{1}_{\:J,\:x}(\tilde{x},h)\:\big|\leq\big|\:\mathsf{1}_{\:J}(\tilde{x},x-h)\:\big|+\big|\:\mathsf{1}_{\:J}(\tilde{x},x+h)\:\big|\leq 2\cdot(9/2+1/\pi).
\end{align*}
Thus, for $\tilde{x}\in(0,1)$ and $h\in(0,1)$, we have:
\begin{align*}
\big|\:\mathsf{1}_{\:x}(\tilde{x},h)-\mathsf{1}_{\:J,\:x}(\tilde{x},h)\:\big|\leq\big|\:\mathsf{1}_{\:x}(\tilde{x},h)\:\big|+\big|\:\mathsf{1}_{\:J,\:x}(\tilde{x},h)\:\big|\leq 2+2\cdot(9/2+1/\pi).
\end{align*}
For convenience, let us denote this bound by $M_{\:0}$. 

Since by Assumption \ref{Lowess:assumption:positive-density}, the probability measure $P(\cdot)$ is absolutely continuous with respect to Lebesgue measure $\lambda(\cdot)$, there exists $\delta_{\:\epsilon}>0$ such that if $\lambda(A)<\delta_{\:\epsilon}$ for some $A\in\mathbb{B}(0,1)$, we have $P(A)<\nicefrac{\epsilon}{(2\cdot M_{\:0})}$. Take $A=R_{\:\nicefrac{\delta_{\:\epsilon}}{4},\:h,\:x}$, and for $J>J_{\:\epsilon,\:\nicefrac{\delta_{\:\epsilon}}{4}}$, we have:
\begin{align*}
&\hspace{-1em}\big|\:\mathbb{E}_{\:P}\big(\mathsf{1}_{\:J,\:x}(\tilde{x},h)\big)-F_{\:x}^{\:(P)}(h)\:\big|\\
&=\big|\:\mathbb{E}_{\:P}\big(\mathsf{1}_{\:J,\:x}(\tilde{x},h)\big)-\mathbb{E}_{\:P}\big(\mathsf{1}_{\:x}(\tilde{x},h)\big)\:\big| \\
&=\big|\:\mathbb{E}_{\:P}\big(\mathsf{1}_{\:x}(\tilde{x},h)-\mathsf{1}_{\:J,\:x}(\tilde{x},h)\big)\:\big|\\
&\leq\mathbb{E}_{\:P}\big(\big|\:\mathsf{1}_{\:x}(\tilde{x},h)-\mathsf{1}_{\:J,\:x}(\tilde{x},h)\:\big|\big) \\
&=\int_{\tilde{x}\in(0,1)}\big|\:\mathsf{1}_{\:x}(\tilde{x},h)-\mathsf{1}_{\:J,\:x}(\tilde{x},h)\:\big|\:dP \\
&=\int_{\tilde{x}\in R_{\:\nicefrac{\delta_{\:\epsilon}}{4},\:h,\:x}}\big|\:\mathsf{1}_{\:x}(\tilde{x},h)-\mathsf{1}_{\:J,\:x}(\tilde{x},h)\:\big|\:dP \\
&\hspace{2em}+\int_{\tilde{x}\in(0,1)\setminus R_{\:\nicefrac{\delta_{\:\epsilon}}{4},\:h,\:x}}\big|\:\mathsf{1}_{\:x}(\tilde{x},h)-\mathsf{1}_{\:J,\:x}(\tilde{x},h)\:\big|\:dP \\
&\leq\int_{\tilde{x}\in R_{\:\nicefrac{\delta_{\:\epsilon}}{4},\:h,\:x}}\frac{\epsilon}{2}\cdot dP+\int_{\tilde{x}\in(0,1)\setminus R_{\:\nicefrac{\delta_{\:\epsilon}}{4},\:h,\:x}}M_{\:0}\cdot dP \\
&=\frac{\epsilon}{2}\cdot P\big(R_{\:\nicefrac{\delta_{\:\epsilon}}{4},\:h,\:x}\big)+M_{\:0}\cdot P\big((0,1)\setminus R_{\:\nicefrac{\delta_{\:\epsilon}}{4},\:h,\:x}\big)\\
& <\frac{\epsilon}{2}\cdot 1+M_{\:0}\cdot \frac{\epsilon}{2\cdot M_{\:0}}=\epsilon.
\end{align*}
Note that $J_{\:\epsilon,\:\nicefrac{\delta_{\:\epsilon}}{2}}$ is independent of $x$ and $h$. Since $\epsilon$ is an arbitrary positive number, we must have that $\lim\limits_{J\to\infty}\mathbb{E}_{\:P}\big(\mathsf{1}_{\:J,\:x}(\tilde{x},h)\big)\to F_{\:x}^{\:(P)}(h)$ uniformly in $x\in(0,1)$ and $h\in(0,1)$. \\

(B)\ Let $x\in(0,1)$, and consider an arbitrary $\alpha\in(0,1)$. We can pick an $\epsilon_{\:\alpha}>0$ such that $\epsilon_{\:\alpha}<1-\alpha$. For the random variable $\tilde{x}$, we have $x-1\leq\tilde{x}\leq x+1$; that is, $\mathsf{1}_{\:x}(\tilde{x},1)=1$ for any $x$, and hence $\mathbb{E}_{\:P}\big(\mathsf{1}_{\:x}(\tilde{x},1)=P(x-1\leq\tilde{x}\leq x+1)=1$ for any $x$. By part A, there exists $J_{\:\alpha}\in\mathbb{N}$ such that $|\mathbb{E}_{\:P}\big(\mathsf{1}_{\:J,\;x}(\tilde{x},1)\big)-P(x-1\leq\tilde{x}\leq x+1)|<\epsilon_{\:\alpha}$ for $J>J_{\:\alpha}$. Thus we conclude that $\mathbb{E}_{\:P}\big(\mathsf{1}_{\:J,\:x}(\tilde{x},1)\big)>1-\epsilon_{\:\alpha}$ if $J>J_{\:\alpha}$. 

Now, from part A of Lemma \ref{Lowess:lemma:delta_expression}, we have:
\begin{align*}
\mathbb{E}_{\:P}\big(\mathsf{1}_{\:J,\:x}(\tilde{x},h)\big)=\frac{4}{\pi}\sum_{j=1}^J\mathbb{E}_{\:P}\big(\mathsf{c}_{\:2j-1}(\tilde{x}-x)\big)\cdot\frac{\mathsf{c}_{\:2j}(h)}{2j-1}.
\end{align*}
This expression is a continuous function of $h$, with $\mathbb{E}_{\:P}\big(\mathsf{1}_{\:J,\:x}(\tilde{x},0)\big)=0$ and $\mathbb{E}_{\:P}\big(\mathsf{1}_{\:J,\:x}(\tilde{x},1)\big)=\alpha'\text{ (say)}>1-\epsilon_{\:\alpha}>\alpha$. So, by the intermediate value theorem, there is a number $h=h_{\:J,\:\alpha,\:x}^{\:(P)}\in(0,1)$ satisfying $\mathbb{E}_{\:P}\big(\mathsf{1}_{\:J,\:x}(\tilde{x},h)\big)=\alpha$. So, we conclude that $\alpha\in\mathcal{A}_{\:J,\:x}^{\:(P)}$ if $J>J_{\:\alpha}$, and this is true for arbitrary $\alpha\in(0,1)$. Hence $\lim\limits_{J\to\infty}\mathcal{A}_{\:J,\:x}^{\:(P)}=(0,1)$. \\

(C)\ For $x\in(0,1)$ and $\alpha\in(0,1)$, consider any sequence of numbers $h_{\:J,\:\alpha,\:x}^{\:(P)}$ when such a sequence exists. From part B, we know that $h_{\:J,\:\alpha,\:x}^{\:(P)}$ exists when $J$ is large ($J>J_{\:\alpha}$, say). For such a $J$, we have $\mathbb{E}_{\:P}\Big(\mathsf{1}_{\:J,\:x}\big(\tilde{x},h_{\:J,\:\alpha,\:x}^{\:(P)}\big)\Big)=\alpha$. Given  $\epsilon>0$, we know from part A that there exists an integer $J_{\:0}$ such that $\big|\:\mathbb{E}_{\:P}\big(\mathsf{1}_{\:x}(\tilde{x},h)\big)-F_{\:x}^{\:(P)}(h)\:\big|<\epsilon$, for any $h\in(0,1)$ if $J>J_{\:0}$. If we replace $h$ with $h_{\:J,\:\alpha,\:x}^{\:(P)}$ in this inequality, we obtain $\big|\:\alpha-F_{\:x}^{\:(P)}\big(h_{\:J,\:\alpha,\;x}^{\:(P)}\big)\:\big|<\epsilon$ when $J>J_{\:\epsilon}=\max(J_{\:\alpha},J_{\:0})$. 

Now, we know that $F_{\:x}^{\:(P)}\big(h^{\:(P)}_{\:\alpha,\:x}\big)=\alpha$. Thus, if $J>J_{\:\epsilon}$, it follows that $\big|\:F_{\:x}^{\:(P)}\big(h^{\:(P)}_{\:\alpha,\:x}\big)-F_{\:x}^{\:(P)}\big(h_{\:J,\:\alpha,\:x}^{\:(P)}\big)\:\big|<\epsilon$. In other words, $\lim\limits_{J\to\infty}F_{\:x}^{\:(P)}\big(h_{\:J,\:\alpha,\:x}^{\:(P)}\big)\to F_{\:x}^{\:(P)}\big(h^{\:(P)}_{\:\alpha,\:x}\big)$, and hence by Lemma \ref{Lowess:lemma:inverse_continuity} we have that $\lim\limits_{J\to\infty}h_{\:J,\:\alpha,\:x}^{\:(P)}\to h^{\:(P)}_{\:\alpha,\:x}$.

Next, we will show that $\lim\limits_{J\to\infty}\bigtriangledown\big(\mathcal{H}_{\:J,\:\alpha,\:x}^{\:(P)}\big)\to 0$. Pick $\epsilon>0$. If $J>J_{\:\nicefrac{\epsilon}{2}}$, then
\begin{align*}
\big|\:h_{\:J,\:\alpha,\:x}^{\:(P)}-h^{\:(P)}_{\:\alpha,\:x}\:\big|<\nicefrac{\epsilon}{2}\text{ and }\big|\:{h'}_{\:J,\:\alpha,\:x}^{\:(P)}-h^{\:(P)}_{\:\alpha,\:x}\:\big|<\nicefrac{\epsilon}{2}
\end{align*}
for $\big\{h_{\:J,\:\alpha,\:x}^{\:(P)},{h'}_{\:J,\:\alpha,\:x}^{\:(P)}\big\}\in\mathcal{H}_{\:J,\:\alpha,\:x}^{\:(P)}$ and $h_{\:J,\:\alpha,\:x}^{\:(P)}\neq{h'}_{\:J,\:\alpha,\:x}^{\:(P)}$. Then, we have:
\begin{align*}
\big|\:h_{\:J,\:\alpha}^{\:(P)}(x)-{h'}_{\:J,\:\alpha}^{\:(P)}(x)\:\big|&=\big|\:\big(h_{\:J,\:\alpha}^{\:(P)}(x)-h_{\:J}^{\:(P)}(x)\big)-\big({h'}_{\:J,\:\alpha}^{\:(P)}(x)-h_{\:J}^{\:(P)}(x)\big)\:\big| \\
&\leq\big|\:h_{\:J,\:\alpha}^{\:(P)}(x)-h_{\:J}^{\:(P)}(x)\:\big|+\big|\:{h'}_{\:J,\:\alpha}^{\:(P)}(x)-h_{\:J}^{\:(P)}(x)\:\big|\\
&<\nicefrac{\epsilon}{2}+\nicefrac{\epsilon}{2}=\epsilon. 
\end{align*}
Since $\epsilon$ is arbitrary, we have proved that $\lim\limits_{J\to\infty}\bigtriangledown\big(\mathcal{H}_{\:J,\:\alpha,\:x}^{\:(P)}\big)\to 0$.
\end{proof}

We now characterize the different solution sets that are possible when considering the expression $\hat{F}_{\:J,\:x}(X,h)=\alpha$. Fix $x\in\mathbb{X}$, $J\in\mathbb{N}$ and $\alpha\in(0,1)$, and for a given data set $X$, denote by $\hat{h}_{\:J,\:\alpha,\:x}(X)$ a solution to the equation $\hat{F}_{\:J,\:x}(X,h)=\alpha$ for $h\in(0,1)$, whenever such a solution exists. For $J\in\mathbb{N}$, define the set:
\begin{align*}
\hat{\mathcal{A}}_{\:J,\:x}(X):=\big\{\alpha\colon\exists\ h\in(0,1),\text{ such that }\hat{F}_{\:J,\:x}(X,h)=\alpha\big\},
\end{align*}
and for $\alpha\in(0,1)$, $J\in\mathbb{N}$, $n\in\mathbb{N}$ and $w\in\Omega$ define the set:
\begin{align*}
\hat{\mathcal{H}}_{\:J,\:\alpha,\:x,\:n}(w):=\Big\{h\in(0,1)\colon\hat{F}_{\:J,\:x}\big(X(w,n),h\big)=\alpha\Big\}.
\end{align*}

Observe that, for small $n$, it is possible that $\hat{\mathcal{H}}_{\:J,\:\alpha,\:x,\:n}(w)=\emptyset$, meaning there is no solution to $\hat{F}_{\:J,\:x}\big(X(w,n),h\big)=\alpha$ for $h\in(0,1)$. It is also possible that there are multiple solutions to $\hat{F}_{\:J,\:x}\big(X(w,n),h\big)=\alpha$; however, $\hat{\mathcal{H}}_{\:J,\:\alpha,\:x,\:n}(w)$ is always a finite set, because the equation $\hat{F}_{\:J,\:x}\big(X(w,n),h\big)=\alpha$ can't have infinitely many solutions for $h\in(0,1)$. 

Next, we motivate the formulation of Assumption \ref{Lowess:assumption:uniform-convergence} and relate this to the solution sets of $\hat{F}_{\:J,\:x}(X,h)=\alpha$ introduced above. By part A of Lemma \ref{Lowess:lemma:population_asymptotic}, we have for $x\in(0,1)$ and $h\in(0,1)$:
\begin{align*}
&\hspace{-2em}\lim\limits_{J\to\infty}\mathbb{E}_{\:P}\big(\mathsf{1}_{\:J,\:x}(\tilde{x},h)\big)=F_{\:x}^{\:(P)}(h)=P(x-h\leq\tilde{x}\leq x+h) \\
&\Leftrightarrow\lim\limits_{J\to\infty}\int_{0}^1\mathsf{1}_{\:J,\:x}(t,h)f^{\:(P)}(t)\,dt=F_{\:x}^{\:(P)}(h).
\end{align*}
Using the expression for $\mathsf{1}_{\:J,\:x}(t,h)$ from part A of Lemma \ref{Lowess:lemma:delta_expression}, we have:
\begin{align}\label{Lowess:eq:c}
\lim\limits_{J\to\infty}\bigg(\frac{4}{\pi}\int_{0}^1\sum_{j=1}^{J}\mathsf{c}_{\:2j-1}(t-x)\cdot \frac{\mathsf{c}_{\:2j}(h)}{2j-1}\cdot f^{\:(P)}(t)\,dt\bigg)=F_{\:x}^{\:(P)}(h).
\end{align}

For any fixed $J\in\mathbb{N}$, we have:
\begin{align*}
\hspace{-2em}\int_0^1\sum_{j=1}^J\mathsf{c}_{\:2j-1}(t-x)\cdot \frac{\mathsf{c}_{\:2j}(h)}{2j-1}\cdot f^{\:(P)}(t)\,dt=\sum_{j=1}^J\int_0^1\mathsf{c}_{\:2j-1}(t-x)\cdot \frac{\mathsf{c}_{\:2j}(h)}{2j-1}\cdot f^{\:(P)}(t)\,dt.
\end{align*}
So, from equation \ref{Lowess:eq:c}, we have:
\begin{align*}
\frac{4}{\pi}\sum_{j=1}^{\infty}\bigg(\int_0^1\mathsf{c}_{\:2j-1}(t-x)\cdot f^{\:(P)}(t)\,dt\bigg)\frac{\mathsf{c}_{\:2j}(h)}{2j-1}=F_{\:x}^{\:(P)}(h).
\end{align*}

Next, observe that $\mathsf{c}_{\:2j-1}(t-x)=\mathsf{c}_{\:2j-1}(t)\cdot \mathsf{c}_{\:2j-1}(x)+\mathsf{c}_{\:2j}(t)\cdot \mathsf{c}_{\:2j}(x)$, and so we have:
\begin{align}\label{Lowess:eq:d}
\frac{4}{\pi}\sum_{j=1}^{\infty}\Big(\mathbb{E}_{\:P}\big(\mathsf{c}_{\:2j-1}(\tilde{x})\big)\cdot \mathsf{c}_{\:2j-1}(x)+\mathbb{E}_{\:P}\big(\mathsf{c}_{\:2j}(\tilde{x})\big)\cdot \mathsf{c}_{\:2j}(x)\Big)\frac{\mathsf{c}_{\:2j}(h)}{2j-1}=F_{\:x}^{\:(P)}(h).
\end{align}
We also have that ${\mathsf{c}_{\:2j}}'(h)=(2j-1)\cdot \mathsf{c}_{\:2j-1}(h)$, and so the derivative of the $j$th term of the left-hand side of equation \ref{Lowess:eq:d} is $\nicefrac{4}{\pi}\Big(\mathbb{E}_{\:P}\big(\mathsf{c}_{\:2j-1}(\tilde{x})\big)\mathsf{c}_{\:2j-1}(x)+\mathbb{E}_{\:P}\big(\mathsf{c}_{\:2j}(\tilde{x})\big)\mathsf{c}_{\:2j}(x)\Big)\mathsf{c}_{\:2j}(h)$. The corresponding partial sum in $J$ is then precisely $\eta_{\:J,\:x}^{\:(P)}(h)$ as defined in Assumption \ref{Lowess:assumption:uniform-convergence}:
\begin{align*}
\eta_{\:J,\:x}^{\:(P)}(h)&=\mathbb{E}_{\:P}\left[\frac{2}{\pi} \Big(\frac{\sin\left(2J(\tilde{x}-x+h)\right)}{\sin(\tilde{x}-x+h)}+\frac{\sin\left(2J(\tilde{x}-x-h)\right)}{\sin(\tilde{x}-x-h)}\Big)\right]\\
&=\frac{4}{\pi} \sum_{j=1}^J\Big(\mathbb{E}_{\:P}\big(\mathsf{c}_{\:2j-1}(\tilde{x})\big)\mathsf{c}_{\:2j-1}(x)+\mathbb{E}_{\:P}\big(\mathsf{c}_{\:2j}(\tilde{x})\big)\mathsf{c}_{\:2j}(x)\Big)\mathsf{c}_{\:2j}(h).
\end{align*}
As seen from part A of Lemma \ref{Lowess:lemma:delta_expression}, $\eta_{\:J,\:x}^{\:(P)}(h)$ can equivalently be expressed as the derivative with respect to $h$ of $\mathbb{E}_{\:P}\big(\mathsf{1}_{\:J,\:x}(\tilde{x},h)\big)$.

From \citet[Theorem 9.13]{Book_1957_Apostol}, we know if $\eta_{\:J,\:x}^{\:(P)}(h)$ converges uniformly to a limit, then this limit will be the derivative of the right-hand side of equation \ref{Lowess:eq:d}: $f^{\:(P)}(x-h)+f^{\:(P)}(x+h)$. Thus, for $x\in(0,1)$ and $J\in\mathbb{N}$, define the corresponding distance $\xi_{\:J,\:x}^{\:(P)} \in \mathbb{R} \cup \infty $ as follows:
\begin{align*}
\xi_{\:J,\:x}^{\:(P)}:=\sup\limits_{h\in(0,1)}\big|\:f^{\:(P)}(x-h)+f^{\:(P)}(x+h)-\eta_{\:J,\:x}^{\:(P)}(h)\:\big|.
\end{align*}
Below, to prove Theorem \ref{Lowess:thm:prop}, we shall use Assumption \ref{Lowess:assumption:uniform-convergence} to force $\xi_{\:J,\:x}^{\:(P)}$ to zero in $J$. For the moment, however, we use $\xi_{\:J,\:x}^{\:(P)}$ simply to characterize limit points with respect to $\hat{\mathcal{H}}_{\:J,\:\alpha,\:x}(w)$, defined for $\alpha\in(0,1)$, $J\in\mathbb{N}$, and $w\in\Omega$ as the set
\begin{align*}
\hat{\mathcal{H}}_{\:J,\:\alpha,\:x}(w):=\bigcup_{n\in\mathbb{N}}\hat{\mathcal{H}}_{\:J,\:\alpha,\:x,\:n}(w).
\end{align*}

\begin{lemma}\label{Lowess:lemma:sample_asymptotic}
Suppose Assumption \ref{Lowess:assumption:positive-density} holds. Then: \\
(A)\ $\lim\limits_{|X|\to\infty}\sup\limits_{\{x,h\}\in(0,1)}\big|\:\hat{F}_{\:J,\:x}(X,h)-\mathbb{E}_{\:P}\big(\mathsf{1}_{\:J,\:x}(\tilde{x},h)\big)\:\big|\overset{a.s.}{=} 0$. \\
(B)\ $\lim\limits_{J\rightarrow\infty}\lim\limits_{|X|\rightarrow\infty}\hat{\mathcal{A}}_{\:J,\:x}(X)\overset{a.s.}{=}(0,1)$. \\
(C)\ Recalling that $h_{\:J,\:\alpha,\:x}^{\:(P)}$ is any solution to $\mathbb{E}_{\:P}\big(\mathsf{1}_{\:J,\:x}(\tilde{x},h)\big)=\alpha$, for $J$ sufficiently large, there exists a null set $N$ such that for $w\in N'$, if $\hat{h}_{\:J,\:\alpha,\:x}\big(X(w,\cdot)\big)$ is a limit point with respect to the set $\hat{\mathcal{H}}_{\:J,\:\alpha,\:x}(w)$, then $\Big|\:F_{\:x}^{\:(P)}\Big(\hat{h}_{\:J,\:\alpha,\:x}\big(X(w,\cdot)\big)\Big)-F_{\:x}^{\:(P)}\Big(h_{\:J,\:\alpha,\:x}^{\:(P)}\Big)\:\Big|\leq\xi^{\:(P)}_{\:J,\:x}$.
\end{lemma}

\begin{proof}
(A)\ Suppose Assumption \ref{Lowess:assumption:positive-density} holds. Let $x\in(0,1)$ and $h\in(0,1)$, and consider arbitrary $\epsilon>0$. Since $|\mathsf{c}_{\:j}(\tilde{x})|\leq 1$, we have $\big|\mathbb{E}_{\:P}\big(\mathsf{c}_{\:j}(\tilde{x})\big)\big|\leq 1$ for  $j\in\mathbb{N}$. By Kolmogorov's strong law of large numbers, we have $\bar{\mathsf{C}}_{\:j}(X)\overset{a.s.}{\longrightarrow}\mathbb{E}_{\:P}\big(\mathsf{c}_{\:j}(\tilde{x})\big)$. So, for each $j$, there exists a null set $N_{\:j}\subset \Omega$, such that $\bar{\mathsf{C}}_{\:j}\big(X(w,n)\big)\to\mathbb{E}_{\:P}\big(\mathsf{c}_{\:j}(\tilde{x})\big)$ as $n\to\infty$, if $w\in {N_{\:j}}'$. 

Take $N=\bigcup_{j\in\mathbb{N}}N_{\:j}$; since $N$ is a countable union of null sets, it is also a null set. Now if $w\in N'$, then $\bar{\mathsf{C}}_{\:j}\big(X(w,n)\big)\to\mathbb{E}_{\:P}\big(\mathsf{c}_{\:j}(\tilde{x})\big)$ for $j\in\mathbb{N}$. Consider arbitrary $w\in N'$; then, from parts A and B of Lemma \ref{Lowess:lemma:delta_expression}, we have:
\begin{align*}
&\hspace{-1em}\hat{F}_{\:J,\:x}\big(X(w,n),h\big)-\mathbb{E}_{\:P}\big(\mathsf{1}_{\:J,\:x}(\tilde{x},h)\big) \\
&=\frac{4}{\pi}\cdot \sum_{j=1}^{J}\bigg(\Big(\bar{\mathsf{C}}_{\:2j-1}\big(X(w,n)\big)-\mathbb{E}_{\:P}\big(\mathsf{c}_{\:2j-1}(\tilde{x})\big)\Big)\cdot \mathsf{c}_{\:2j-1}(x) \\
&\hspace{3cm}+\Big(\bar{\mathsf{C}}_{\:2j}\big(X(w,n)\big)-\mathbb{E}_{\:P}\big(\mathsf{c}_{\:2j}(\tilde{x})\big)\Big)\cdot \mathsf{c}_{\:2j}(x)\bigg)\cdot\frac{\mathsf{c}_{\:2j}(h)}{2j-1} \\
&=\frac{4}{\pi}\cdot \sum_{j=1}^{J}r_{\:j,\:x}\big(X(w,n),\tilde{x}\big)\cdot \frac{\mathsf{c}_{\:2j}(h)}{2j-1},\text{ where we define} \\
&\hspace{-1em}r_{\:j,\:x}\big(X(w,n),\tilde{x}\big):=\Big(\bar{\mathsf{C}}_{\:2j-1}\big(X(w,n)\big)-\mathbb{E}_{\:P}\big(\mathsf{c}_{\:2j-1}(\tilde{x})\big)\Big)\cdot \mathsf{c}_{\:2j-1}(x) \\
&\hspace{10em}+\:\Big(\bar{\mathsf{C}}_{\:2j}\big(X(w,n)\big)-\mathbb{E}_{\:P}\big(\mathsf{c}_{\:2j}(\tilde{x})\big)\Big)\cdot \mathsf{c}_{\:2j}(x).
\end{align*}
Now, from Abel's identity for partial sums, we have:
\begin{align*}
&\hspace{-3em}\hat{F}_{\:J,\:x}\big(X(w,n),h\big)-\mathbb{E}_{\:P}\big(\mathsf{1}_{\:J,\:x}(\tilde{x},h)\big)=\frac{4}{\pi}\cdot \sum_{j=1}^Jr_{\:j,\:x}\big(X(w,n),\tilde{x}\big)\cdot \frac{\mathsf{c}_{\:2j}(h)}{2j-1} \\
&=r_{\:J,\:x}\big(X(w,n),\tilde{x}\big)\cdot \Big(\frac{4}{\pi}\cdot \sum_{k=1}^J\frac{\mathsf{c}_{\:2k}(h)}{2k-1}\Big) \\
&\qquad+\sum_{j=1}^{J-1}\Big(r_{\:j,\:x}\big(X(w,n),\tilde{x}\big)-r_{\:j+1,\:x}\big(X(w,n),\tilde{x}\big)\Big)\cdot \Big(\frac{4}{\pi}\cdot \sum_{k=1}^j\frac{\mathsf{c}_{\:2k}(h)}{2k-1}\Big).
\end{align*}
Recall that immediately before Lemma \ref{Quantile:lemma:indicator-uniform-convergence}, we defined for $h\in(0,1)$ and $j\in\mathbb{N}$ the quantity
\begin{align*}
&\hspace{-1em}\mathsf{1}_{\:j}(h)=\frac{1}{2}-\frac{2}{\pi}\cdot\sum_{k=1}^j\frac{\sin\big((2k-1)h\big)}{2k-1}=\frac{1}{2}-\frac{2}{\pi}\cdot \sum_{k=1}^j\frac{\mathsf{c}_{\:2k}(h)}{2k-1} \\
&\Leftrightarrow\:\frac{4}{\pi}\cdot \sum_{k=1}^j\frac{\mathsf{c}_{\:2k}(h)}{2k-1}=2\cdot \Big(\frac{1}{2}-\mathsf{1}_{\:j}(h)\Big).
\end{align*}
Thus we may write
\begin{multline*}
\hat{F}_{\:J,\:x}\big(X(w,n),h\big)-\mathbb{E}_{\:P}\big(\mathsf{1}_{\:J,\:x}(\tilde{x},h)\big) =r_{\:J,\:x}\big(X(w,n),\tilde{x},x\big)\cdot 2\cdot \Big(\frac{1}{2}-\mathsf{1}_{\:J}(h)\Big) \\
+\sum_{j=1}^{J-1}\Big(r_{\:j,\:x}\big(X(w,n),\tilde{x}\big)-r_{\:j+1,\:x}\big(X(w,n),\tilde{x}\big)\Big)\cdot 2\cdot\Big(\frac{1}{2}-\mathsf{1}_{\:j}(h)\Big).
\end{multline*}

It follows from part B of Lemma \ref{Quantile:lemma:indicator-uniform-convergence} that $\mathsf{1}_{\:j}(h)$ is uniformly bounded for $j\in\mathbb{N}$ and $h\in(0,1)$. So, there exists $M\in\mathbb{R}$ such that $\big|\:\nicefrac{1}{2}-\mathsf{1}_{\:j}(h)\:\big|<M$ for $j\in\mathbb{N}$ and $h\in(0,1)$, and
\begin{align*}
&\hspace{-1em}\Big|\:\hat{F}_{\:J,\:x}\big(X(w,n),h\big)-\mathbb{E}_{\:P}\big(\mathsf{1}_{\:J,\:x}(\tilde{x},h)\big)\:\Big| \\
&\leq\big|\:r_{\:J,\:x}\big(X(w,n),\tilde{x}\big)\:\big|\cdot 2\cdot\Big|\:\Big(\mathsf{1}_{\:J,\:x}(h)-\frac{1}{2}\Big)\:\Big| \\
&\quad+\sum_{j=1}^{J-1}\left|\:\Big(r_{\:j}\big(X(w,n),\tilde{x}\big)-r_{\:j+1}\big(X(w,n),\tilde{x}\big)\Big)\:\right|\cdot 2\cdot \left|\:\Big(\mathsf{1}{\:j}(h)-\frac{1}{2}\Big)\:\right| \\
&<\left|\:r_{\:J}\big(X(w,n),\tilde{x}\big)\:\right|\cdot 2\cdot M \\
&\quad+\sum_{j=1}^{J-1}\Big(\big|\:r_{\:j,\:x}\big(X(w,n),\tilde{x}\big)\:\big|+\big|\:r_{\:j+1,\;x}\big(X(w,n),\tilde{x}\big)\:\big|\Big)\cdot 2\cdot M \\
&<2M\cdot\sum_{j=1}^J\big|\:r_{\:j,\:x}\big(X(w,n),\tilde{x}\big)\:\big| \\
&=2M\cdot \sum_{j=1}^J\Big|\:\Big(\bar{\mathsf{C}}_{\:2j-1}\big(X(w,n)\big)-\mathbb{E}_{\:P}\big(\mathsf{c}_{\:2j-1}(\tilde{x})\big)\Big)\cdot \mathsf{c}_{\:2j-1}(x) \\
&\hspace{5cm}+\:\Big(\bar{\mathsf{C}}_{\:2j}\big(X(w,n)\big)-\mathbb{E}_{\:P}\big(\mathsf{c}_{\:2j}(\tilde{x})\big)\Big)\cdot \mathsf{c}_{\:2j}(x)\:\Big| \\
&\leq 2M \cdot \sum_{j=1}^{2J}\Big|\:\Big(\bar{\mathsf{C}}_{\:j}\big(X(w,n)\big)-\mathbb{E}_{\:P}\big(\mathsf{c}_{\:j}(\tilde{x})\big)\Big)\:\Big|,
\end{align*}
with the final inequality following as $|\mathsf{c}_{\:j}(x)|\leq 1\text{ for }j\in\mathbb{N},x\in(0,1)$.

Now, for $j\in\mathbb{N}$, pick $n_{\:w,\:j}\in\mathbb{N}$ large enough so that for $n>n_{\:w,\:j}$, we have
\begin{align*}
\left|\:\bar{\mathsf{C}}_{\:j}\big(X(w,n)\big)-\mathbb{E}_{\:P}\big(\mathsf{c}_{\:j}(\tilde{x})\big)\right|<\frac{\epsilon}{2M}\cdot 2^{-j}.
\end{align*}
Define $n_{\:w,\:(1:2J)}:=\max(n_{\:w,\:1},n_{\:w,\:2},\dots,n_{\:w,\:2J})$. Then for $n>n_{\:w,\:(1:2J)}$, we have
\begin{align*}
\Big|\:\hat{F}_{\:J,\:x}\big(X(w,n),h\big)-\mathbb{E}_{\:P}\big(\mathsf{1}_{\:J,\:x}(\tilde{x},h)\big)\:\Big|<2M\cdot \sum_{j=1}^{2J}\frac{\epsilon}{2M}\cdot 2^{-j}<\epsilon.
\end{align*}
Observe that $n_{\:w,\:(1:2J)}$ does not depend on $x$ or $h$, and so the above inequality holds uniformly for any $x\in(0,1)$ and $h\in(0,1)$. In other words, for $w\in N'$, if $n>n_{\:w,\:(1:2J)}$, then
\begin{align*}
\underset{\:\{x,h\}\in(0,1)}{\sup}\Big|\:\hat{F}_{\:J,\:x}\big(X(w,n),h\big)-\mathbb{E}_{\:P}\big(\mathsf{1}_{\:J,\:x}(\tilde{x},h)\big)\:\Big|<\epsilon.
\end{align*}
Since $P(N)=0$, we conclude that 
$\underset{|X|\to\infty}{\lim}\underset{\{x,h\}\in(0,1)}{\sup}\Big|\:\hat{F}_{\:J,\:x}(X,h)-\mathbb{E}_{\:P}\big(\mathsf{1}_{\:J,\:x}(\tilde{x},h)\big)\:\Big|\overset{a.s.}{=}0$.

(B)\ Let $x\in(0,1)$ and $\alpha\in(0,1)$, and observe that we may choose an $\epsilon_{\:\alpha}>0$ such that $\epsilon_{\:\alpha}<1-\alpha$. Take $\alpha^{\:(1)}\in(\alpha+\epsilon_{\:\alpha},1)$. Recall from part B of Lemma \ref{Lowess:lemma:population_asymptotic} that $\lim\limits_{J\rightarrow\infty}\mathcal{A}_{\:J,\:x}^{\:(P)}=(0,1)$, where $\mathcal{A}_{\:J,\:x}^{\:(P)}=\big\{\alpha\colon\exists\ h\in(0,1)\text{ such that }\mathbb{E}_{\:P}\big(\mathsf{1}_{\:J,\:x}(\tilde{x},h)\big)=\alpha\big\}$. So, there exists a $J_{\:\alpha}$ such that if $J>J_{\:\alpha}$, then we can find an $h_{\:J,\:\alpha^{\:(1)},\:x}^{\:(P)}\in(0,1)$ such that $\mathbb{E}_{\:P}\Big(\mathsf{1}_{\:J,\:x}\big(\tilde{x},h_{\:J,\:\alpha^{\:(1)},\:x}^{\:(P)}\big)\Big)=\alpha^{\:(1)}$. 

By part A we know that there exists a null set $N$ such that for $w\in N'$, we can find an $n_{\:w,\:(1:2J)}\in\mathbb{N}$, such that for $n>n_{\:w,\:(1:2J)}$ and any $h\in(0,1)$, we have
\begin{align*}
\big|\:\hat{F}_{\:J,\:x}\big(X(w,n),h\big)-\mathbb{E}_{\:P}\big(\mathsf{1}_{\:J,\:x}(\tilde{x},h)\big)\:\big|\leq\epsilon_{\:\alpha}.
\end{align*}
In particular, for $h=h_{\:J,\:\alpha^{\:(1)},\:x}^{\:(P)}$, if $n>n_{\:w,\:(1:2J)}$ we have
\begin{align*}
\Big|\:\hat{F}_{\:J,\:x}\big(X(w,n),h_{\:J,\:\alpha^{\:(1)},\:x}^{\:(P)}\big)-\mathbb{E}_{\:P}\Big(\mathsf{1}_{\:J,\:x}\big(\tilde{x},h_{\:J,\:\alpha^{\:(1)},\:x}^{\:(P)}\big)\Big)\:\Big|\leq\epsilon_{\:\alpha}.
\end{align*}
It thus follows that if $n>n_{\:w,\:(1:2J)}$, then $\hat{F}_{\:J,\:x}\big(X(w,n),h_{\:J,\:\alpha^{\:(1)},\:x}^{\:(P)}\big)\geq\alpha^{\:(1)}-\epsilon_{\:\alpha}>\alpha$. 

Observe from part B of Lemma \ref{Lowess:lemma:delta_expression} that $\hat{F}_{\:J,\:x}\big(X(w,n),0\big)=0$, and furthermore that $\hat{F}_{\:J,\:x}\big(X(w,n),h\big)$ is a continuous function of $h$. So, by the intermediate value theorem, we can find $h$ such that $\hat{F}_{\:J,\:x}\big(X(w,n),h\big)=\alpha$. Recalling that $\alpha\in(0,1)$ is arbitrary, we thus conclude that $\alpha\in\hat{\mathcal{A}}_{\:J,\;x}\big(X(w,n)\big)$ for any $\alpha\in(0,1)$, as long as $n>n_{\:w,\:(1:2J)}$, $w\in N'$, and $J>J_{\:\alpha}$. Therefore, since $P(N)=0$, we conclude that $
\lim\limits_{J\rightarrow\infty}\lim\limits_{|X|\rightarrow\infty}\hat{\mathcal{A}}_{\:J,\:x}(X)\overset{a.s.}{=}(0,1)$.

(C)\ Given $\alpha\in(0,1)$ and $x\in(0,1)$, let us consider arbitrary $\epsilon>0$. We know from part B of Lemma \ref{Lowess:lemma:population_asymptotic} that there exists $J_{\:1}\in\mathbb{N}$ such that $h_{\:J,\:\alpha,\:x}^{\:(P)}$ exists for $J>J_{\:1}$, which implies that $\mathbb{E}_{\:P}\big(\mathsf{1}_{\:J,\:}(\tilde{x},h_{\:J,\:\alpha,\:x}^{\:(P)})\big)=\alpha$. 

We also know from part B that there exists some $J_{\:2}\in\mathbb{N}$ such that, for $J>J_{\:2}$, we have a null set $N_{\:1,\:J}\subset \Omega$ such that, for $w\in N'_{\:1,\:J}$, there exists a $n_{\:w,\:J,\:1}\in\mathbb{N}$ such that $\hat{h}_{\:J,\:\alpha,\:x}\big((X(w,n)\big)$ exists and we have $\hat{F}_{\:J,\:x}\Big(X(w,n),\hat{h}_{\:J,\:\alpha,\:x}\big(X(w,n)\big)\Big)=\alpha$. 

Next, part A implies that for any $J\in\mathbb{N}$, there exists a null set $N_{\:2,\:J}$, such that if $w\in N'_{\:2,\:J}$, then there exists $n_{\:w,\:J,\:2}$ such that for $n>n_{\:w,\:J,\:2}$, we will have that $\big|\:\hat{F}_{\:J,\:x}\big(X(w,n),h\big)-\mathbb{E}_{\:P}\big(\mathsf{1}_{\:J,\:x}(\tilde{x},h)\big)\:\big|<\nicefrac{\epsilon}{2}$ for any $h\in(0,1)$. 

Set $J_{\:0}=\max(J_{\:1},J_{\:2})$ and $N=\bigcup_{\:J>J_{\:0}}\big(N_{\:1,\:J}\cup N_{\:2,\:J}\big)$. Clearly $N$ is a null set. Now, for $w\in N'$ and $J>J_{\:0}$, define $n_{\:w,\:J}=\operatorname{max}(n_{\:w,\:J,\:1},n_{\:w,\:J,\:2})$. Then, for $n>n_{\:w,\:J}$, $w\in N'$, and $J>J_{\:0}$, we have for $h\in(0,1)$:
\begin{align*}
\big|\:\hat{F}_{\:J,\:x}\big(X(w,n),h\big)-\mathbb{E}_{\:P}\big(\mathsf{1}_{\:J,\:x}(\tilde{x},h)\big)\:\big|<\frac{\epsilon}{2}.
\end{align*}
So, we can take $h=\hat{h}_{\:J,\:\alpha,\:x}\big(X(w,n)\big)$, and we will have that
\begin{align*}
\bigg|\:\hat{F}_{\:J,\:x}\Big(X(w,n),\hat{h}_{\:J,\:\alpha,\:x}\big(X(w,n)\big)\Big)-\mathbb{E}_{\:P}\bigg(\mathsf{1}_{\:J,\:x}\Big(\tilde{x},\hat{h}_{\:J,\:\alpha,\:x}\big(X(w,n)\big)\Big)\bigg)\:\bigg|<\frac{\epsilon}{2}.
\end{align*}
Since $\hat{F}_{\:J,\:x}\Big(X(w,n),\hat{h}_{\:J,\:\alpha,\:x}\big(X(w,n)\big)\Big)=\alpha=\mathbb{E}_{\:P}\Big(\mathsf{1}_{\:J,\:x}\big(\tilde{x},h_{\:J,\:\alpha,\:x}^{\:(P)}\big)\Big)$, we then have:
\begin{align*}
\bigg|\:\mathbb{E}_{\:P}\Big(\mathsf{1}_{\:J,\:x}\big(\tilde{x},h_{\:J,\:\alpha,\:x}^{\:(P)}\big)\Big)-\mathbb{E}_{\:P}\bigg(\mathsf{1}_{\:J,\:x}\Big(\tilde{x},\hat{h}_{\:J,\:\alpha,\:x}\big(X(w,n)\big)\Big)\bigg)\:\bigg|<\frac{\epsilon}{2}.
\end{align*}

Recall from the discussion following equation \ref{Lowess:eq:d} that $\eta^{\:(P)}_{\:J,\:x}(h)$ is the derivative of $\mathbb{E}_{\:P}\big(\mathsf{1}_{\:J,\:x}(\tilde{x},h)\big)$ with respect to $h$. Hence, by the fundamental theorem of calculus, we have:
\begin{align*}
\mathbb{E}_{\:P}\Big(\mathsf{1}_{\:J,\:x}\big(\tilde{x},h_{\:J,\:\alpha,\:x}^{\:(P)}\big)\Big)-\mathbb{E}_{\:P}\bigg(\mathsf{1}_{\:J,\:x}\Big(\tilde{x},\hat{h}_{\:J,\:\alpha,\:x}\big(X(w,n)\big)\Big)\bigg) &=\int_{\hat{h}_{\:J,\:\alpha,\:x}\big(X(w,n)\big)}^{h_{\:J,\:\alpha,\:x}^{\:(P)}}\eta^{\:(P)}_{\:J,\:x}(h)~dh.
\end{align*}
It thus follows that for $n>n_{\:w,\:J}$, $w\in N'$, and $J>J_{\:0}$, $\left|\:\int_{\hat{h}_{\:J,\:\alpha,\:x}\big(X(w,n)\big)}^{h_{\:J,\:\alpha,\:x}^{\:(P)}}\eta^{\:(P)}_{\:J,\:x}(h)~dh\:\right|<\frac{\epsilon}{2}.$

Recall that by Assumption \ref{Lowess:assumption:positive-density}, $F_{\:x}^{\:(P)}(h)=\int_{x-h}^{x+h}f^{\:(P)}(t)\,dt$ has the derivative $f_{\:x}^{\:(P)}(h)=f^{\:(P)}(x-h)+f^{\:(P)}(x+h)$, for $h\in(0,1)$. Therefore, for $n>n_{\:w,\:J}$, $w\in N'$, and $J>J_{\:0}$, by the fundamental theorem of calculus, we also have:
\begin{align*}
&\hspace{-1em}\Big|\:F_{\:x}^{\:(P)}\big(h_{\:J,\:\alpha,\:x}^{\:(P)}\big)-F_{\:x}^{\:(P)}\Big(\hat{h}_{\:J,\:\alpha,\:x}\big(X(w,n)\big)\Big)\:\Big|\\
&=\bigg|\:\int_{\hat{h}_{\:J,\:\alpha,\:x}\big(X(w,n)\big)}^{h_{\:J,\:\alpha,\:x}^{\:(P)}}f_{\:x}^{\:(P)}(h)~dh\:\bigg| \\
&\leq\bigg|\:\int_{\hat{h}_{\:J,\:\alpha,\:x}\big(X(w,n)\big)}^{h_{\:J,\:\alpha,\:x}^{\:(P)}}\big(f_{\:x}^{\:(P)}(h)-\eta^{\:(P)}_{\:J,\:x}(h)\big)~dh\:\bigg|+\,\bigg|\:\int_{\hat{h}_{\:J,\:\alpha,\:x}\big(X(w,n)\big)}^{h_{\:J,\:\alpha,\:x}^{\:(P)}}\eta^{\:(P)}_{\:J,\:x}(h)~dh\:\bigg| \\
&<\int_{\hat{h}_{\:J,\:\alpha,\:x}\big(X(w,n)\big)}^{h_{\:J,\:\alpha,\:x}^{\:(P)}}\big|\:f_{\:x}^{\:(P)}(h)-\eta^{\:(P)}_{\:J,\:x}(h)\:\big|~dh+\frac{\epsilon}{2}\\
&\leq\int_{\hat{h}_{\:J,\:\alpha,\:x}\big(X(w,n)\big)}^{h_{\:J,\:\alpha,\:x}^{\:(P)}}\xi^{\:(P)}_{\:J,\:x}\:~dh+\frac{\epsilon}{2}\\
&=\xi^{\:(P)}_{\:J,\:x}\cdot\big|\:\hat{h}_{\:J,\:\alpha,\:x}\big(X(w,n)\big)-h_{\:J,\:\alpha,\:x}^{\:(P)}\:\big|+\frac{\epsilon}{2}.
\end{align*}
Since $\big\{\hat{h}_{\:J,\:\alpha,\:x}\big(X(w,n)\big),h_{\:J,\:\alpha,\:x}^{\:(P)}\big\}\in(0,1)$, we conclude from this expression that 
\begin{align*}
\Big|\:F_{\:x}^{\:(P)}\big(h_{\:J,\:\alpha,\:x}^{\:(P)}\big)-F_{\:x}^{\:(P)}\Big(\hat{h}_{\:J,\:\alpha,\:x}\big(X(w,n)\big)\Big)\:\Big|<\xi_{\:J,\:x}^{\:(P)}+\frac{\epsilon}{2}.
\end{align*}

Now, suppose $\hat{h}_{\:J,\:\alpha,\:x}\big(X(w,\cdot)\big)$ is a limit point with respect to the set $\hat{\mathcal{H}}_{\:J,\:\alpha,\:x}(w)$. Then there exists a sub-sequence $\{n_{\:k}\}_{k\in\mathbb{N}}$ such that
\begin{align*} 
\lim\limits_{k\to\infty}\hat{h}_{\:J,\:\alpha,\:x}\big(X(w,n_{\:k})\big)=\hat{h}_{\:J,\:\alpha,\:x}\big(X(w,\cdot)\big).
\end{align*}
By continuity of probability, it then follows directly that
\begin{align*}
\lim\limits_{k\to\infty}F_{\:x}^{\:(P)}\Big(\hat{h}_{\:J,\:\alpha,\:x}\big(X(w,n_{\:k})\big)\Big)= F_{\:x}^{\:(P)}\Big(\hat{h}_{\:J,\:\alpha,\:x}\big(X(w,\cdot)\big)\Big),
\end{align*}
and moreover there exists some $k_{\:w,\:J}\in\mathbb{N}$ such that for $k>k_{\:w,\:J}$, we have
\begin{align*}
\Big|\:F_{\:x}^{\:(P)}\Big(\hat{h}_{\:J,\:\alpha,\:x}\big(X(w,\cdot)\big)\Big)-F_{\:x}^{\:(P)}\Big(\hat{h}_{\:J,\:\alpha,\:x}\big(X(w,n_{\:k})\big)\Big)\:\Big|<\frac{\epsilon}{2}.
\end{align*}

Hence if we choose $k$ such that $k>k_{\:w,\:J}$ and $n_k>\max(n_{\:w,\:J,\:1},n_{\:w,\:J,\:2})$, then 
\begin{align*}
&\hspace{-2em}\Big|\:F_{\:x}^{\:(P)}\Big(\hat{h}_{\:J,\:\alpha,\:x}\big(X(w,\cdot)\big)\Big)-F_{\:x}^{\:(P)}\big(h_{\:J,\:\alpha,\:x}^{\:(P)}\big)\:\Big| \\
&\leq\Big|\:F_{\:x}^{\:(P)}\Big(\hat{h}_{\:J,\:\alpha,\:x}\big(X(w,\cdot)\big)\Big)-F_{\:x}^{\:(P)}\Big(\hat{h}_{\:J,\:\alpha,\:x}\big(X(w,n_{\:k})\big)\Big)\:\Big| \\
&\qquad+\,\Big|\:F_{\:x}^{\:(P)}\big(h_{\:J,\:\alpha,\:x}^{\:(P)}\big)-F_{\:x}^{\:(P)}\Big(\hat{h}_{\:J,\:\alpha,\:x}\big(X(w,n_{\:k})\big)\Big)\:\Big| \\
&<\ \Big(\frac{\epsilon}{2}+\xi_{\:J,\:x}^{\:(P)}\Big)+\frac{\epsilon}{2}.
\end{align*}
Finally, since $\epsilon$ is arbitrary, we must then have the claimed result that for any $w\in N'$,
\begin{align*}
\Big|\:F^{\:(P)}\Big(\hat{h}_{\:J,\:\alpha,\:x}\big(X(w,\cdot)\big)\Big)-F^{\:(P)}\big(h_{\:J,\:\alpha,\:x}^{\:(P)}\big)\:\Big|\leq\xi_{\:J,\:x}^{\:(P)}.
\end{align*}
\end{proof}

\subsection{Main proof of Theorem \ref{Lowess:thm:prop}}
\begin{proof}
Suppose Assumption \ref{Lowess:assumption:positive-density} and Assumption \ref{Lowess:assumption:uniform-convergence} hold. First of all, since Assumption \ref{Lowess:assumption:positive-density} is true, from part C of Lemma \ref{Lowess:lemma:population_asymptotic}, we know that $\lim\limits_{J\to\infty}h_{\:J,\:\alpha,\:x}^{\:(P)}=h^{\:(P)}_{\:\alpha,\:x}$, and hence that $\lim\limits_{J\to\infty}F_x\left(h_{\:J,\:\alpha,\:x}^{\:(P)}\right)=F_x\left(h^{\:(P)}_{\:\alpha,\:x}\right)$. Secondly, since Assumption \ref{Lowess:assumption:uniform-convergence} is true, for fixed $x\in(0,1)$, $\eta_{\:J,\:x}^{\:(P)}(h)$ converges uniformly to $f^{\:(P)}(x-h)+f^{\:(P)}(x+h)$ for $h\in(0,1)$, and so, for the set of real numbers $\{\xi_{\:J,\:x}^{\:(P)}\}_{J\in\mathbb{N}}$, we have $\lim\limits_{J\to\infty}\xi_{\:J,\:x}^{\:(P)}\to 0$ for any $x\in(0,1)$. Now fix $\epsilon>0$, and observe that we can chose a $J_{\:1}$ large enough such that for $J>J_{\:1}$, both $\left|\:F_{\:x}\big(h_{\:J,\:\alpha,\:x}^{\:(P)}\big)-F_{\:x}\big(h^{\:(P)}_{\:\alpha,\:x}\big)\:\right|<\nicefrac{\epsilon}{2}$ and $\xi_{\:J,\:x}^{\:(P)}<\nicefrac{\epsilon}{2}$. 

Observe that $\limsup\limits_{n\to\infty}\hat{h}_{\:J,\:\alpha,\:x}\big(X(w,n)\big)$ is a limit point of the tail set $\hat{\mathcal{H}}_{\:J,\:\alpha,\:x}(w)$. Since Assumption \ref{Lowess:assumption:positive-density} holds, from part C of Lemma \ref{Lowess:lemma:sample_asymptotic}, we know that for $J > J_{\:2}$ there exists some null set $N$, such that for $w\in N'$, we have the inequality:
\begin{align*}
\left|\:F_{\:x}^{\:(P)}\Big(\limsup\limits_{n\to\infty}\hat{h}_{\:J,\:\alpha,\:x}\big(X(w,n)\big)\Big)-F_{\:x}^{\:(P)}\Big(h_{\:J,\:\alpha,\:x}^{\:(P)}\Big)\:\right|\leq\xi_{\:J,\:x}^{\:(P)}.
\end{align*}

Then we have for any $J>\max(J_{\:1},J_{\:2})$ that
\begin{align*}
&\hspace{-2em}\Big|\:F_{\:x}^{\:(P)}\Big(\limsup\limits_{n\to\infty}\hat{h}_{\:J,\:\alpha,\:x}\big(X(w,n)\big)\Big)-F_{\:x}^{\:(P)}\Big(h^{\:(P)}_{\:\alpha,\:x}\Big)\:\Big| \\
&\leq\Big|\:F_{\:x}^{\:(P)}\Big(\limsup\limits_{n\to\infty}\hat{h}_{\:J,\:\alpha,\:x}\big(X(w,n)\big)\Big)-F_{\:x}^{\:(P)}\Big(h_{\:J,\:\alpha,\:x}^{\:(P)}\Big)\:\Big| \\
&\hspace{2em}+\Big|\:F_{\:x}\Big(h_{\:J,\:\alpha,\:x}^{\:(P)}\Big)-F_{\:x}\Big(h^{\:(P)}_{\:\alpha,\:x}\Big)\:\Big| \\
&<\xi_{\:J,\:x}^{\:(P)}+\frac{\epsilon}{2},
\end{align*}
which is less than $\epsilon$ by the result of the preceding paragraph. Since $\epsilon$ is arbitrary, we conclude
\begin{align*}
\lim\limits_{J\to\infty}F_{\:x}^{\:(P)}\Big(\limsup\limits_{n\to\infty}\hat{h}_{\:J,\:\alpha,\:x}\big(X(w,n)\big)\Big)=F_{\:x}^{\:(P)}\Big(h^{\:(P)}_{\:\alpha,\:x}\Big).
\end{align*}

Next, under Assumption \ref{Lowess:assumption:positive-density}, part A of Lemma \ref{Lowess:lemma:inverse_continuity} implies that
\begin{align*}
\lim\limits_{J\to\infty}\limsup\limits_{n\to\infty}\hat{h}_{\:J,\:\alpha,\:x}\big(X(w,n)\big)=h^{\:(P)}_{\:\alpha,\:x}.
\end{align*}
Since this is true for any $w\in N'$ relative to the null set $N$, we have that
\begin{align*}
\lim\limits_{J\to\infty}\limsup\limits_{|X|\to\infty}\hat{h}_{\:J,\:\alpha,\:x}(X)\overset{a.s.}{=}h^{\:(P)}_{\:\alpha,\:x}.
\end{align*}

 Under Assumption \ref{Lowess:assumption:positive-density}, from part B of Lemma \ref{Lowess:lemma:exact_convergence}, we have the result that
\begin{align*}
\lim\limits_{|X|\to\infty}\hat{h}_{\:\alpha,\:x}(X)\overset{a.s.}{=}h^{\:(P)}_{\:\alpha,\:x},
\end{align*}
and so we must have that 
\begin{align*}
\lim\limits_{J\to\infty}\limsup\limits_{|X|\to\infty}\hat{h}_{\:J,\:\alpha,\:x}(X)\overset{a.s.}{=}\lim\limits_{|X|\to\infty}\hat{h}_{\:\alpha,\:x}(X).
\end{align*}

The result  $\lim\limits_{J\to\infty}\liminf\limits_{|X|\to\infty}\hat{h}_{\:J,\:\alpha,\:x}(X)\overset{a.s.}{=}\lim\limits_{|X|\to\infty}\hat{h}_{\:\alpha,\:x}(X)$ follows analogously.
\end{proof}

\end{document}